\documentclass{kais}

\pdfoutput=1 

\usepackage{wrapfig}
\usepackage[dcucite]{harvard}
\usepackage{xspace}
\usepackage{setspace}

\usepackage{fancybox}

\pubyear{2016}
\pagerange{\pageref{firstpage}--\pageref{lastpage}}
\volume{000}

\usepackage{xspace,comment,multirow}
\usepackage{graphicx}

\usepackage[cmex10]{amsmath}
\usepackage{amsthm}
\usepackage{amssymb}
\usepackage{mathtools}

\usepackage{algorithmicx}
\usepackage{algorithm}
\usepackage{algpseudocode}

\algblockdefx[parallel]{parfor}{endpar}
[1][]{$\textbf{parallel for}$ #1 $\textbf{do}$}
{$\textbf{end parallel}$}

\newcommand{\pluseq}{\mathrel{+}=}

\usepackage{nicefrac}
\algnotext{EndFor}
\algnotext{EndIf}
\algnotext{EndWhile}
\algnotext{EndProcedure}

\usepackage{array}
\newcolumntype{P}[1]{>{\centering\arraybackslash}p{#1}}
\newcolumntype{M}[1]{>{\centering\arraybackslash}m{#1}}

\usepackage{eqparbox}
\usepackage[tight,footnotesize]{subfigure}
\usepackage[caption=false,font=footnotesize]{subfig}

\usepackage{stfloats}

\usepackage{url}
\newcommand*{\eg}{e.g.\@\xspace}
\newcommand*{\ie}{i.e.\@\xspace}
\newcommand{\etal}{et~al.\@\xspace}

\newcolumntype{d}{>{\displaystyle}c}
\newcolumntype{H}{>{\setbox0=\hbox\bgroup}c<{\egroup}@{}}

\newcommand\TTT{\rule{0pt}{3.2ex}}
\newcommand\BBB{\rule[-1.4ex]{0pt}{0pt}}

\usepackage{morefloats}

\usepackage{booktabs}

\usepackage{tabularx}

\newtheorem{theorem}{Theorem}
\newtheorem{thm1}{Theorem}
\newtheorem{lemma}[thm1]{Lemma}
\newtheorem{thm2}{Theorem}
\newtheorem{cor}[thm2]{Corollary}

\theoremstyle{definition}
\newtheorem{mydef}{Definition}

\usepackage{etoolbox}
\usepackage{paralist}
\usepackage{xpatch}

\makeatletter
\let\author@tabular\tabular
\makeatother
\begin{document}
\label{firstpage}

\title{Graphlet Decomposition: Framework, Algorithms, and Applications}

\author[N.K. Ahmed et al]
{Nesreen K. Ahmed$\overset{^{1}}{,}$ 
Jennifer Neville$\overset{^{2}}{,}$ 
Ryan Rossi$\overset{^{3}}{,}$\\ 
Nick G. Duffield$\overset{^{4}}{,}$ 
and Theodore L. Willke$\overset{^{1}}{}$ 
\\
\affilsize
\vspace*{-2mm}
$^1$Intel Labs, Intel Corporation\\
\affilsize
\vspace*{-2mm}
$^2$Department of Computer Science, Purdue University\\
\affilsize
\vspace*{-2mm}
$^3$Palo Alto Research Center (PARC)\\
\affilsize
$^4$Department of Electrical and Computer Engineering, Texas A\&M University
}
 
\maketitle

\begin{abstract}
From social science to biology, numerous applications often rely on graphlets for intuitive and meaningful characterization of networks at both the global macro-level as well as the local micro-level. While graphlets have witnessed a tremendous success and impact in a variety of domains, there has yet to be a fast and efficient approach for computing the frequencies of these subgraph patterns. However, existing methods are not scalable to large networks with millions of nodes and edges, which impedes the application of graphlets to new problems that require large-scale network analysis. To address these problems, we propose a fast, efficient, and parallel algorithm for counting graphlets of size $k=\{3,4\}$-nodes that take only a fraction of the time to compute when compared with the current methods used. The proposed graphlet counting algorithms leverages a number of proven combinatorial arguments for different graphlets. For each edge, we count a few graphlets, and with these counts along with the combinatorial arguments, we obtain the exact counts of others in constant time. On a large collection of $300$+ networks from a variety of domains, our graphlet counting strategies are on average $460$x faster than current methods. This brings new opportunities to investigate the use of graphlets on much larger networks and newer applications as we show in the experiments. To the best of our knowledge, this paper provides the largest graphlet computations to date as well as the largest systematic investigation on over $300$+ networks from a variety of domains.
\end{abstract}


\section{Introduction}
\label{sec:intro}
Recursive decomposition of networks is a widely used approach in network analysis to factorize the complex structure of real-world networks into small subgraph patterns of size $k$ nodes. These patterns are called \emph{graphlets}~\cite{prvzulj2004modeling}. Graphlets (also known as motifs~\cite{milo2002network}) are defined as subgraph patterns recurring in real-world networks at frequencies that are statistically significant from those in random networks. Given a network, we can count the number of embedding of each graphlet in the network, creating a profile of sufficient statistics that characterizes the network structure~\cite{shervashidze2009efficient}. While knowing the graphlet frequencies does not uniquely define the network structure, it has been shown that graphlet frequencies often carry significant information about the local network structure in a variety of domains~\cite{Holland_Lein,faust2010puzzle,frank1988triad}. This is in contrast to global topological properties (\eg, diameter, degree distribution), where networks with similar/exact global topological properties can exhibit significantly different local structures.

\subsection{Graphlets, Scalability, \& Applications}

From social science to biology, graphlets have found numerous applications and were used as the building blocks of network analysis~\cite{milo2002network}. In social science, graphlet analysis (typically known as $k$-subgraph census) is widely adopted in sociometric studies~\cite{Holland_Lein,frank1988triad}. Much of the work in this vein focused on analyzing triadic tendencies as important structural features of social networks (\eg, transitivity or triadic closure) as well as analyzing triadic configurations as the basis for various social network theories (\eg, social balance, strength of weak ties, stability of ties, or trust~\cite{granovetter1983strength}). In biology~\cite{prvzulj2004modeling,milenkoviae2008uncovering}, graphlets were widely used for protein function prediction~\cite{shervashidze2009efficient}, network alignment~\cite{milenkovic2010optimal}, and phylogeny~\cite{kuchaiev2010topological} to name a few. More recently, there has been an increased interest in exploring the role of graphlet analysis in computer networking~\cite{feldman2008automatic,hales2008motifs,becchetti2008efficient} (\eg, for web spam detection, analysis of peer-to-peer protocols and Internet AS graphs), chemoinformatics~\cite{ralaivola2005graph,kashima2010graph}, image segmentation~\cite{zhang2013probabilistic}, among others~\cite{zhang2013discovering}.

While graphlet counting and discovery have witnessed a tremendous success and impact in a variety of domains from social science to biology, there has yet to be a fast and efficient approach for computing the frequencies of these patterns. For instance, Shervashidze~\etal~\cite{shervashidze2009efficient} takes hours to count graphlets on relatively small biological networks (\ie, few hundreds/thousands of nodes/edges) and uses such counts as features for graph classification~\cite{vishwanathan2010graph}. Previous work showed that graphlet counting is computationally intensive since the number of possible $k$-subgraphs in a graph $G$ increases exponentially with $k$ in $\mathcal{O}(|V|^k)$ and can be computed in $\mathcal{O}(|V|.\Delta^{k-1})$ for any bounded degree graph, where $\Delta$ is the maximum degree of the graph~\cite{shervashidze2009efficient}.

To address these problems, we propose a fast, efficient, and parallel algorithm for counting graphlets of size $k=\{3,4\}$-nodes that take only a fraction of the time to compute when compared with the current methods used. The proposed graphlet counting algorithm leverages a number of proven combinatorial arguments for different graphlets. For each edge, we count a few graphlets, and with these counts along with the combinatorial arguments, we obtain the exact counts of others in constant time. On a large collection of $300$+ networks from a variety of domains, our graphlet counting strategies are on average $460$x faster than current methods. This brings new opportunities to investigate the use of graphlets on much larger networks and newer applications as we show in our experiments. To the best of our knowledge, this paper provides the largest graphlet computations to date as well as the largest systematic investigation on over $300$+ networks.

Furthermore, a number of important machine learning tasks are likely to benefit from such an approach, including graph anomaly detection~\cite{noble2003graph}, as well as using graphlets as features for improving community detection~\cite{schaeffer2007graph}, role discovery~\cite{rossi2014roles}, graph classification~\cite{vishwanathan2010graph}, and relational learning~\cite{getoor2007introduction}.

We test the scalability of our proposed approach experimentally on $300+$ networks from a variety of domains, such as biological, social, and technological domains. We compare our approach to the state-of-the-art exact counting methods such as RAGE~\cite{marcus2012rage}, FANMOD~\cite{wernicke2006fanmod}, and Orca~\cite{hovcevar2014combinatorial}. We found that RAGE~\cite{marcus2012rage} took 2400 seconds to count graphlets on a small $26$k node graph, whereas our proposed method is $460$x faster, taking only $0.01$ seconds. We also note that FANMOD~\cite{wernicke2006fanmod}, another recent approach, takes $172800$ seconds, and Orca~\cite{hovcevar2014combinatorial} takes $2.5$ seconds for the same small graph. Our exact graphlet analysis is well-suited for shared-memory multi-core architectures (CPU and GPU), distributed architectures (MPI), and hybrid implementations that leverage the advantages of both.

\subsection{Contributions}
\begin{list}{{$\bullet$}}{\itemsep=3px}
\item \textbf{Algorithms.} A fast, efficient, and parallel graphlet counting algorithm that leverages a number of combinatorial arguments that we show for different graphlets. The combinatorial arguments we show in this paper enable us to obtain significant improvement on the scalability of graphlet counting. 

\item \textbf{Scalability.} The proposed graphlet counting algorithm achieves on average $460$x runtime improvement over the state-of-the-art methods. In addition, we analyze graphlet counts on graphs of sizes that are beyond the scope of the state-of-the-art (\eg, on graphs with hundred million nodes and billion edges). 

\item \textbf{Effectiveness.} Largest graphlet computations to date and largest systematic evaluation on over $300$+ large-scale networks from a variety of domains.

\item \textbf{Applications.} We systematically investigate a variety of existing and new applications for graphlet counting, such as finding unique patterns in graphs, graph similarity, and graph classification. 
\end{list}

\section{Background}
\label{sec:background}
Graphlets are subgraph patterns recurring in real-world networks at frequencies that are significantly higher than those in random networks~\cite{milo2002network,prvzulj2004modeling}. Previous work showed that graphlets can be used to define universal classes of networks~\cite{milo2002network}. Moreover, graphlets are at the heart and foundation of many network analysis tasks (\eg, network classification, network alignment, etc.)~\cite{prvzulj2004modeling,milenkoviae2008uncovering,hayes2013graphlet}. In this paper, we introduce an efficient algorithm to compute the number of embedding of each graphlet of size $k=\{2,3,4\}$ nodes in the network (see Table~\ref{table:graphlet_notation} for notation).

{
\setlength{\tabcolsep}
{3.1pt}
\begin{table}
\caption{Summary of graphlet notation}
\medskip
\label{table:graphlet_notation}
\small \scriptsize
\scalebox{1.0}{
\begin{tabularx}{1.0\textwidth}{M{0.2cm} M{0.4cm}M{0.4cm} M{2cm} M{1.3cm} M{0.8cm}M{0.4cm}M{0.4cm}M{0.7cm}M{0.4cm}M{0.4cm}M{0.4cm}M{0.3cm} M{0.3cm}M{0.3cm} }
\midrule
\multicolumn{15}
{p{1.00\textwidth}}
{
Summary of the notation and properties for the graphlets of size $k = \{2,3,4\}$.
Note that $\rho$ denotes density, $\Delta$ and $\bar{d}$ denote the max and mean degree, whereas assortativity is denoted by $r$.
Also, $|T|$ denotes the total number of triangles, $K$ is the max k-core number, $\chi$ denotes the Chromatic number, whereas ${\rm D}$ denotes the diameter, ${\rm B}$ denotes the max betweenness, and $|C|$ denotes the number of components.
Note that if $|C|>1$, then $r$, ${\rm D}$, and ${\rm B}$ are from the largest component.
}
\\
\midrule
&
\multicolumn{2}{c}{\textbf{Graphlet}} 
 & Description & Complement & $\rho$ & $\Delta$ & $\bar{d}$ & $r$ & $|T|$ & $K$ & $\chi$ & ${\rm D}$ & ${\rm B}$ & $|C|$\\ 
\midrule 
\multicolumn{14}{l}{($k=4$)$-$\textsc{Graphlets}} \\
\midrule 
\TTT\BBB
\multirow{8}{*}{\rotatebox{90}{\mbox{}
}} 
\multirow{10}{*}{\rotatebox{90}{\mbox{\textsc{Connected}}}} 
&  \includegraphics[scale=0.04]{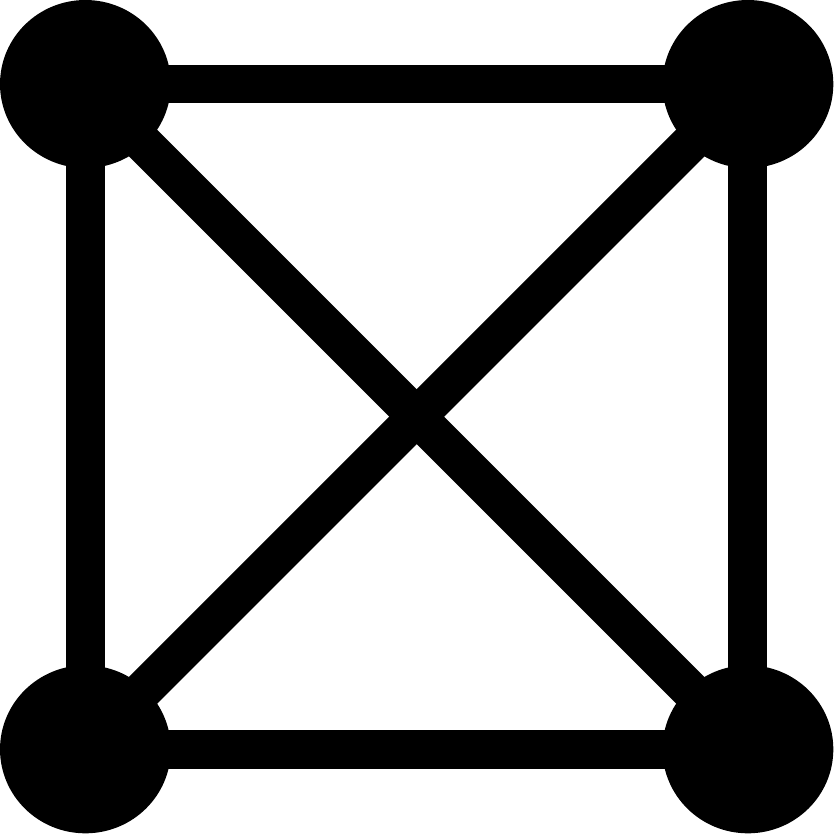} 
& $g_{4_1}$  & 4-clique 
&  \includegraphics[scale=0.04]{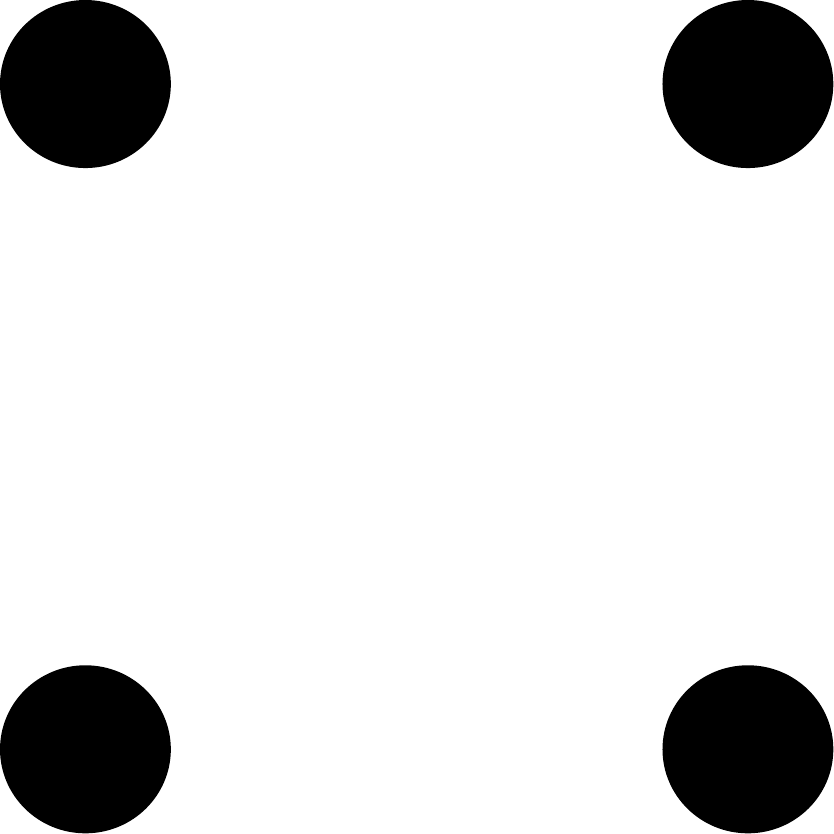} 
& 1.00 & 3 & 3.0 & 1.00 & 4 & 3 & 4 & 1 & 0 & 1\\ 
\TTT\BBB
&  \includegraphics[scale=0.04]{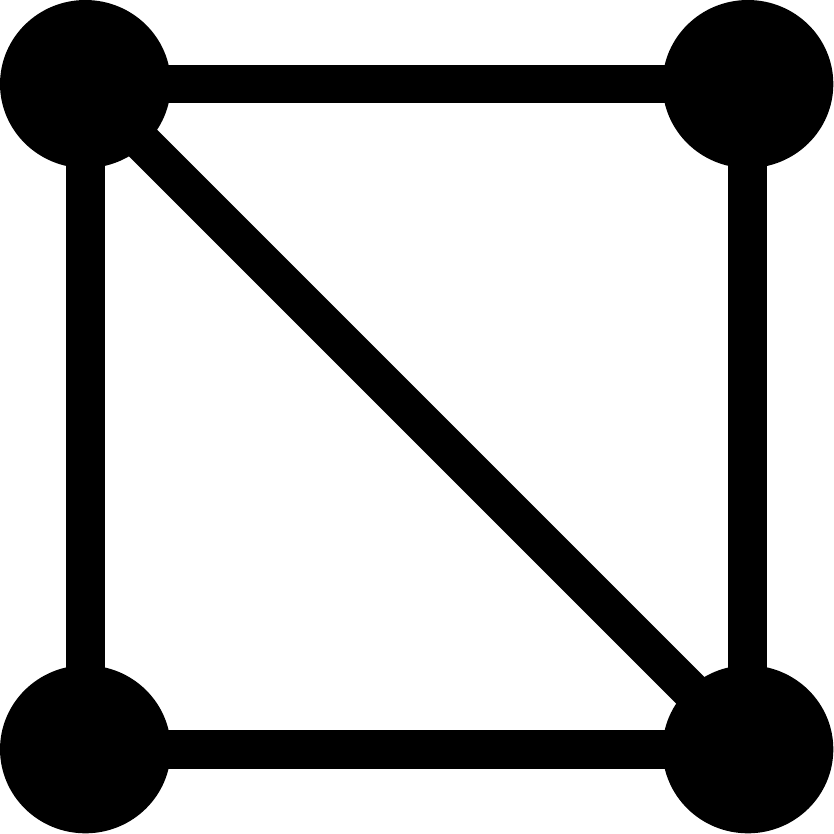} 
& $g_{4_2}$  & 4-chordalcycle 
&  \includegraphics[scale=0.04]{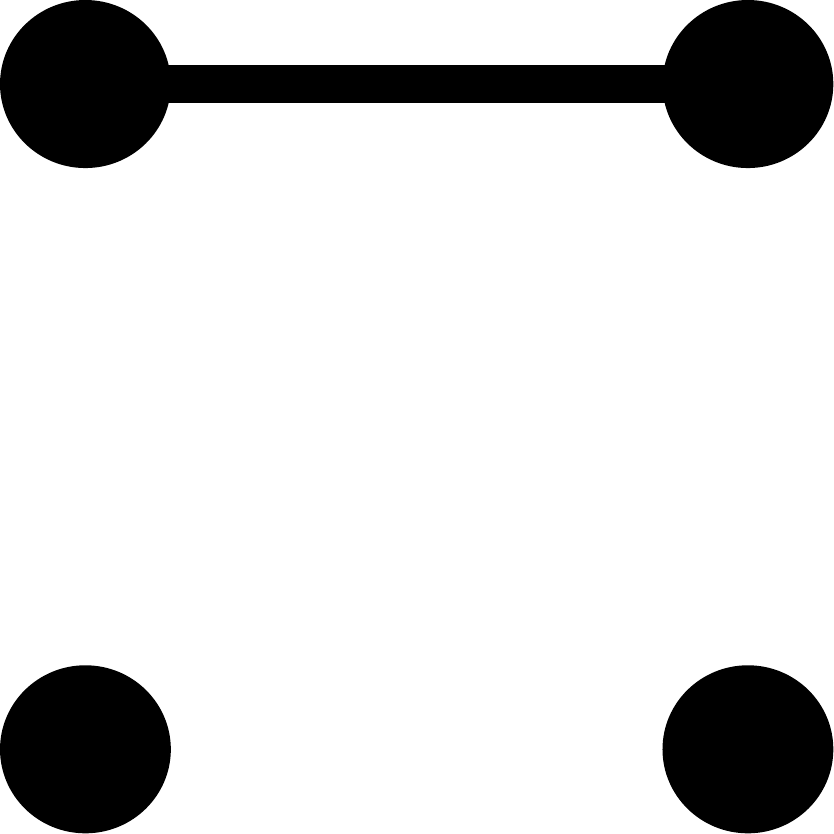} 
& 0.83 & 3 & 2.5 & -0.66 & 2 & 2 & 3 & 2 & 1 & 1 \\ 
\TTT\BBB
&  \includegraphics[scale=0.04]{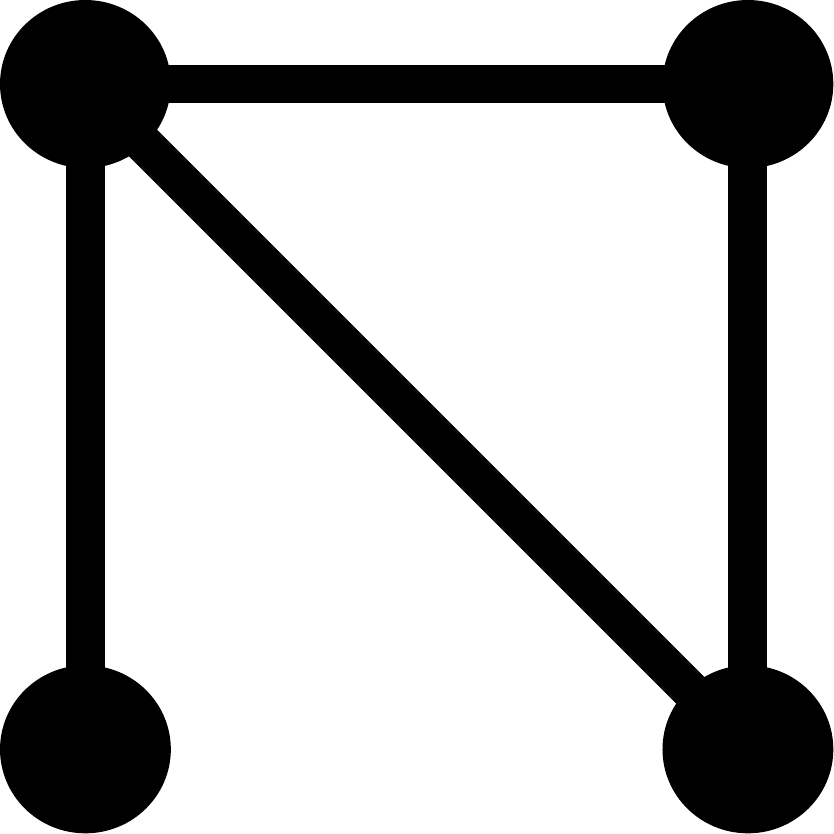} 
& $g_{4_3}$  & 4-tailedtriangle 
&  \includegraphics[scale=0.04]{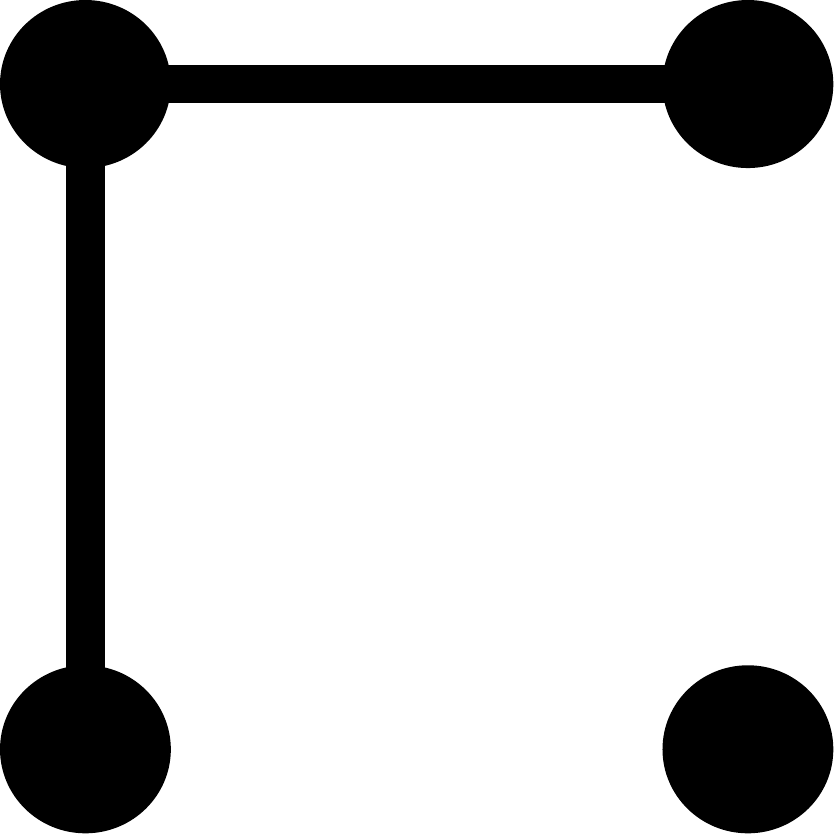} 
& 0.67 & 3 & 2.0 & -0.71 & 1 & 2 & 3 & 2 & 2 & 1\\
\TTT\BBB
&  \includegraphics[scale=0.04]{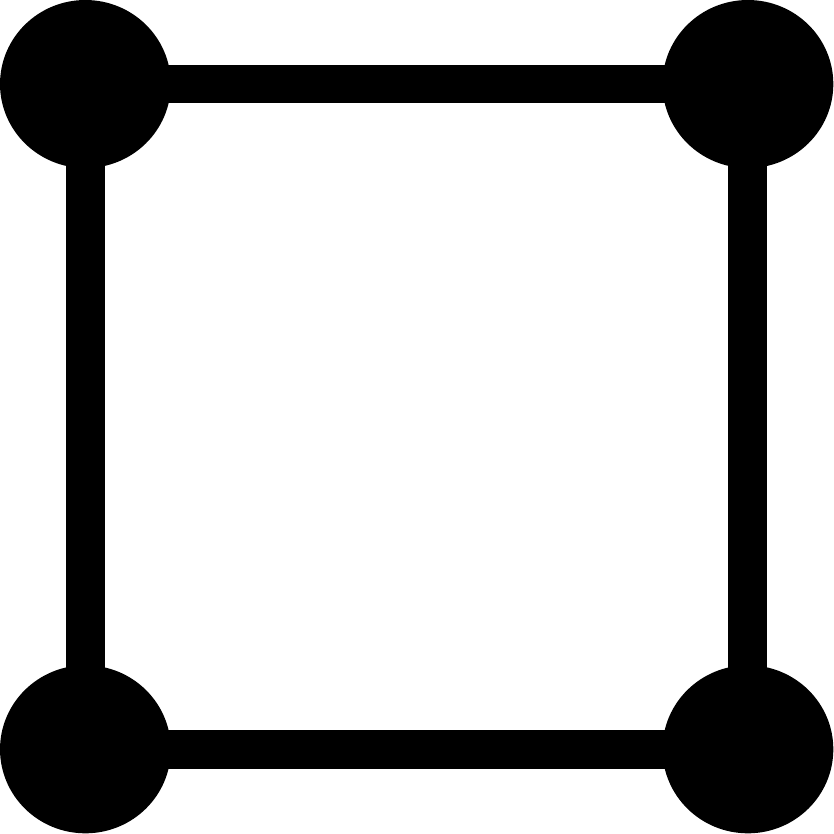} 
& $g_{4_4}$  & 4-cycle 
&  \includegraphics[scale=0.04]{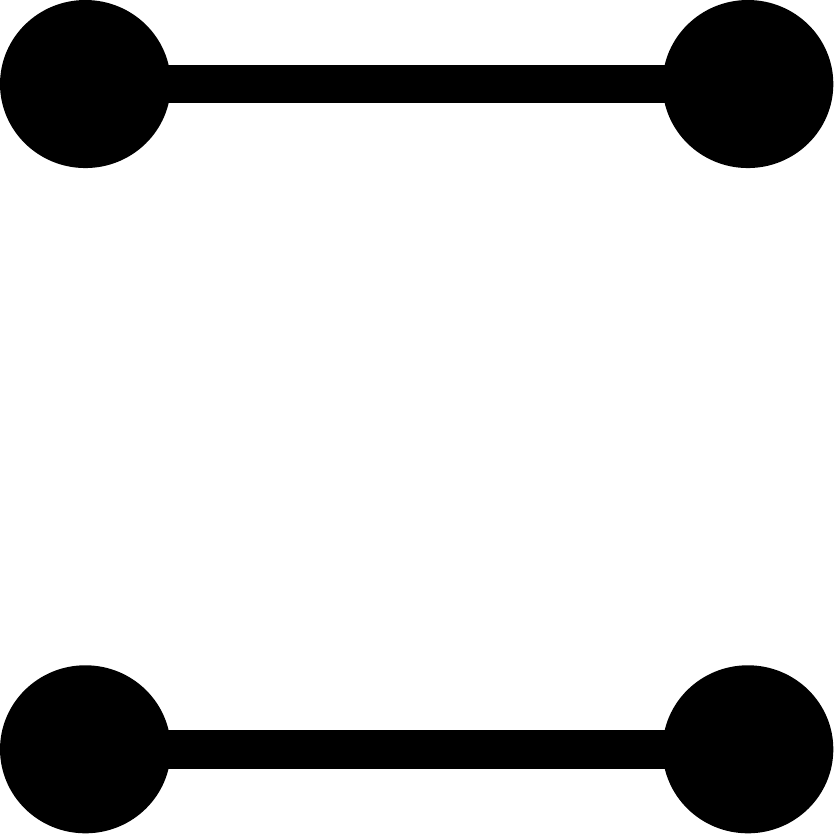} 
& 0.67 & 2 & 2.0 & 1.00 & 0 & 2 & 2 & 2 & 1 & 1\\
\TTT\BBB
&  \includegraphics[scale=0.04]{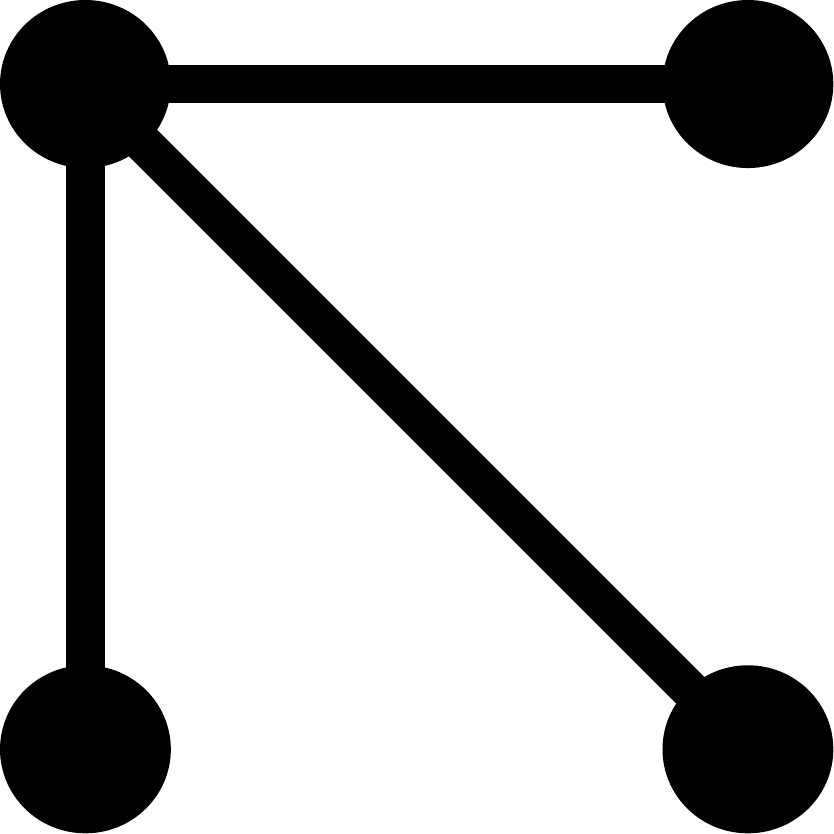} 
& $g_{4_5}$  & 3-star 
&  \includegraphics[scale=0.04]{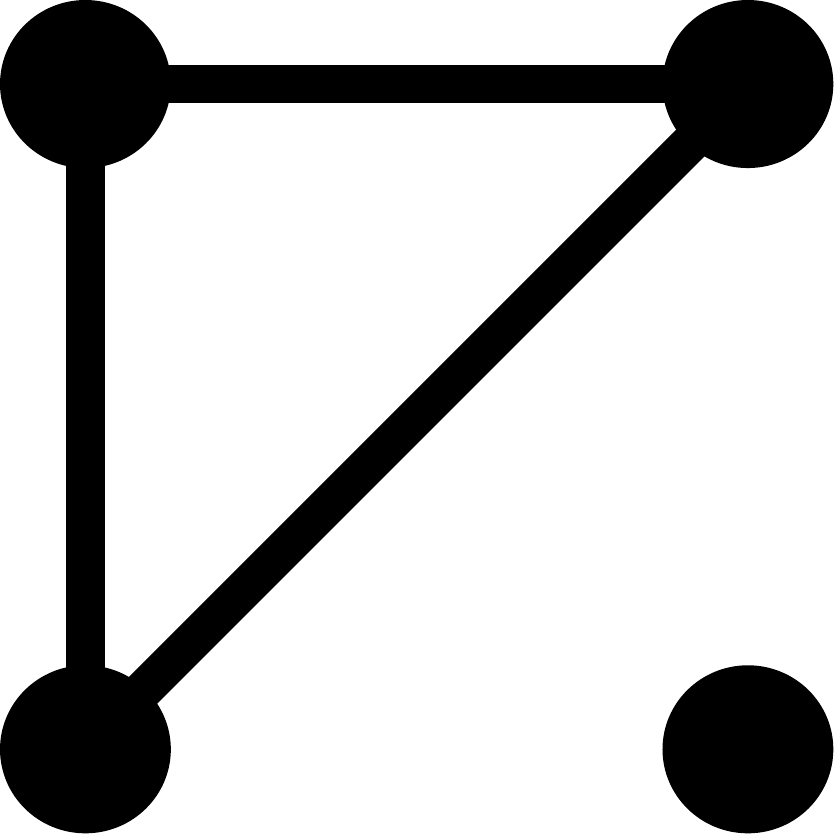} 
& 0.50 & 3 & 1.5 & -1.00 & 0 & 1 & 2 & 2 & 3 & 1\\
\TTT\BBB
&  \includegraphics[scale=0.04]{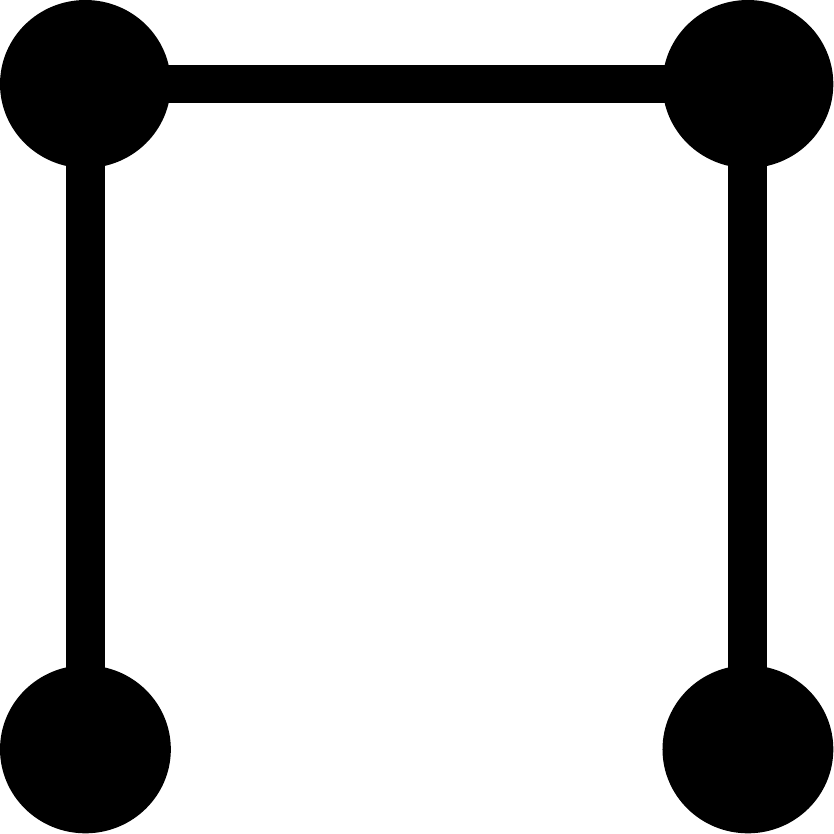} 
& $g_{4_6}$  & 4-path 
&  \includegraphics[scale=0.04]{graphics/graphlets/4-path.pdf} 
& 0.50 & 2 & 1.5 & -0.50 & 0 & 1 & 2 & 3 & 2 & 1\\
\midrule 
\TTT\BBB
\multirow{8}{*}{\rotatebox{90}{\mbox{\textsc{Disconnected}}}} 
&  \includegraphics[scale=0.04]{graphics/graphlets/4-node-triangle.pdf} 
& $g_{4_7}$  & 4-node-1-triangle 
&  \includegraphics[scale=0.04]{graphics/graphlets/4-star.pdf} 
& 0.50 & 2 & 1.5 & 1.00 & 1 & 2 & 3 & 1 & 0 & 2\\
\TTT\BBB
&  \includegraphics[scale=0.04]{graphics/graphlets/4-node-star.pdf} 
& $g_{4_8}$  & 4-node-2-star 
&  \includegraphics[scale=0.04]{graphics/graphlets/tailed-triangle.pdf} 
& 0.33 & 2 & 1.0 & -1.00 & 0 & 1 & 2 & 2 & 1 & 2\\
\TTT\BBB
&  \includegraphics[scale=0.04]{graphics/graphlets/4-node-2edges.pdf} 
& $g_{4_9}$  & 4-node-2-edge 
&  \includegraphics[scale=0.04]{graphics/graphlets/4-cycle.pdf} 
& 0.33 & 1 & 1.0 & 1.00 & 0 & 1 & 2 & 1 & 0 & 2\\
\TTT\BBB
&  \includegraphics[scale=0.04]{graphics/graphlets/4-node-1edge.pdf} 
& $g_{4_{10}}$  & 4-node-1-edge 
&  \includegraphics[scale=0.04]{graphics/graphlets/chordal-cycle.pdf}
& 0.17 & 1 & 0.5 & 1.00 & 0 & 1 & 2 & 1 & 0 & 3\\
\TTT\BBB
&  \includegraphics[scale=0.04]{graphics/graphlets/4-node-indep.pdf} 
& $g_{4_{11}}$  & 4-node-independent 
&  \includegraphics[scale=0.04]{graphics/graphlets/4-clique.pdf} 
& 0.00 & 0 & 0.0 & 0.00 & 0 & 0 & 1 & $\infty$ & 0 & 4\\
\midrule 
\multicolumn{14}{l}{($k=3$)$-$\textsc{Graphlets}} \\
\midrule 
\TTT\BBB
\multirow{8}{*}{\rotatebox{90}{\mbox{}
}} 
&  \includegraphics[scale=0.04]{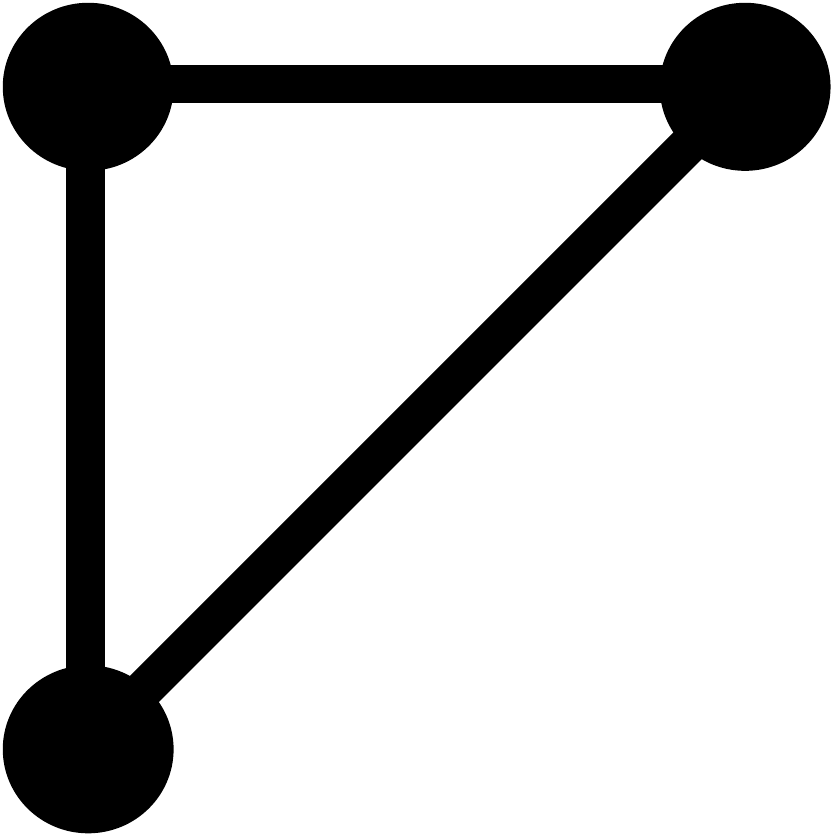} 
& $g_{3_1}$  & triangle 
&  \includegraphics[scale=0.04]{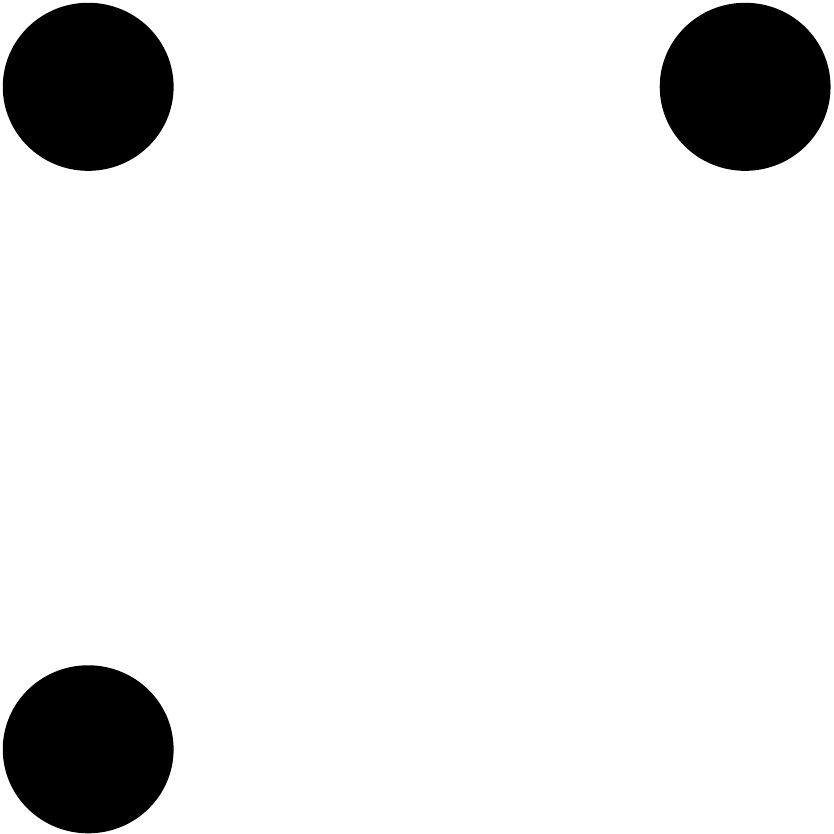} 
& 1.00 & 2 & 2.0 & 1.00 & 1 & 2 & 3 & 1 & 0 & 1\\ 
\TTT\BBB
&  \includegraphics[scale=0.04]{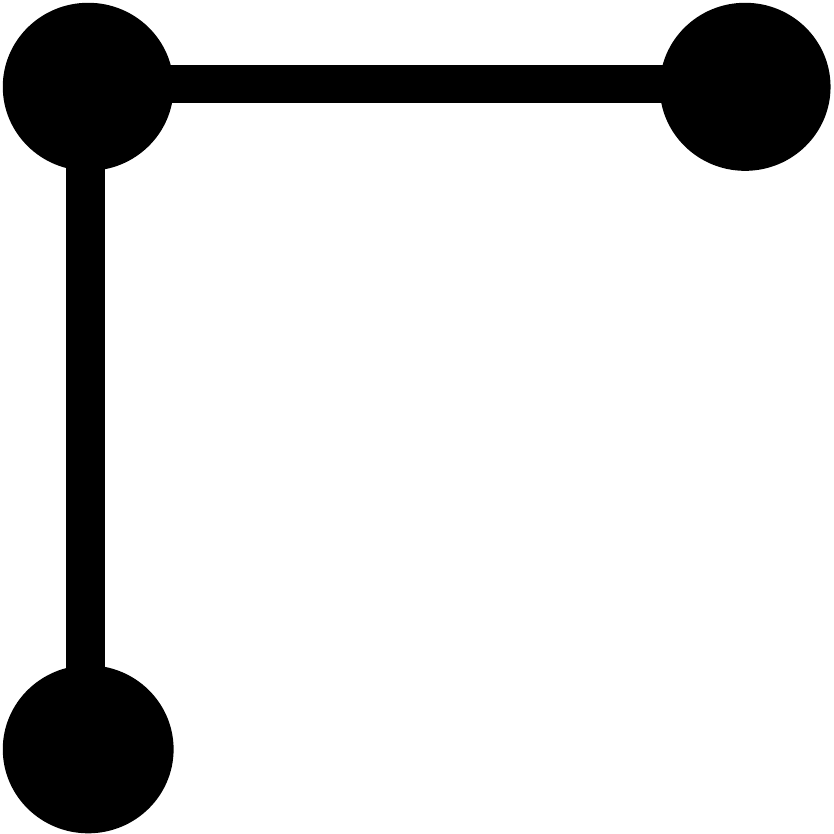} 
& $g_{3_2}$  & 2-star 
&  \includegraphics[scale=0.04]{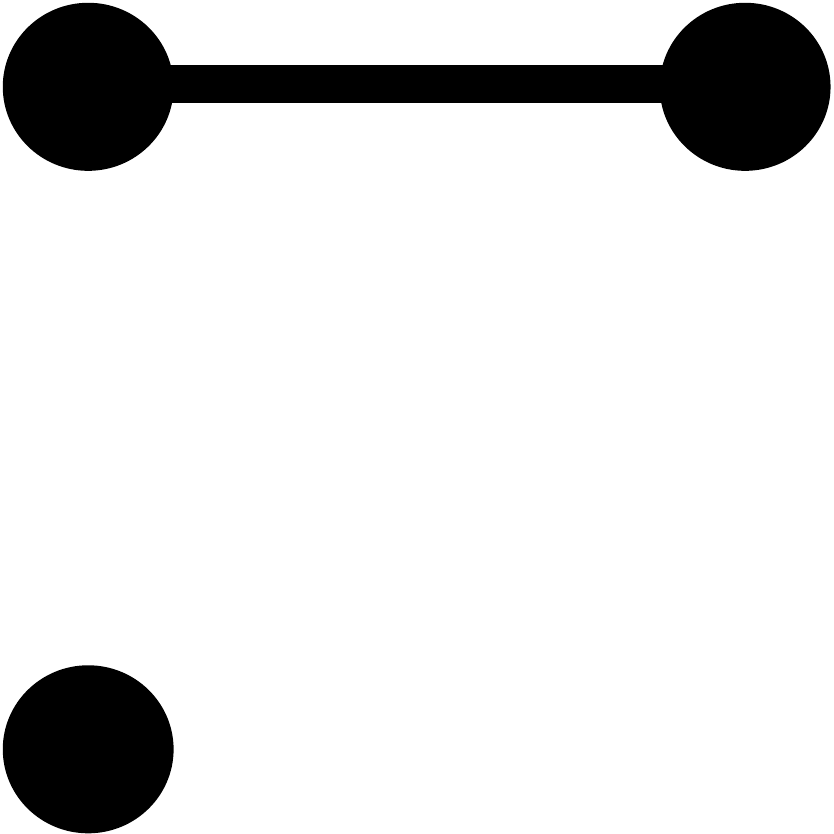} 
& 0.67 & 2 & 1.33 & -1.00  & 0 & 1 & 2 & 2 & 1 & 1\\ 
\TTT\BBB
&  \includegraphics[scale=0.04]{graphics/graphlets/3-disconnected-1edge.pdf} 
& $g_{3_3}$  & 3-node-1-edge 
&  \includegraphics[scale=0.04]{graphics/graphlets/3-path.pdf} 
& 0.33 & 1 & 0.67 & 1.00  & 0 & 1 & 2 & 1 & 0 & 2\\ 
\TTT\BBB
&  \includegraphics[scale=0.04]{graphics/graphlets/3-disconnected-indep.pdf} 
& $g_{3_4}$  & 3-node-independent 
&  \includegraphics[scale=0.04]{graphics/graphlets/3-triangle.pdf} 
& 0.00 & 0 & 0.00 & 0.00  & 0 & 0 & 1 & $\infty$ & 0 & 3\\ 
\midrule 
\multicolumn{14}{l}{($k=2$)$-$\textsc{Graphlets}} \\
\midrule 
\TTT\BBB
\multirow{8}{*}{\rotatebox{90}{\mbox{}
}} 
&  \includegraphics[scale=0.04]{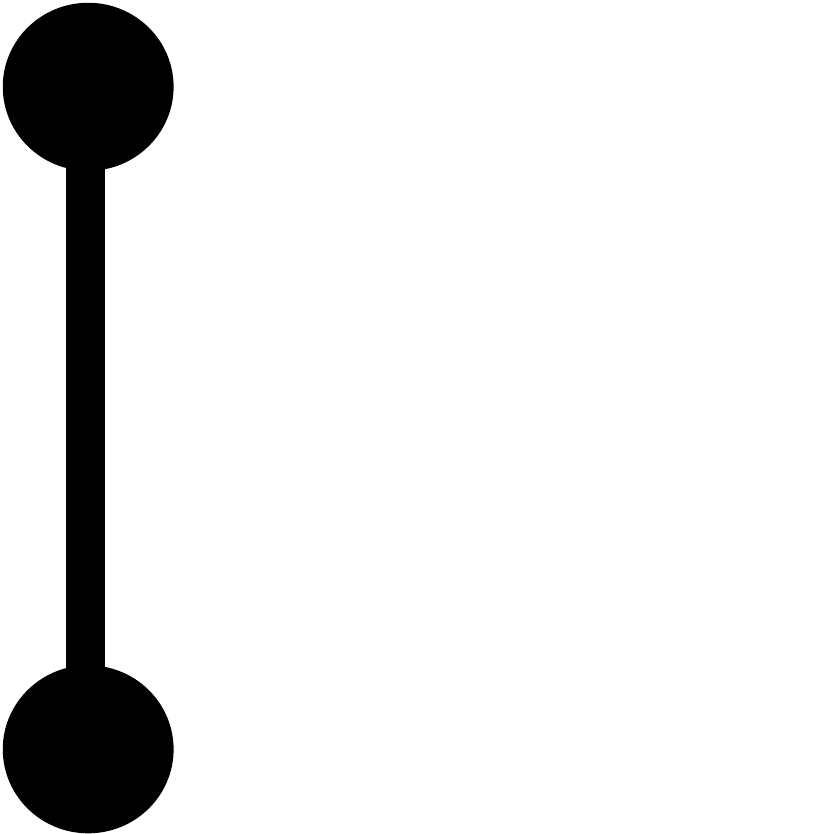} 
& $g_{2_1}$  & edge 
&  \includegraphics[scale=0.04]{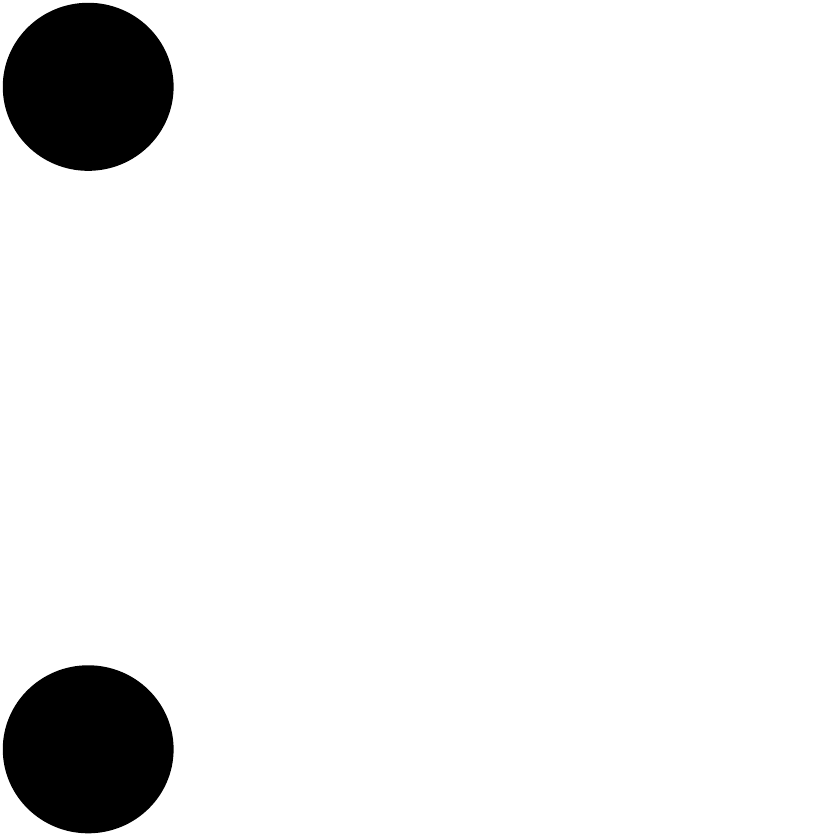} 
& 1.00 & 1 & 1.0 & 1.00  & 0 & 1 & 2 & 1 & 0 & 1\\ 
\TTT\BBB
&  \includegraphics[scale=0.04]{graphics/graphlets/2-disconnected.pdf} 
& $g_{2_2}$  & 2-node-independent 
&  \includegraphics[scale=0.04]{graphics/graphlets/2-edge.pdf} 
& 0.00 & 0 & 0.0 & 0.00  & 0 & 0 & 1 & $\infty$ & 0 & 2\\ 
\midrule
\end{tabularx}
}
\vspace{-3.mm}
\end{table}
}

\subsection{Notation and Definitions}
Given an undirected simple input graph $G=(V,E)$, a graphlet of size $k$ nodes is defined as any subgraph $G_k \subset G$ which consists of a subset of $k$ nodes of the graph $G$. 
In this paper, we mainly focus on computing the frequencies of induced graphlets. An \emph{induced} graphlet is an induced subgraph that consists of \emph{all} edges between its nodes that are present in the input graph (as described in Definition~\ref{def:ind_graphlet}). 
In addition, we distinguish between \emph{connected} and \emph{disconnected} graphlets (see Table~\ref{table:graphlet_notation}). A graphlet is connected if there is a path from any node to any other node in the graphlet (see Definition~\ref{def:conn_graphlet}). Table~\ref{table:graphlet_notation} provides a summary of the notation and properties of all possible induced graphlets of size $k = \{2,3,4\}$. 

\begin{mydef}{Induced Graphlet:}\label{def:ind_graphlet} 
an induced graphlet $G_k = (V_k,E_k)$ is a subgraph that consists of a subset of $k$ vertices of the graph $G = (V,E)$ (i.e., $V_k \subset V$) together with all the edges whose endpoints are both in this subset (\ie, $E_k = \{ \forall e \in E \,|\, e = (u,v) \wedge u,v \in V_k \}$).  
\end{mydef}

\begin{mydef}{Connected Graphlet:}\label{def:conn_graphlet}
a graphlet $G_k=(V_k,E_k)$ is connected when there is a path from any node to any other node in the graphlet (\ie, $ \forall u, v \in V_k,$\,$\exists P_{u-v}:u,...,w,...,v$, such that $d(u,v) \geq 0 \wedge d(u,v) \neq \infty$). By definition, there exist one and only one connected component in a graphlet $G_k$ (\ie, $|C|=1$) if and only if $G_k$ is connected.    
\end{mydef}

\vspace{1.mm}
\noindent\shadowbox{\begin{minipage}[t]{0.97\linewidth}
\noindent
{\sf \bfseries \sc Problem Definition}. 
Given a family of graphlets of size $k$ nodes $\mathcal{G}_k = \{g_{k_1}, g_{k_2}, ..., g_{k_m}\}$, our goal is to count the number of embeddings (appearances) of each graphlet $g_{k_i} \in \mathcal{G}_k$ in the input graph $G$. In other words, we need to count the number of induced graphlets $G_k$ in $G$ that are isomorphic to each graphlet $g_{k_i} \in \mathcal{G}_k$ in the family, such a number is denoted by $G \choose g_{k_i}$ ~\cite{HGTbook}.
\end{minipage}}
\vspace{1.mm}
 
A graphlet $g_{k_i} \in \mathcal{G}_k$ is embedded in the graph $G$, if and only if there is an injective mapping $\sigma: V_{g_{k_i}} \rightarrow V$, with $e=(u,v) \in E_{g_{k_i}}$ if and only if $e'=(\sigma(u),\sigma(v)) \in E$. Table~\ref{table:graphlet_notation} shows that $|\mathcal{G}_k|=\{2,4,11\}$ when $k=\{2,3,4\}$ respectively. Further, given a family $\mathcal{G}_k = \{g_{k_1}, g_{k_2}, ..., g_{k_m}\}$ of graphlets of size $k$ nodes, we define $f(g_{k_i},G)$ as the relative frequency of any graphlet $g_{k_i} \in \mathcal{G}_k$ in the input graph $G$.

\subsection{Relationship to Graph Complement} 
The complement of a graph $G$, denoted by $\bar{G}$, is the graph defined on the same vertices as $G$ such that two vertices are connected in $\bar{G}$ if and only if they are not connected in $G$. Therefore, the graph sum $G+\bar{G}$ gives the complete graph on the set of vertices of $G$. There are direct relationships between the frequencies of graphlets and the frequencies of their complement. For each graphlet $g_{k_i}$, there exists a non-isomorphic complementary graphlet pattern $\bar{g}_{k_i}$, such that two vertices are connected in $\bar{g}_{k_i}$ if and only if they are not connected in $g_{k_i}$~\cite{HGTbook}. For example, cliques and independent sets of size $k$ nodes are pairs of complementary graphlets. Similarly, chordal cycles of size $4$ nodes are complementary to the $4$-node-$1$edge graphlet (see Table~\ref{table:graphlet_notation}). It is also worth noting that the $4$-path graphlet is a self-complementary pattern, which means the $4$-path is isomorphic to itself. From this discussion, it is clear that the number of embeddings of each graphlet $g_{k_i} \in \mathcal{G}_k$ in the input graph $G$ is equivalent to the number of embeddings of its complementary graphlet $\bar{g}_{k_i}$ in the complement graph $\bar{G}$. In other words, $f(g_{k_i},G) = f(\bar{g}_{k_i},\bar{G})$~\cite{HGTbook}.

\subsection{Relationship to Graph/Matrix Reconstruction Theorems} 
The graph reconstruction conjecture~\cite{HGTbook}, states that an undirected graph $G$ can be uniquely determined up to an isomorphism, from the set of all possible vertex-deleted subgraphs of $G$ (\ie, $\{G_v\}_{v \in V}$)~\cite{mckay1997small}. Verification of this conjecture for all possible graphs up to $6$ vertices was carried by Kelly~\cite{kelly1957}, and later was extended to up to $11$ vertices by McKay~\cite{mckay1997small}. Clearly, if two graphs are isomorphic (\ie, $G \cong G'$), then their graphlet frequencies would be the same (\ie, $f_k(G) = f_k(G')$), but the reverse remains a conjecture for the general case of graphs. In contrast, the matrix reconstruction theorem has been resolved~\cite{manvel1971reconstruction}, which states that any $N \times N$ matrix can be reconstructed from its list of all possible principal minors obtained by the deletion of the $k$-th row and the $k$-th column~\cite{manvel1971reconstruction}, which is the foundation of a class of graph kernels called the \emph{graphlet kernel}~\cite{shervashidze2009efficient}.

\subsection{Related Work}
In this section, we briefly discuss some of the related work, highlighting various graph mining and machine learning tasks that would benefit from our approach. Much of the previous work focused on counting certain types of graphlets (\eg, only connected graphlets such as cliques and cycles)~\cite{kloks2000finding,wernicke2006fanmod,hovcevar2014combinatorial}. However, a number of graph mining and machine learning tasks rely on counting all graphlets of a certain size. 

For example, some previous work used the full spectrum of graphlet frequencies to define a domain-independent coordinate system in which collections of graphs can be compactly represented and analyzed within a common space~\cite{ugander2013subgraph}. Moreover, a variety of graph kernels have been proposed in machine learning (\eg, graphlet, subtree, and random walk kernels)~\cite{vishwanathan2010graph,costa2010fast,shervashidze2009efficient} to bridge the gap between graph learning and kernel methods. And some types of the graph kernels, in particular the graphlet kernel, rely on counting all graphlets. However, a general limitation of most graph kernels (including the graphlet kernel) is that they scale poorly to large graphs with more than few hundreds/thousands of nodes~\cite{vishwanathan2010graph}. Thus, our fast algorithms would speedup the computations of these methods and their related applications in graph modeling, similarity, and comparisons. 

Recently, there is an increased interest in sampling and other heuristic approaches for obtaining approximate counts of various graphlets~\cite{bhuiyan2012guise,gonen2009approximating}.
However, our approach focuses on exact graphlet counting and thus sampling methods are outside the scope of this paper. Nevertheless, the analysis and combinatorial arguments we show in this paper can be used along with efficient sampling methods to provide more accurate and efficient approximations. 

In addition, the aim and scope of this paper is different from the aforementioned problem of graph reconstruction. While graph reconstruction tries to test for the notion of isomorphism and structure equivalence between graphs, our goal is to relax the notion of equivalence to some form of \emph{structural similarity} between graphs, such that the graph similarity is measured using the feature representation of graphlets.

\section{Framework}
\label{sec:framework}
In this section, we describe our approach for graphlet counting that takes only a fraction of the time to compute when compared with the current methods used. We introduce a number of combinatorial arguments that we show for different graphlets. The proposed graphlet counting algorithm leverages these combinatorial arguments to obtain significant improvement on the scalability of graphlet counting. For each edge, we count only a few graphlets, and with these counts along with the combinatorial arguments, we derive the exact counts of the others in constant time.

\subsection{Searching Edge Neighborhoods}  

Our proposed algorithm iterates over all the edges of the input graph $G = (V,E)$. For each edge $e=(u,v) \in E$, we define the \emph{neighborhood} of an edge $e$, denoted by $\mathcal{N}(e)$, as the set of all nodes that are connected to the endpoints of $e$ --- \ie, $\mathcal{N}(e) = \{\mathcal{N}(u) \setminus \{v\}\} \cup \{\mathcal{N}(v) \setminus \{u\}\}$, where $\mathcal{N}(u) \textrm{ and } \mathcal{N}(v)$ are the set of neighbors of $u \textrm{ and } v$ respectively. Given a single edge $e=(u,v) \in E$, we explore the subgraph surrounding this edge --- \ie, the subgraph induced by both its endpoints and the nodes in its neighborhood. We call this subgraph the \emph{egonet} of the edge $e$, where $e$ is the center (ego) of the subgraph.

We search for possible graphlet patterns of size $k=\{3,4\}$ in the egonets of all edges in the graph. By searching egonets of edges, we first map the problem to the local (lower-dimensional) space induced by the neighborhood of each edge, and then merge the search results for all edges. Searching over a local low-dimensional space of edge neighborhoods is clearly more efficient than searching over the global high-dimensional space of the whole graph. Moreover, searching over a local low-dimensional space of edge neighborhoods is amenable to parallel implementation, which offers additional speedup over iterative methods. Note that exhaustive search of the egonet of any edge $e \in E$ yields at least $\mathcal{O}(\Delta^{k-1})$ asymptotically, where $\Delta$ is the maximum degree in $G$. Clearly, exhaustive search is computationally intensive for large graphs, and our approach is more efficient as we will show next. 

\subsection{Counting Graphlets of Size $(k=3)$ Nodes}
Algorithm~\ref{alg:parallel-graphlet} (\textsc{TriadCensus}) shows how to count graphlets of size $k=3$ for each edge.
There are four possible graphlets of size $k=3$ nodes, where only $g_{3_1}$ (\ie, triangle patterns) and $g_{3_2}$ (\ie, $2$-star patterns) are connected graphlets (see Table~\ref{table:graphlet_notation}).

\algrenewcommand{\alglinenumber}[1]{\scriptsize#1:}
\begin{figure}
\begin{center}
\begin{minipage}{1.0\linewidth}
\begin{algorithm}[H]
\caption{\,\small{Our exact triad census algorithm for counting all $3$-node graphlets. The algorithm takes an undirected graph as input and returns the frequencies of all $3$-node graphlets $f(\mathcal{G}_3,G)$.}}
\label{alg:parallel-graphlet}

\begin{spacing}{1.2}
\fontsize{8}{9}\selectfont
\begin{algorithmic}[1]
\Procedure {TriadCensus}{$G=\left (V,E\right )$}
\State Initialize Array $X$
\parfor[$e=(u,v) \in E$]
\State $\mathrm{Star}_u = \emptyset,  \mathrm{Star}_v = \emptyset, \mathrm{Tri}_e = \emptyset$ 
\For {$w \in \mathcal{N}(u)$ \label{alg:tri_start}}
\If{$w = v$} \textbf{continue}
\EndIf
\State Add $w$ to $\mathrm{Star}_u$ and set $X(w) = 1$
\EndFor
\For {$w \in \mathcal{N}(v)$}
\If{$w = u$} \textbf{continue} \EndIf
\If{$X(w) = 1$}\Comment{found triangle}
\State Add $w$ to $\mathrm{Tri}_e$
\State Remove $w$ from $\mathrm{Star}_u$
\Else \,Add $w$ to $\mathrm{Star}_v$\label{alg:tri_end}
\EndIf
\EndFor
\State $f(g_{3_1},G) \pluseq |\mathrm{Tri}_e|$\label{alg:motifs3_st} 
\State $f(g_{3_2},G) \pluseq |\mathrm{Star}_u| + |\mathrm{Star}_v|$
\State $f(g_{3_3},G) \pluseq |V| - |\mathcal{N}(u) \cup \mathcal{N}(v)|$\label{alg:g3_3}\label{alg:motifs3_en}
\For {$w \in \mathcal{N}(u)$} $X(w) = 0$ \EndFor 
\endpar
\State $f(g_{3_1},G) = \nicefrac{1}{3}.f(g_{3_1},G)$ 
\State $f(g_{3_2},G) = \nicefrac{1}{2}.f(g_{3_2},G)$
\State $f(g_{3_4},G) = {|V| \choose 3} -  f(g_{3_1},G) - f(g_{3_2},G) - f(g_{3_3},G)$\label{alg:g3_4} 
\State \textbf{return} $f(\mathcal{G}_3,G)$
\EndProcedure
\end{algorithmic}
\end{spacing}
\end{algorithm}
\end{minipage}
\end{center}
\end{figure}

\subsubsection*{\textbf{Connected graphlets of size $k=3$.}} Lines~\ref{alg:tri_start}---\ref{alg:tri_end} of Algorithm~\ref{alg:parallel-graphlet} show how to find and count triangles incident to an edge. For any edge $e = (u,v)$, a triangle $(u,v,w)$ exists, if and only if $w$ is connected to \emph{both} $u$ and $v$. Let $\mathrm{Tri}_e$ be the set of all nodes that form a triangle with $e = (u,v)$, and $|\mathrm{Tri}_e|$ be the number of such triangles. Then, $\mathrm{Tri}_e$ is the set of overlapping nodes in the neighborhoods of $u$ and $v$ --- $\mathrm{Tri}_e = \mathcal{N}(u) \cap \mathcal{N}(v)$. Note that Algorithm~\ref{alg:parallel-graphlet} counts each triangle three times (one time for each edge in the triangle), and therefore we divide the total count by $3$ as in Equation~\eqref{eq:triangle},

\begin{equation}\label{eq:triangle}
f(g_{3_1},G) = \frac{1}{3}.\sum\limits_{e=(u,v) \in E} |\mathrm{Tri}_e|
\end{equation}

Now we need to count $2$-star patterns (\ie, $g_{3_2}$). For any edge $e = (u,v)$, let $\textrm{Star}_e$ be the set of all nodes that form a $2$-star with $e$, and $|\textrm{Star}_e|$ be the number of such star patterns. A $2$-star pattern $(u,v,w)$ exists, if and only if $w$ is connected to \emph{either} $u$ or $v$ but not both. Accordingly, $\textrm{Star}_e = \mathrm{Star}_u \cup \mathrm{Star}_v$, where $\mathrm{Star}_u$ and $\mathrm{Star}_v$ are the set of nodes that form a $2$-star with $e$ centered at $u$ and $v$ respectively. More formally, $\mathrm{Star}_u$ can be defined as $\mathrm{Star}_u = \{ w \in \mathcal{N}(u) \setminus \{v\} | w \notin \mathcal{N}(v) \}$, and $\mathrm{Star}_v$ can be defined as $\mathrm{Star}_v = \{ w \in \mathcal{N}(v) \setminus \{u\} | w \notin \mathcal{N}(u) \}$. 

Similar to counting triangles, Algorithm~\ref{alg:parallel-graphlet} counts each $2$-star pattern two times (one time for each edge in the $2$-star). Thus, we divide the sum for all edges by $2$ as follows, 
\begin{equation}\label{eq:2-star}
f(g_{3_2},G) = \frac{1}{2}.\sum\limits_{e=(u,v) \in E} |\mathrm{Star}_u|+|\mathrm{Star}_v|
\end{equation}

\subsubsection*{\textbf{Disconnected graphlets of size $k=3$.}}
There are two disconnected graphlets of size $k=3$ nodes, $g_{3_3}$ (\ie, the $3$-node-1-edge pattern) and $g_{3_4}$  (\ie, the independent set defined on $3$ nodes) (see Table~\ref{table:graphlet_notation}). Lines~\ref{alg:g3_3} and \ref{alg:g3_4} show how to count these patterns. 

Equation~\eqref{eq:g3_3} shows that the number of $3$-node-1-edge graphlets per edge $e$ is equivalent to the number of all nodes that are not in the neighborhood subgraph (egonet) of edge $e$ (\ie, $V \setminus \{\mathcal{N}(u) \cup \mathcal{N}(v)\}$),
\begin{equation}\label{eq:g3_3}
f(g_{3_3},G) = \sum\limits_{e=(u,v) \in E} |V| - |\mathcal{N}(u) \cup \mathcal{N}(v)|
\end{equation} 

\noindent
where $|\mathcal{N}(u) \cup \mathcal{N}(v)| = |\mathrm{Tri}_e|+|\mathrm{Star_e}|+|\{u,v\}|$. Note that the number of $3$-node-1-edge graphlets can be computed in $o(1)$ for each edge.

Given that the total number of graphlets of size $3$ nodes is $N \choose 3$, Equation~\eqref{eq:g3_4} shows how to compute the frequency of $g_{3_4}$, which clearly can be done in $o(1)$,
\begin{equation}\label{eq:g3_4}
f(g_{3_4},G) = {|V| \choose 3} -  \big( f(g_{3_1},G) + f(g_{3_2},G) + f(g_{3_3},G) \big)
\end{equation}

The complexity of counting all graphlets of size $k=3$ is $\mathcal{O}(|E|.\Delta)$ asymptotically as we show next in Lemma~\ref{thm:graphlet_3}. 
\begin{lemma}\label{thm:graphlet_3}
Algorithm~\ref{alg:parallel-graphlet} counts all graphlets of size $k=3$--nodes in $\mathcal{O}(|E|.\Delta)$. 
\end{lemma}
\begin{proof}
\noindent
For each edge $e=(u,v)$ such that $e \in E$, the runtime complexity of counting all triangle and $2$-star patterns incident to $e$ (\ie, $\mathrm{Tri}_e, \textrm{Star}_e$ respectively) is $O(|\mathcal{N}(u)|+|\mathcal{N}(v)|)$, and is asymptotically $O(\Delta)$ where $\Delta$ is the maximum degree in the graph. Further, the runtime complexity of counting all $3$-node-1-edge patterns of size $k=3$ incident to $e$ can be counted in constant time $o(1)$. Therefore, the total runtime complexity for counting all graphlets of size $k=3$ in the graph is $\mathcal{O} \Big(\sum\limits_{e \in E} (\Delta + o(1)) \Big) = \mathcal{O}(|E|.\Delta)$.       
\end{proof}

\section{Counting Graphlets of Size $(k=4)$ Nodes}
\label{sec:motifs-k=4}
An exhaustive search of the egonet of any edge to count all $4$-node graphlets independently yields $\mathcal{O}(\Delta^{3})$ asymptotically, where $\Delta$ is the maximum degree in $G$. Clearly, exhaustive search is computationally intensive for large graphs. On the other hand, our approach is hierarchical and more efficient as we show next. 

For each edge $e=(u,v)$, we start by finding triangles and $2$-star patterns. Our central principle is that any $4$-node graphlet $g_{4_i}$ can be decomposed into four $3$-node graphlets~\cite{HGTbook}, obtained by deleting one node from $g_{4_i}$ each time. Thus, we jointly count all possible $4$-node graphlets by leveraging the knowledge obtained from finding $3$-node graphlets and some combinatorial arguments that describe the relationships between pairs of graphlets. We summarize this procedure in the following steps:
\begin{list}{{$\bullet$}}{\itemsep=2px}
\item\underline{\textsc{Step~$1$}:} For each edge $e$, find all neighborhood nodes forming triangle and $2$-star patterns with $e$.
\item\underline{\textsc{Step~$2$}:} For each edge $e$, use the knowledge from step~$1$ to count only $4$-cliques and $4$-cycles.
\item\underline{\textsc{Step~$3$}:} For each edge $e$, use the knowledge from step~$1$ and some combinatorial arguments to compute unrestricted counts for all $4$-node graphlets in constant time.
\item\underline{\textsc{Step~$4$}:} Merge the counts from all edges in the graph, and use combinatorial arguments involving unrestricted counts to obtain the counts of all other graphlets. 
\end{list}

Note that we refer to the unrestricted counts as the counts that can be computed in constant time and using only the knowledge obtained from step~$1$. Next, we discuss the details of our approach. We start by discussing the graphlet transition diagram to show the pairwise relationships between different $4$-node graphlets. Then, we discuss a general principle for counting $4$-node graphlets, which leverages the graphlet transition diagram and some combinatorial arguments to improve the performance of graphlet counting.

\subsection{Graphlet Transition Diagram}
Assume that each graphlet is a state, Fig.~\ref{fig:transition_diagram} shows all possible \emph{$\pm 1$ edge} transitions between the states of all $4$-node graphlets. We can transition from one graphlet to another by the deletion (denoted by dashed right arrows) or addition (denoted by solid left arrows) of a single edge. We define six different classes of possible edge roles denoted by the colors from black to orange (see Table in the top-right corner in Fig.~\ref{fig:transition_diagram}). An \emph{edge role} is an edge-level connectivity pattern (\eg, a chord edge), where two edges belong to the same role (\ie, class) if they are similar in their topological features. For each edge, we define a topological feature vector that consists of the number of triangles and $2$-stars incident to this edge. Then, we classify edges to one of the six roles based on their feature vectors. Thus, all edges that appear in $4$-node graphlets are colored by their roles. In addition, the transition arrows are colored similar to the edge roles to denote which edge type should be deleted/added to transition from one graphlet to another. Note that a single edge deletion/addition changes the role (class) of other edges in the graphlet. The table in the top-left corner of Fig.~\ref{fig:transition_diagram} shows the number of edge roles per each graphlet.

\begin{figure}
\centering
\includegraphics[width=3.7in]{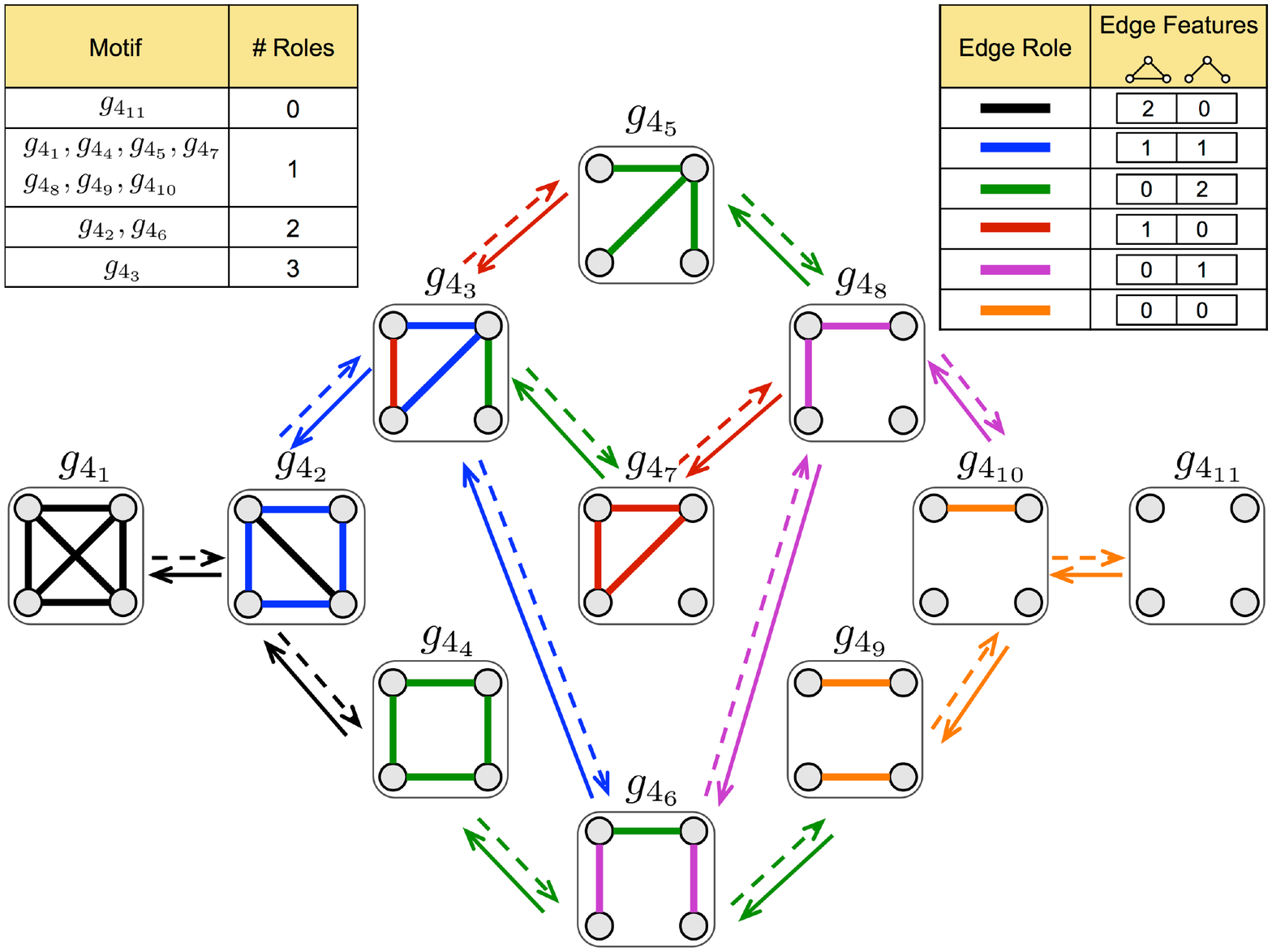}
\caption{\textbf{(4--node) graphlet transition diagram:} Figure shows all possible \emph{$\pm 1$ edge} transitions between the set of all $4$-node graphlets. Dashed right arrows denote the deletion of one edge to transition from one graphlet to another. Solid left arrows denote the addition of one edge to transition from one graphlet to another. Edges are colored by their feature-based roles, where the set of feature are defined by the number of triangles and $2$-stars incident to an edge (see Table in the top-right corner). We define six different classes of edge roles colored from black to orange (see Table in the top-right corner). Dashed/solid arrows are colored similar to the edge roles to denote which edge would be deleted/added to transition from one graphlet to another. The table in the top-left corner shows the number of edge roles per each graphlet.
}
\label{fig:transition_diagram}
\end{figure} 

For example, consider the $4$-clique graphlet ($g_{4_1}$), where each edge participates exactly in two triangles. Therefore, all the edges in a $4$-clique graphlet ($g_{4_1}$) belong to the first role (denoted by the black color). Similarly, consider the $4$-chordalcycle ($g_{4_2}$), where each edge (except the chord edge) participates exactly in one triangle and one $2$-star. Therefore, all edges in a $4$-chordalcycle "$g_{4_2}$" belong to the second role (denoted by the blue color) except for the chord edge which belongs to the first role (denoted by the black color). Fig.~\ref{fig:transition_diagram} shows how to transition from the $4$-clique to the $4$-chordalcycle "$g_{4_2}$" by deleting one (any) edge from the $4$-clique.

\subsection{General Principle for Counting Graphlets of size $k=4$}
Generally speaking, suppose we have $N^{(e)}$ distinct $4$-node subgraphs that contains an edge $e = (u,v)$, 
\begin{equation}
N^{(e)} = \big|\big\{\{u,v,w,r\} \; | \; w,r \in V \setminus \{u,v\} \wedge w \neq r\big\}\big|
\end{equation}

Each subgraph $\{u,v,w,r\}$ in this collection may satisfy one or two properties $a_i, a_j \in A = \{T, S_u, S_v, I \}$. These properties describe the topological properties of nodes $w$ and $r$ with respect to edge $e$, such that $A_w = a_i$ if $\{u,v,w\}$ forms subgraph pattern $a_i$, and $A_r = a_j$ if $\{u,v,r\}$ forms subgraph pattern $a_j$. For example, $A_w = T$ if $w$ forms a triangle with $e$, and $A_w = S_u \textrm{ or } S_v$ if $w$ forms a $2$-star with $e$ centered around $u$ or $v$ respectively. Also, $A_w = I$ if $w$ is independent (disconnected) from $e$. We clarify these properties by example in Fig.~\ref{fig:egonet-example}.

\begin{figure}
\centering
\includegraphics[width=2.6in]{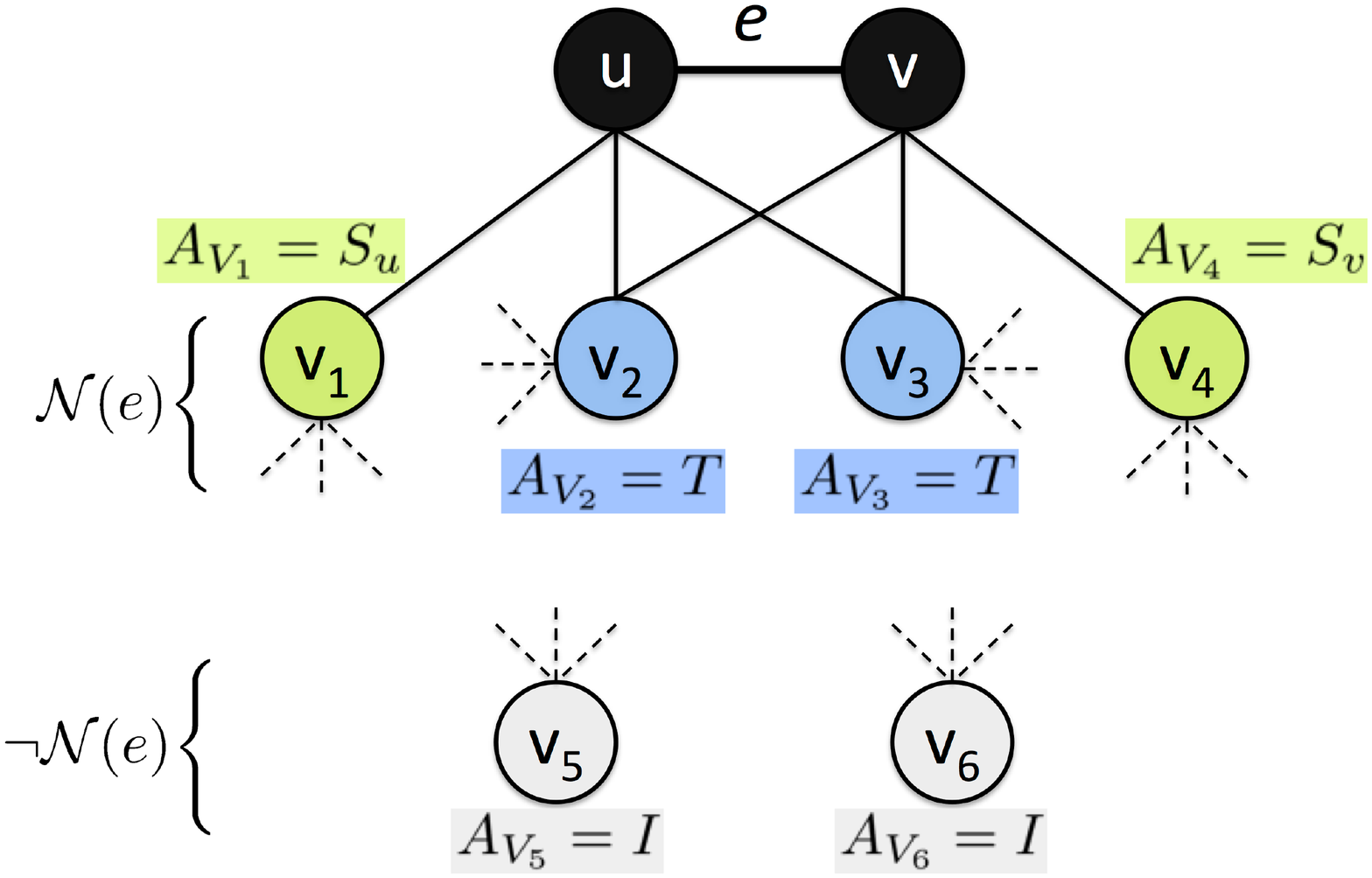}
\caption{
Let $T$ denote the nodes forming triangles with edge $(u,v)$ (\ie, $V_2,V_3$), whereas $S_u$ and $S_v$ denote the nodes forming $2$-stars centered at $u$ and $v$ respectively (\ie, $V_1,V_4$), and let $I$ denote the nodes that are not connected to edge $e$ (\ie, $V_5,V_6$). Further, the dotted lines represent edges incident to these nodes.
}
\label{fig:egonet-example}
\end{figure}

Let $N^{(e)}_{a_i,a_j}$ denote the number having properties $a_i, a_j \in A$,

\begin{equation}\label{eqn:n_ai_aj}
N^{(e)}_{a_i,a_j} = \Bigg|\Bigg\{\{u,v,w,r\} \; \Big| 
\begin{subarray}{l} 
w,r \in V \setminus \{u,v\} \\
\wedge w \neq r \\
\wedge A_w = a_i, A_r =a_j 
\end{subarray}
\Bigg\}\Bigg|
\end{equation}

Now that we defined the topological properties of nodes $w$ and $r$ relative to edge $e$, we need to define whether nodes $w$ and $r$ are connected themselves. Let $e'_{wr}$ represent whether $w$ and $r$ are connected or not, such that $e'_{wr} = 1$ if $(w,r) \in E$ and $e'_{wr} = 0$ otherwise. Accordingly, let $N^{(e)}_{a_i,a_j,e'_{wr}}$ denotes the number of $4$-node graphlets $\{u,v,w,r\}$, where $w,r$ satisfy property $a_i,a_j \in A$ and $e'_{wr} \in \{0,1\}$,

\begin{equation}\label{eqn:n_ai_aj_1}
N^{(e)}_{a_i,a_j,e'_{wr}} = \Bigg|\Bigg\{\{u,v,w,r\} \; \Bigg| 
\begin{subarray}{l} 
w,r \in V \setminus \{u,v\} \\
\wedge w \neq r \\
\wedge A_w = a_i, A_r =a_j \\
\wedge e'_{wr} \in \{0,1\}
\end{subarray}
\Bigg\}\Bigg|
\end{equation}

For example, $N^{(e)}_{T,T,1}$ is the number of all graphlets $\{u,v,w,r\}$ containing edge $e$, where both $w$ and $r$ are forming triangles with $e$ and there exist an edge between $w$ and $r$. Using Equations~\eqref{eqn:n_ai_aj} and ~\eqref{eqn:n_ai_aj_1}, we provide a general principle for graphlet counting in the following theorem. 

\begin{theorem}{General Principle for Graphlet Counting:}\label{thm:inc_exc}
Given a graph $G$, for any edge $e=(u,v)$ in $G$, and for any properties $a_i, a_j \in A$, the number of $4$-node graphlets $\{u,v,w,r\}$ satisfies the following rule, 
\begin{equation}
N^{(e)}_{a_i,a_j,0} = N^{(e)}_{a_i,a_j} - N^{(e)}_{a_i,a_j,1}
\end{equation}
\end{theorem}

\begin{proof}
Suppose there is a subgraph $\{u,v,w,r\}$ containing edge $e$, where nodes $w$ and $r$ satisfy $a_i, a_j$ properties respectively, and $(w,r) \in E$. Then the expression on the right side counts this subgraph once in the $N^{(e)}_{a_i,a_j}$ term, and once in the $N^{(e)}_{a_i,a_j,1}$. By the principle of inclusion-exclusion~\cite{stanley1986enumerative}, the total contribution of the subgraph $\{u,v,w,r\}$ in $N^{(e)}_{a_i,a_j,0}$ is zero. Thus, $N^{(e)}_{a_i,a_j,0}$ is the number of graphlets having properties $a_i, a_j$, but $(w,r) \notin E$.        
\end{proof}

Clearly, it is sufficient to compute $N^{(e)}_{a_i,a_j}$ and $N^{(e)}_{a_i,a_j,1}$ only, and use Theorem~\ref{thm:inc_exc} to compute $N^{(e)}_{a_i,a_j,0}$ in constant time. Note that $N^{(e)}_{a_i,a_j}$ is an unrestricted count and can be computed in constant time using the knowledge we have from finding $3$-node graphlets.

To simplify the discussion in the following sections, we precisely show how to compute $N^{(e)}_{a_i,a_j}$, the number of $4$-node graphlets $\{u,v,w,r\}$ such that $w,r$ satisfy property $a_i,a_j \in A$ respectively. Let $\mathcal{W}_{a_i}$ be the set of nodes with property $a_i \in A$ (\ie, $\mathcal{W}_{a_i} = \{ w \in V \setminus \{u,v\} \;|\; A_w = a_i, \forall a_i \in A\}$), and similarly $\mathcal{R}_{a_j}$ be the set of nodes with property $a_j \in A$ (\ie, $\mathcal{R}_{a_j} = \{ r \in V \setminus \{u,v\} \;|\; A_r = a_j, \forall a_j \in A\}$). If $a_i = a_j$, then $\mathcal{W}_{a_i} = \mathcal{R}_{a_j}$. Thus,

\begin{equation}\label{eqn:comb_n_ai_ai}
N^{(e)}_{a_i,a_i} = {|\mathcal{W}_{a_i}| \choose 2} = \frac{1}{2}.(|\mathcal{W}_{a_i}| - 1).|\mathcal{W}_{a_i}|
\end{equation}   

However, if $a_i \neq a_j$, then $\mathcal{W}_{a_i}$ and $\mathcal{R}_{a_j}$ are mutually exclusive (\ie, $\mathcal{W}_{a_i} \cap \mathcal{R}_{a_j} = \emptyset$). 

\noindent
Thus, we get the following,

\begin{equation}\label{eqn:comb_n_ai_aj}
N^{(e)}_{a_i,a_j} = |\mathcal{W}_{a_i}|.|\mathcal{R}_{a_j}| 
\end{equation}

\subsection{Analysis \& Combinatorial Arguments}
\label{sec:lemmas}
In this section, we discuss combinatorial arguments involving unrestricted counts that can be computed computed directly from our knowledge of $3$-node graphlets.  These combinatorial arguments capture the relationships between the counts of pairs of $4$-node graphlets. The proofs of these relationships are based on Theorem~\ref{thm:inc_exc} and the transition diagram in Fig.~\ref{fig:transition_diagram}. For each pair of graphlets $g_{4_i}$ and $g_{4_j}$, we show the relationship for each edge in the graph (in Corollary~\ref{thm:cliq}--\ref{thm:4node-1edg}), then we show a generalization for the whole graph (in Lemma~\ref{thm:rel-cliq-chcy}--\ref{thm:rel-4node-2edge-1edge}).

\subsubsection{Relationship between $4$-Cliques \& $4$-ChordalCycles}
\begin{cor}\label{thm:cliq}
For any edge $e=(u,v)$ in the graph, the number of $4$-cliques containing $e$ is $N^{(e)}_{T,T,1}$.
\end{cor}

\begin{cor}\label{thm:chcy}
For any edge $e=(u,v)$ in the graph, the number of $4$-chordalcycles, where $e$ is the chord edge of the cycle (denoted by the black color in Fig.~\ref{fig:transition_diagram}), is $N^{(e)}_{T,T,0}$.
\end{cor}

\begin{lemma}\label{thm:rel-cliq-chcy}
For any graph $G$, the relationship between the counts of $4$-cliques (\ie, $f(g_{4_1},G)$) and $4$-chordalcycles (\ie, $f(g_{4_2},G)$) is,
\begin{equation}
f(g_{4_2},G) = \sum\limits_{e \in E} {|\mathrm{Tri}_e| \choose 2} - 6.f(g_{4_1},G) \nonumber
\end{equation}
\end{lemma}

\begin{proof}
From Theorem~\ref{thm:inc_exc} and the addition principle~\cite{stanley1986enumerative}, the total count for all edges in $G$ is,  
\begin{equation}\label{eqn:add_n_TT}
\sum\limits_{e \in E} N^{(e)}_{T,T,0} = \sum\limits_{e \in E} N^{(e)}_{T,T} - \sum\limits_{e \in E} N^{(e)}_{T,T,1}
\end{equation}

Given that $N^{(e)}_{T,T}$ is the number of $4$-node subgraphs $\{u,v,w,r\}$ containing $e$, such that $A_w=T,A_r=T$. Thus, from Eq.~\eqref{eqn:comb_n_ai_ai}, $N^{(e)}_{T,T} = {|\mathrm{Tri}_e| \choose 2}$. From Corollary~\ref{thm:cliq}, each $4$-clique will be counted $6$ times (once for each edge in the clique). Thus, the total count of $4$-cliques in $G$ is $f(g_{4_1},G) = \frac{1}{6}.\sum\limits_{e \in E} N^{(e)}_{T,T,1}$. Similarly, from Corollary~\ref{thm:chcy}, each $4$-chordalcycle is counted only once for each chord edge. Thus, the total count of $4$-chordalcycles in $G$ is $f(g_{4_2},G) = \sum\limits_{e \in E} N^{(e)}_{T,T,0}$. By direct substitution in Eq.~\eqref{eqn:add_n_TT}, this lemma is true. 
\end{proof}

\subsubsection{Relationship between $4$-Cycles \& $4$-Paths}
\begin{cor}\label{thm:cyc}
For any edge $e=(u,v)$ in the graph, the number of $4$-cycles containing $e$ is $N^{(e)}_{S_u,S_v,1}$.
\end{cor}

\begin{cor}\label{thm:path}
For any edge $e=(u,v)$ in the graph, the number of $4$-paths containing $e$, where $e$ is the middle edge in the path (denoted by the green color in Fig.~\ref{fig:transition_diagram}), is $N^{(e)}_{S_u,S_v,0}$.
\end{cor}

\begin{lemma}\label{thm:rel-cyc-path}
For any graph $G$, the relationship between the counts of $4$-cycles (\ie, $f(g_{4_4},G)$) and $4$-paths (\ie, $f(g_{4_6},G)$) is,
\begin{equation}
f(g_{4_6},G) = \sum\limits_{e \in E} {|\mathrm{Star}_u|.|\mathrm{Star}_v|} - 4.f(g_{4_4},G) \nonumber
\end{equation}
\end{lemma}

\begin{proof}
From Theorem~\ref{thm:inc_exc} and the addition principle~\cite{stanley1986enumerative}, the total count for all edges in $G$ is,  
\begin{equation}\label{eqn:add_n_SuSv}
\sum\limits_{e \in E} N^{(e)}_{S_u,S_v,0} = \sum\limits_{e \in E} N^{(e)}_{S_u,S_v} - \sum\limits_{e \in E} N^{(e)}_{S_u,S_v,1}
\end{equation}

Given that $N^{(e)}_{S_u,S_v}$ is the number of $4$-node subgraphs $\{u,v,w,r\}$ containing $e$, such that $w,r$ $A_w=S_u,A_r=S_v$. Thus, from Eq.~\eqref{eqn:comb_n_ai_aj}, $N^{(e)}_{S_u,S_v} = |\mathrm{Star}_u|.|\mathrm{Star}_v|$. From Corollary~\ref{thm:cyc}, each $4$-cycle will be counted $4$ times (once for each edge in the cycle). Thus, the total count of $4$-cycles in $G$ is $f(g_{4_4},G) = \frac{1}{4}.\sum\limits_{e \in E} N^{(e)}_{S_u,S_v,1}$. Similarly, from Corollary~\ref{thm:path}, each $4$-path is counted only once for each middle edge in the path. Thus, the total count of $4$-paths in $G$ is $f(g_{4_6},G) = \sum\limits_{e \in E} N^{(e)}_{S_u,S_v,0}$. By direct substitution in Eq.~\eqref{eqn:add_n_SuSv}, this lemma is true. 
\end{proof}

\subsubsection{Relationship between $4$-TailedTriangles \& $4$-ChordalCycles}
\begin{cor}\label{thm:ttri}
For any edge $e=(u,v)$ in the graph, the number of $4$-tailedtriangles where $e$ is part of both the triangle and $2$-star patterns (denoted by the blue color in Fig.~\ref{fig:transition_diagram}), is $N^{(e)}_{T,S_u \vee S_v,0}$.
\end{cor}

\begin{cor}\label{thm:chcy2}
For any edge $e=(u,v)$ in the graph, the number of $4$-chordalcycles where $e$ is a cycle edge (denoted by the blue color in Fig.~\ref{fig:transition_diagram}), is $N^{(e)}_{T,S_u \vee S_v,1}$.
\end{cor}

\begin{lemma}\label{thm:rel-ttri-chcy2}
For any graph $G$, the relationship between the counts of $4$-chordalcycles (\ie, $f(g_{4_2},G)$) and $4$-tailedtriangles (\ie, $f(g_{4_3},G)$) is,
\begin{equation}
2.f(g_{4_3},G) = \sum\limits_{e \in E} {|\mathrm{Tri}_e|.(|\mathrm{Star}_u|+|\mathrm{Star}_v|)} - 4.f(g_{4_2},G) \nonumber
\end{equation}
\end{lemma}

\begin{proof}
From Theorem~\ref{thm:inc_exc} and the addition principle~\cite{stanley1986enumerative}, the total count for all edges in $G$ is,  
\begin{equation}\label{eqn:add_n_TSuSv}
\sum\limits_{e \in E} N^{(e)}_{T,S_u \vee S_v,0} = \sum\limits_{e \in E} N^{(e)}_{T,S_u \vee S_v} -\sum\limits_{e \in E} N^{(e)}_{T,S_u \vee S_v,1}
\end{equation}

Given that $N^{(e)}_{T,S_u \vee S_v} = N^{(e)}_{T,S_u} +  N^{(e)}_{T,S_v}$ is the number of $4$-node subgraphs $\{u,v,w,r\}$ containing $e$, such that $A_w=T,A_r=S_u \vee S_v$. Thus, from Eq.~\eqref{eqn:comb_n_ai_aj}, $N^{(e)}_{T,S_u \vee S_v} = |\mathrm{Tri}_e|.(|\mathrm{Star}_u|+|\mathrm{Star}_v|)$. Now, from Corollary~\ref{thm:chcy2}, each $4$-chordalcycle is counted $4$ times (once for each edge in the cycle). Thus, the total count of $4$-chordalcycle in $G$ is $f(g_{4_2},G) = \frac{1}{4}.\sum\limits_{e \in E} N^{(e)}_{T,S_u \vee S_v,1}$. Similarly, from Corollary~\ref{thm:ttri}, each $4$-tailedtriangle will be counted $2$ times (once for each blue edge as in Fig.~\ref{fig:transition_diagram}). Thus, the total count of $4$-tailedtriangle in $G$ is $f(g_{4_3},G) = \frac{1}{2}.\sum\limits_{e \in E} N^{(e)}_{T,S_u \vee S_v,0}$. By direct substitution in Eq.~\eqref{eqn:add_n_TSuSv}, this lemma is true. 
\end{proof}

\subsubsection{Relationship between $4$-TailedTriangles \& $3$-Stars}
\begin{cor}\label{thm:ttri2}
For any edge $e=(u,v)$ in the graph, the number of $4$-tailedtriangles with $e$ as the tail edge (denoted by the green color in Fig.~\ref{fig:transition_diagram}) and $u$ is part of the triangle, is $N^{(e)}_{S_u,S_u,1}$.
\end{cor}

In a similar fashion, the number of $4$-tailedtriangles with $e$ as the tail edge and $v$ is part of the triangle is $N^{(e)}_{S_v,S_v,1}$. Thus, the total number of $4$-tailedtriangles with $e$ as the tail edge and $u \vee v$ is part of the triangle is $N^{(e)}_{S_{\textbf{.}},S_{\textbf{.}},1} = N^{(e)}_{S_u,S_u,1} + N^{(e)}_{S_v,S_v,1}$.

\begin{cor}\label{thm:3star}
For any edge $e=(u,v)$ in the graph, the number of $3$-star centered around $u$ is $N^{(e)}_{S_u,S_u,0}$.
\end{cor}

Again, the number of $3$-stars centered around $v$ is $N^{(e)}_{S_v,S_v,0}$. Thus, the total number of $3$-stars centered around $u \textrm{ or } v$ is $N^{(e)}_{S_\textbf{.},S_\textbf{.},0} = N^{(e)}_{S_u,S_u,0}+N^{(e)}_{S_v,S_v,0}$. 

\begin{lemma}\label{thm:rel-ttri-3star}
For any graph $G$, the relationship between the counts of $3$-stars (\ie, $f(g_{4_5},G)$) and $4$-tailedtriangles (\ie, $f(g_{4_3},G)$) is,
\begin{equation}
3.f(g_{4_5},G) = \sum\limits_{e \in E} {{|\mathrm{Star}_u| \choose 2} + {|\mathrm{Star}_v| \choose 2}} - f(g_{4_3},G) \nonumber
\end{equation}
\end{lemma}

\begin{proof}
From Theorem~\ref{thm:inc_exc} and the addition principle~\cite{stanley1986enumerative}, the total count for all edges in $G$ is,  
\begin{equation}\label{eqn:add_n_SuSu_SvSv}
\sum\limits_{e \in E} N^{(e)}_{S_\textbf{.},S_\textbf{.},0} = \sum\limits_{e \in E} N^{(e)}_{S_\textbf{.},S_\textbf{.}} - \sum\limits_{e \in E} N^{(e)}_{S_\textbf{.},S_\textbf{.},1}
\end{equation}

Given that $N^{(e)}_{S_\textbf{.},S_\textbf{.}} = N^{(e)}_{S_u,S_u}+N^{(e)}_{S_v,S_v}$ is the number of $4$-node subgraphs $\{u,v,w,r\}$ containing $e$, such that $A_w=S_u \wedge A_r=S_u \textrm{ or } A_w=S_v \wedge A_r=S_v$. Thus, from Eq.~\eqref{eqn:comb_n_ai_ai}, $N^{(e)}_{S_\textbf{.},S_\textbf{.}} = {|\mathrm{Star}_u| \choose 2} + {|\mathrm{Star}_v| \choose 2}$. Now, from Corollary~\ref{thm:3star}, each $3$-star is counted $3$ times (once for each edge in the star). Thus, the total count of $3$-stars in $G$ is $f(g_{4_5},G) = \frac{1}{3}.\sum\limits_{e \in E} N^{(e)}_{S_\textbf{.},S_\textbf{.},0}$. Similarly, from Corollary~\ref{thm:ttri2}, each $4$-tailedtriangle will be counted once for each tail edge (denoted by the green color in Fig.~\ref{fig:transition_diagram}). Thus, the total count of $4$-tailedtriangle in $G$ is $f(g_{4_3},G) = \sum\limits_{e \in E} N^{(e)}_{S_\textbf{.},S_\textbf{.},1}$. This holds whether the patterns are centered around $u$ or $v$. By direct substitution in Eq.~\eqref{eqn:add_n_SuSu_SvSv}, this lemma is true. 
\end{proof}

\subsubsection{Relationship between $4$-TailedTriangles \& $4$-Node-1-Triangles}
\begin{cor}\label{thm:4node-tri}
For any edge $e=(u,v)$ in the graph, the number of $4$-node-1-triangle is $N^{(e)}_{T,I,0}$.
\end{cor}

\begin{cor}\label{thm:ttri3}
For any edge $e=(u,v)$ in the graph, the number of $4$-tailedtriangles with $e$ participating in the triangle but not connected to the tail edge (denoted by the red color in Fig.~\ref{fig:transition_diagram}), is $N^{(e)}_{T,I,1}$.
\end{cor}

\begin{proof}
Suppose there is a subgraph $\{u,v,w,r\}$ containing $e$. $\{u,v,w,r\}$ is a $4$-tailedtriangle with $e$ participating in the triangle but not connected to the tail edge, if and only if there are some nodes $w,r$ such that $w \in \mathrm{Tri}_e$, $r \not \mathcal{N}(e)$, and $(w,r) \in E$. This means $r$ is independent of $e$, and $w$ forms a triangle with $e$. As such, $A_w=T$ and $A_r=I$ and $e'_{wr}=1$. More generally, any subgraph $\{u,v,w,r\}$ containing $e$ contributes once in the count $N^{(e)}_{T,I,1}$ if and only if it is a $4$-tailedtriangle with $e$ participating in the triangle but not connected to the tail edge. In Theorem~\ref{thm:inc_exc}, we showed that $N^{(e)}_{T,I,1} \leq N^{(e)}_{T,I}$.
\end{proof}

\begin{lemma}\label{thm:rel-4node-tri-ttri}
For any graph $G$, the relationship between the counts of $4$-tailedtriangles (\ie, $f(g_{4_3},G)$) and $4$-node-1-triangles (\ie, $f(g_{4_7},G)$) is,
\begin{equation}
3.f(g_{4_7},G) = \sum\limits_{e \in E} \Big(\mathrm{Tri}_e.\left( |V| - |\mathcal{N}(u) \cup \mathcal{N}(v)| \right)\Big) - f(g_{4_3},G) \nonumber
\end{equation}
\end{lemma}

\begin{proof}
From Theorem~\ref{thm:inc_exc} and the addition principle~\cite{stanley1986enumerative}, the total count for all edges in $G$ is,  
\begin{equation}\label{eqn:add_n_T_I}
\sum\limits_{e \in E} N^{(e)}_{T,I,0} = \sum\limits_{e \in E} N^{(e)}_{T,I} - \sum\limits_{e \in E} N^{(e)}_{T,I,1}
\end{equation}

Given that $N^{(e)}_{T,I}$ is the number of $4$-node subgraphs $\{u,v,w,r\}$ containing $e$, such that $A_w=T, A_r=I$. And, the number of nodes independent of $e$ is $|V| - |\mathcal{N}(u) \cup \mathcal{N}(v)|$. Thus, from Eq.~\eqref{eqn:comb_n_ai_aj}, $N^{(e)}_{T,I} = \mathrm{Tri}_e . \Big( |V| - |\mathcal{N}(u) \cup \mathcal{N}(v)| \Big)$. Now, from Corollary~\ref{thm:ttri3}, each $4$-tailedtriangle is counted one time (once for the red edge as in Fig.~\ref{fig:transition_diagram}). Thus, the total count of $4$-tailedtriangles in $G$ is $f(g_{4_3},G) = \sum\limits_{e \in E} N^{(e)}_{T,I,1}$. Similarly, from Corollary~\ref{thm:4node-tri}, each $4$-node-1-triangle will be counted $3$ times (once for each edge in the triangle). Thus, the total count of $4$-node-1-triangles in $G$ is $f(g_{4_7},G) = \frac{1}{3}.\sum\limits_{e \in E} N^{(e)}_{T,I,0}$. By direct substitution in Eq.~\eqref{eqn:add_n_T_I}, this lemma is true. 
\end{proof}

\subsubsection{Relationship between $4$-Paths \& $4$-node-$2$-Stars}
\begin{cor}\label{thm:path2}
For any edge $e=(u,v)$ in the graph, the number of $4$-paths where $e$ is the start or end of the path (denoted by the purple color in Fig.~\ref{fig:transition_diagram}), is $N^{(e)}_{S_u \vee S_v,I,1}$.
\end{cor}

\begin{cor}\label{thm:4node_2star}
For any edge $e=(u,v)$ in the graph, the number of $4$-node-2-stars where $e$ is one of the star edges (denoted by the purple color in Fig.~\ref{fig:transition_diagram}), is $N^{(e)}_{S_u \vee S_v,I,0}$.
\end{cor}

\begin{lemma}\label{thm:rel-4node-2star-path}
For any graph $G$, the relationship between the counts of $4$-paths (\ie, $f(g_{4_6},G)$) and $4$-node-$2$-stars (\ie, $f(g_{4_8},G)$) is,
\begin{equation}
2.f(g_{4_8},G) = \sum\limits_{e \in E} |\mathrm{Star}_e|.(|V| - |\mathcal{N}(u) \cup \mathcal{N}(v)|) - 2.f(g_{4_6},G) \nonumber
\end{equation}
\end{lemma}

\begin{proof}
From Theorem~\ref{thm:inc_exc} and the addition principle~\cite{stanley1986enumerative}, the total count for all edges in $G$ is,  
\begin{equation}\label{eqn:add_n_SuSv_I}
\sum\limits_{e \in E} N^{(e)}_{S_u \vee S_v,I,0} = \sum\limits_{e \in E} N^{(e)}_{S_u \vee S_v,I} -\sum\limits_{e \in E} N^{(e)}_{S_u \vee S_v,I,1}
\end{equation}

Given that $N^{(e)}_{S_u \vee S_v,I} = N^{(e)}_{S_u,I} +  N^{(e)}_{S_v,I}$ is the number of $4$-node subgraphs $\{u,v,w,r\}$ containing $e$, such that $A_w=S_u \vee S_v, A_r=I$. And, the number of nodes independent of $e$ is $|V| - |\mathcal{N}(u) \cup \mathcal{N}(v)|$. Thus, from Eq.~\eqref{eqn:comb_n_ai_aj}, $N^{(e)}_{S_u \vee S_v,I} = |\mathrm{Star}_e|.\left( |V| - |\mathcal{N}(u) \cup \mathcal{N}(v)| \right)$, such that $|\mathrm{Star}_e| = |\mathrm{Star}_u| + |\mathrm{Star}_v|$. Now, from Corollary~\ref{thm:path2}, each $4$-path is counted $2$ times (for both the start and end edges in the path, denoted by the purple in Fig.~\ref{fig:transition_diagram}). Thus, the total count of $4$-paths in $G$ is $f(g_{4_6},G) = \frac{1}{2}.\sum\limits_{e \in E} N^{(e)}_{S_u \vee S_v,I,1}$. Similarly, from Corollary~\ref{thm:4node_2star}, each $4$-node-$2$-star will be counted $2$ times (once for each edge in the star, denoted by the purple in Fig.~\ref{fig:transition_diagram}). Thus, the total count of $4$-node-$2$-star in $G$ is $f(g_{4_8},G) = \frac{1}{2}.\sum\limits_{e \in E} N^{(e)}_{S_u \vee S_v,I,0}$. By direct substitution in Eq.~\eqref{eqn:add_n_SuSv_I}, this lemma is true. 
\end{proof}

\subsubsection{Relationship between $4$-node-$2$-edges \& $4$-node-$1$-edge}
\begin{cor}\label{thm:4node-2edgs}
For any edge $e=(u,v)$ in the graph, the number of $4$-node-$2$-edges where $e$ is any of the two independent edges in the graphlet (denoted by the orange color in Fig.~\ref{fig:transition_diagram}), is $N^{(e)}_{I,I,1}$.
\end{cor}

\begin{cor}\label{thm:4node-1edg}
For any edge $e=(u,v)$ in the graph, the number of $4$-node-$1$-edge where $e$ is an isolated/single edge in the graphlet (denoted by the orange color in Fig.~\ref{fig:transition_diagram}), is $N^{(e)}_{I,I,0}$.
\end{cor}

\begin{lemma}\label{thm:rel-4node-2edge-1edge}
For any graph $G$, the relationship between the counts of $4$-node-$2$-edge graphlets (\ie, $f(g_{4_9},G)$) and $4$-node-$1$-edge graphlets (\ie, $f(g_{4_{10}},G)$) is,
\begin{equation}
f(g_{4_{10}},G) = \sum\limits_{e \in E} {{|V| - |\mathcal{N}(u) \cup \mathcal{N}(v)|} \choose 2} - 2.f(g_{4_9},G) \nonumber
\end{equation}
\end{lemma}

\begin{proof}
From Theorem~\ref{thm:inc_exc} and the addition principle~\cite{stanley1986enumerative}, the total count for all edges in $G$ is,  
\begin{equation}\label{eqn:add_n_I_I}
\sum\limits_{e \in E} N^{(e)}_{I,I,0} = \sum\limits_{e \in E} N^{(e)}_{I,I} -\sum\limits_{e \in E} N^{(e)}_{I,I,1}
\end{equation}

Given that $N^{(e)}_{I,I}$ is the number of $4$-node subgraphs $\{u,v,w,r\}$ containing $e$, such that $A_w=I, A_r=I$. And, the number of nodes independent of $e$ is $|V| - |\mathcal{N}(u) \cup \mathcal{N}(v)|$. Thus, from Eq.~\eqref{eqn:comb_n_ai_ai}, $N^{(e)}_{I,I} = {{|V| - |\mathcal{N}(u) \cup \mathcal{N}(v)|} \choose 2}$. Now, from Corollary~\ref{thm:4node-2edgs}, each $4$-node-$2$-edge is counted $2$ times (for the two edges in the graphlet, denoted by the orange in Fig.~\ref{fig:transition_diagram}). Thus, the total count of $4$-node-$2$-edges in $G$ is $f(g_{4_{9}},G) = \frac{1}{2}.\sum\limits_{e \in E} N^{(e)}_{I,I,1}$. Similarly, from Corollary~\ref{thm:4node-1edg}, each $4$-node-$1$-edge will be counted once (for the isolated/single edge in the graphlet, denoted by the orange in Fig.~\ref{fig:transition_diagram}). Thus, the total count of $4$-node-$1$-edge in $G$ is $f(g_{4_{10}},G) = \sum\limits_{e \in E} N^{(e)}_{I,I,0}$. By direct substitution in Eq.~\eqref{eqn:add_n_I_I}, this lemma is true. 
\end{proof}

While it is straightforward to compute $N^{(e)}_{I,I}$ for each edge $e$, this is not the case for $N^{(e)}_{I,I,1}$ or $N^{(e)}_{I,I,0}$, as they require searching outside the local edge neighborhood. However, since $N^{(e)}_{I,I,1}$ is the number of edges outside the egonet of $e$, it can be computed as,
\begin{align*}
N^{(e)}_{I,I,1} &= |E| - |\mathcal{N}(u) \setminus \{v\}| - |\mathcal{N}(v) \setminus \{u\}| - |\{e\}|  \\
&\quad - [N^{(e)}_{T,T,1} + N^{(e)}_{T,S_u \vee S_v,1} + N^{(e)}_{T,I,1}]  \\ 
&\quad - [N^{(e)}_{S_{\textbf{.}},S_{\textbf{.}},1} + N^{(e)}_{S_u,S_v,1} + N^{(e)}_{S_{\textbf{.}},I,1}]
\end{align*}

\noindent
Thus, the total number of $4$-node-$2$-edges is,
\begin{align}\label{eqn:g_4_9}
2.f(g_{4_9},G) &= \sum\limits_{e \in E} N^{(e)}_{I,I,1} \\
&= \sum\limits_{e \in E} |E| - |\mathcal{N}(u) \setminus \{v\}| - |\mathcal{N}(v) \setminus \{u\}| - |\{e\}| \nonumber \\
&\quad - [6.f(g_{4_1},G) + 4.f(g_{4_2},G) + 2.f(g_{4_3},G)] \nonumber \\
&\quad - [4.f(g_{4_4},G) + 2.f(g_{4_6},G)] \nonumber
\end{align}

\noindent
Finally, the number of $4$-node-independent graphlets ($g_{4_{11}}$) is,
\begin{equation}\label{eqn:g_4_11}
f(g_{4_{11}},G) = {|V| \choose 4} - \sum\limits_{i=1}^{10} f(g_{4_i},G)
\end{equation}

\subsection{Algorithm}
Algorithm~\ref{alg:parallel-graphlet_all} (\textsc{GraphletCounting}) shows how to count all graphlets of size $k=\{3,4\}$ nodes efficiently (using Lemma~\ref{thm:rel-cliq-chcy}---~\ref{thm:rel-4node-2edge-1edge}). As discussed previously, we start by finding all triangle and $2$-star patterns in Lines~\ref{tri_st}--\ref{tri_en} (\ie, \textsc{Step~$1$}). Then, in Lines~\ref{cliqcycy_st}---\ref{cliqcycy_en} we only count $4$-cliques and $4$-cycles (\ie, \textsc{Step~$2$}). Then, Lines~\ref{ubounds_st}---\ref{ubounds_en} compute unrestricted counts for all $4$-node graphlets in constant time (using knowledge from \textsc{Step~$1$} and $2$, \ie, \textsc{Step~$3$}), and finally Lines~\ref{getcounts_st}---\ref{getcounts_en} compute the final counts (using the lemma proved in Section~\ref{sec:lemmas}) (\ie, \textsc{Step~$4$}). Our approach counts all $4$-cliques and $4$-cycles in $\mathcal{O}(m.\Delta.T_{max})$ and $\mathcal{O}(m.\Delta.S_{max})$ respectively, where $T_{max}$ is the maximum number of triangles incident to an edge and $T_{max} \ll \Delta$ for sparse graphs, and $S_{max}$ is the maximum number of stars incident to an edge and $S_{max} \leq \Delta$,  as we show in Lemma~\ref{thm:clique_4} and \ref{thm:cycle_4}. This is more efficient than $\mathcal{O}(|V|.\Delta^3)$ given by~\cite{shervashidze2009efficient}, and $\mathcal{O}(\Delta.|E|+|E|^2)$ given by~\cite{marcus2012rage}.

\noindent
\begin{lemma}\label{thm:clique_4}
Alg.~\ref{alg:parallel-graphlet_all} counts all $4$-cliques in $\mathcal{O}(|E|.\Delta.T_{max})$, where $T_{max}$ is the maximum number of triangles incident to an edge. 
\end{lemma}

\begin{proof}
For each edge $e=(u,v) \in E$, the runtime complexity of counting all $4$-cliques incident to $e$ is equivalent to finding the set of all edges $e'=(w,w')$ such that $\{e'=(w,w') \in E | w,w' \in \mathrm{Tri}_e \wedge w \neq w'\}$, where $\mathrm{Tri}_e$ is the set of triangles incident to $e$. First, we show in Lem.~\ref{thm:graphlet_3} that the runtime complexity of finding all triangles incident to $e$ is $\mathcal{O}(\Delta)$. Second, as described in Alg.~\ref{alg:parallel-graphlet_all} the runtime complexity of checking whether any two distinct nodes $w, w' \in \mathrm{Tri}_e$ are connected by an edge $e'=(w,w')$ is $\mathcal{O}(\sum\limits_{w \in \mathrm{Tri}_e} \Delta) = \mathcal{O}(|\mathrm{Tri}_e|.\Delta)$, and can be computed asymptotically $\mathcal{O}(T_{max}.\Delta)$, where $T_{max}$ is the maximum triangle degree (i.e., the maximum number of triangles incident to an edge and $T_{max} \ll \Delta$). Therefore, the total runtime complexity is $\mathcal{O} \Big(\sum\limits_{e \in E} (\Delta + T_{max}.\Delta) \Big) = \mathcal{O}(|E|.\Delta.T_{max})$.
\end{proof}

\noindent
\begin{lemma}\label{thm:cycle_4}
Alg.~\ref{alg:parallel-graphlet_all} counts all $4$-cycles of size $k=4$ in $\mathcal{O}(|E|.\Delta.S_{max})$, where $S_{max}$ is the maximum number of $2$-stars incident to an edge (proof is similar to Lem.~\ref{thm:clique_4}).  
\end{lemma}
\begin{proof}
For each edge $e=(u,v) \in E$, the runtime complexity of counting all $4$-cycles incident to $e$ is equivalent to finding the set of all edges $e'=(w,w')$ such that $\{e'=(w,w') \in E | w \in \mathrm{Star}_u \wedge w' \in \mathrm{Star}_v, w \neq w'\}$. First, we show in Lem.~\ref{thm:graphlet_3} that the runtime complexity of finding all $2$-star patterns incident to $e$ is $\mathcal{O}(\Delta)$. Second, Alg.~\ref{alg:parallel-graphlet_all} shows the runtime complexity of checking whether any two distinct nodes $w \in \mathrm{Star}_u$, and $w' \in \mathrm{Star}_v$ are connected by an edge $e'=(w,w')$ is $\mathcal{O}(\sum\limits_{w \in \mathrm{Star}_u} \Delta) = \mathcal{O}(|\mathrm{Star}_u|.\Delta)$, and is asymptotically $\mathcal{O}(S_{max}.\Delta)$ (where $S_{max}$ is the maximum number of $2$-stars incident to an edge, and $S_{max} \leq \Delta$). Therefore, the total runtime complexity is $\mathcal{O} \Big(\sum\limits_{e \in E} (\Delta + S_{max}.\Delta) \Big) = \mathcal{O}(|E|.\Delta.S_{max})$.
\end{proof}

\algrenewcommand{\alglinenumber}[1]{\scriptsize#1:}
\begin{figure}
\vspace{-4.mm}
\begin{center}
\begin{minipage}{1.0\linewidth}
\begin{algorithm}[H]
\caption{\,\small{Our exact graphlet census algorithm for counting all $3,4$-node graphlets. The algorithm takes an undirected graph as input and returns the frequencies of all $3,4$-node graphlets}}
\label{alg:parallel-graphlet_all}
\begin{spacing}{1.2}
\fontsize{8}{9}\selectfont
\begin{algorithmic}[1]
\Procedure {GraphletCounting}{$G=\left (V,E\right )$}
\State Initialize Array $X$
\State $N_{T,T} = 0$, $N_{S_u,S_v} = 0$, $N_{T,S_u \vee S_v} = 0$, $N_{S_\textbf{.},S_\textbf{.}} = 0$
\State $N_{T,I} = 0$, $N_{S_u \vee S_v,I} = 0$, $N_{I,I} = 0$, $N_{I,I,1} = 0$
\parfor[$e=(u,v) \in E$]
\State $\mathrm{Star}_u = \emptyset,  \mathrm{Star}_v = \emptyset, \mathrm{Tri}_e = \emptyset$ 
\For {$w \in \mathcal{N}(u)$} \label{tri_st}
\If{$w = v$} \textbf{continue}
\EndIf
\State Add $w$ to $\mathrm{Star}_u$ and set $X(w) = 1$
\EndFor
\For {$w \in \mathcal{N}(v)$}
\If{$w = u$} \textbf{continue}
\EndIf
\If{$X(w) = 1$} \Comment{found triangle}
\State Add $w$ to $\mathrm{Tri}_e$ and set $X(w) = 2$
\State Remove $w$ from $\mathrm{Star}_u$
\Else \,Add $w$ to $\mathrm{Star}_v$ and set $X(w)=3$ \label{tri_en}
\EndIf
\EndFor
\State Compute $f(\mathcal{G}_3,G)$ as in Lines~\ref{alg:motifs3_st}---\ref{alg:motifs3_en} of Alg.~\ref{alg:parallel-graphlet}
\State //\,Get Counts of $4$-Cliques \& $4$-Cycles
\State $f(g_{4_1},G) \pluseq$ \textsc{CliqueCount}$(X,\mathrm{Tri}_e)$ \label{cliqcycy_st}
\State $f(g_{4_4},G) \pluseq$ \textsc{CycleCount}$(X,\mathrm{Star}_u)$ \label{cliqcycy_en}
\State //\,Get Unrestricted Counts for $4$-Node Connected Graphlets 
\State $N_{T,T} \pluseq {|\mathrm{Tri}_e| \choose 2}$ \label{ubounds_st}
\State $N_{S_u,S_v} \pluseq |\mathrm{Star}_u|.|\mathrm{Star}_v|$
\State $N_{T,S_u \vee S_v} \pluseq |\mathrm{Tri}_e|.(|\mathrm{Star}_u|+|\mathrm{Star}_v|)$
\State $N_{S_u,S_u} = {|\mathrm{Star}_u| \choose 2}$ and $N_{S_v,S_v} = {|\mathrm{Star}_v| \choose 2}$
\State $N_{S_\textbf{.},S_\textbf{.}} \pluseq N_{S_u,S_u} + N_{S_v,S_v}$
\State //\,Get Unrestricted Counts for $4$-Node Disconnected Graphlets
\State $N_{T,I} \pluseq \mathrm{Tri}_e . (|V| - |\mathcal{N}(u) \cup \mathcal{N}(v)|)$
\State $N_{S_u,I} = |\mathrm{Star}_u|.(|V| - |\mathcal{N}(u) \cup \mathcal{N}(v)|)$
\State $N_{S_v,I} = |\mathrm{Star}_v|.(|V| - |\mathcal{N}(u) \cup \mathcal{N}(v)|)$
\State $N_{S_u \vee S_v,I} \pluseq N_{S_u,I} +  N_{S_v,I}$
\State $N_{I,I} \pluseq {{|V| - |\mathcal{N}(u) \cup \mathcal{N}(v)|} \choose 2}$
\State $N_{I,I,1} \pluseq |E| - |\mathcal{N}(u) \setminus \{v\}| - |\mathcal{N}(v) \setminus \{u\}| - 1$ \label{ubounds_en}
\For {$w \in \mathcal{N}(v)$} $X(w) = 0$ \EndFor 
\endpar
\State Use Lemma~\ref{thm:rel-cliq-chcy}---\ref{thm:rel-4node-tri-ttri} to compute $f(g_{4_i},G)$ for $i=1:8$ \label{getcounts_st}
\State Use Eq.~\eqref{eqn:g_4_9} to compute $f(g_{4_9},G)$ and Lemma~\ref{thm:rel-4node-2edge-1edge} for $f(g_{4_{10}},G)$
\State Use Eq.~\eqref{eqn:g_4_11} to compute $f(g_{4_{11}},G)$ \label{getcounts_en}
\State \textbf{return} $f(\mathcal{G}_3,G), f(\mathcal{G}_4,G)$
\EndProcedure
\hrulefill
\Procedure {CliqueCount}{$X,\mathrm{Tri}_e$}
\State $\textrm{cliq}_e = 0$
\ForAll {\textrm{node} $w \in \mathrm{Tri}_e$}
	\For{$r \in \mathcal{N}(w)$}
		\If{$X(r) = 2$} $\textrm{cliq}_e \pluseq 1$ \Comment{found $4$-Clique}
		\EndIf
	\EndFor
	\State $X(w) = 0$ 
\EndFor
\State \textbf{return} $\textrm{cliq}_e$
\EndProcedure
\hrulefill
\Procedure {CycleCount}{$X,\mathrm{Star}_u$}
\State $\textrm{cyc}_e = 0$
\ForAll {\textrm{node} $w \in \mathrm{Star}_u$}
	\For{$r \in \mathcal{N}(w)$}
		\If{$X(r) = 3$} $\textrm{cyc}_e \pluseq 1$ \Comment{found $4$-Cycle}
		\EndIf
	\EndFor
	\State $X(w) = 0$
\EndFor
\State \textbf{return} $\textrm{cyc}_e$ 
\EndProcedure
\end{algorithmic}
\end{spacing}
\vspace{-1.mm}
\end{algorithm}
\end{minipage}
\end{center}
\vspace{-3.mm}
\end{figure}

\section{Experiments}
\label{sec:experiments}
We proceed by first demonstrating how fast our algorithm (Algorithm~\ref{alg:parallel-graphlet_all}) counts all graphlets of size $k=\{3,4\}$ (both connected and disconnected graphlets) on various networks. We make all our implementations, further experiments, and proofs available in an online appendix\footnote{{\scriptsize \url{http://nesreenahmed.com/graphlets}}}. In this paper, we show detailed results for $60$ networks categorized in $8$ broad classes from social, facebook~\cite{traud2012social}, biological, web, technological, co-authorship, infrastructure, among other domains~\cite{nr-aaai15} (see the links\footnote{{\scriptsize \url{http://networkrepository.com/}}} for data download). And, in the online appendix, we present a more extensive collection of $300$+ networks, including both large sparse networks as well as dense networks from the DIMACs challenge\footnote{{\scriptsize \url{http://dimacs.rutgers.edu/Challenges/}}}. Note that for all of the networks, we discard edge weights, self-loops, and edge direction. To the best of our knowledge, this is the largest study for graphlet counting, and these are the largest graphlet computations published to date. Our own implementation of Algorithm.~\ref{alg:parallel-graphlet_all} uses shared memory, but the algorithm is well-suited for other architectures.

{
\setlength{\tabcolsep}{2.5pt}
\begin{table}
\caption{Runtime \& Statistics for a Subset of $60$ Networks. The numbers are appended by $K$ for thousands, $M$ for millions, $B$ for billions, $T$ for trillions, and $P$ for quadrillions.}
\medskip
\label{table:graphlets-fb100}
\small\scriptsize
\scalebox{0.9}{
\begin{tabularx}{1.15\textwidth}{rrHHHrHrrHHrrrrrrHHHHHHHHrHrHH@{}}
\toprule
\multicolumn{25}{c}{} &
\multicolumn{3}{c}{\textbf{Seconds}}
\\
$\mathbf{graph}$ & 
$|V|$ & $\Delta$ & $\bar{d}$ & $\kappa$ & $|E|$ & $|E|_{\rm I}$ & $|g_{3_1}|$ & $|g_{3_2}|$ & $3_{\rm 1e}$ & $3_{\rm I}$ & $|g_{4_1}|$ & $|g_{4_2}|$ & $|g_{4_4}|$ & $|g_{4_6}|$ & $|g_{4_5}|$ & $|g_{4_3}|$ & $4_{\rm 1e}$ & $4_{\rm 2e}$ & $4_{\rm S3}$ & $4_T$ & $4_{\rm I}$ & $k$ & $w$ & $r$ &
Alg.\ref{alg:parallel-graphlet_all} & $chk$ & RAGE & & \\
\bottomrule
\\
\textsf{\tiny soc-brightkite}  & 57k & 1.1k & 7 & 0.11 & 213k & 1.6B & 494k & 12M & 12B & 30T & 2.9M & 12M & 2.7M & 533M & 1.3B & 114M & 341T & 22B & 672B & 28B & 0 & 45 & 46 & 0.01 & 0.2 & 0 & 273.03 &  & \\ 
\textsf{\tiny socfb-Berkeley13}  & 23k & 3.4k & 74 & 0.11 & 852k & 262M & 5.4M & 125M & 19B & 2.0T & 27M & 153M & 87M & 17B & 25B & 2.7B & 217T & 343B & 2.8T & 120B & 0 & 40 & 41 & 0.01 & 4.94 & 0 & 2514.59 &  & \\ 
\textsf{\tiny socfb-Wisconsin87}  & 24k & 3.5k & 70 & 0.12 & 836k & 283M & 4.9M & 107M & 20B & 2.2T & 23M & 121M & 59M & 12B & 21B & 1.9B & 232T & 335B & 2.5T & 114B & 0 & 51 & 52 & -0 & 3.93 & 0 & 1450.31 &  & \\ 
\textsf{\tiny socfb-FSU53}  & 28k & 2.6k & 74 & 0.15 & 1.0M & 384M & 7.9M & 130M & 28B & 3.5T & 63M & 242M & 95M & 16B & 10B & 2.9B & 389T & 515B & 3.5T & 214B & 0 & 50 & 51 & 0.1 & 5.55 & 0 & 2192.94 &  & \\ 
\textsf{\tiny socfb-MSU24}  & 32k & 5.3k & 69 & 0.12 & 1.1M & 523M & 6.5M & 139M & 36B & 5.6T & 33M & 183M & 106M & 16B & 32B & 2.6B & 575T & 606B & 4.4T & 206B & 0 & 48 & 49 & 0.01 & 5.67 & 0 & 1904.09 &  & \\ 
\textsf{\tiny socfb-Texas80}  & 32k & 1.8k & 77 & 0.15 & 1.2M & 497M & 9.6M & 160M & 38B & 5.2T & 68M & 316M & 122M & 21B & 11B & 3.9B & 595T & 717B & 5.0T & 299B & 0 & 73 & 74 & 0.16 & 7.53 & 0 & 2967.01 &  & \\ 
\textsf{\tiny socfb-Michigan23}  & 30k & 2.0k & 78 & 0.13 & 1.2M & 453M & 8.3M & 162M & 35B & 4.5T & 49M & 277M & 146M & 23B & 13B & 3.5B & 523T & 665B & 4.8T & 244B & 0 & 21 & 22 & 0.12 & 7.57 & 0 & 2995.83 &  & \\ 
\textsf{\tiny socfb-Indiana69}  & 30k & 1.4k & 87 & 0.14 & 1.3M & 441M & 9.4M & 181M & 38B & 4.3T & 60M & 269M & 141M & 25B & 13B & 3.8B & 564T & 822B & 5.3T & 275B & 0 & 64 & 65 & 0.13 & 8.44 & 0 & 3212.10 &  & \\ 
\textsf{\tiny socfb-UIllinois20}  & 31k & 4.6k & 82 & 0.14 & 1.3M & 473M & 9.4M & 172M & 39B & 4.8T & 64M & 273M & 130M & 23B & 27B & 3.8B & 587T & 772B & 5.2T & 283B & 0 & 26 & 27 & 0.03 & 7.88 & 0 & 3088.77 &  & \\ 
\textsf{\tiny socfb-UF21}  & 35k & 8.2k & 83 & 0.12 & 1.5M & 615M & 12M & 266M & 51B & 7.2T & 98M & 433M & 186M & 40B & 150B & 7.2B & 882T & 1.0T & 8.8T & 418B & 0 & 67 & 68 & -0.01 & 14.49 & 0 & N/A &  & \\ 
\textsf{\tiny soc-flickr}  & 514k & 4.4k & 12 & 0.15 & 3.2M & 132B & 59M & 963M & 1.6T & 23P & 1.7B & 14B & 6.7B & 244B & 326B & 90B & 420P & 4.7T & 493T & 30T & 0 & 0 & 1 & 0.16 & 182.57 & 0 & N/A &  & \\ 
\textsf{\tiny soc-orkut}  & 3.1M & 33k & 76 & 0.04 & 117M & 4.7T & 628M & 44B & 360T & 17E & 3.2B & 48B & 70B & 19T & 98T & 1.5T & 18E & 6.8P & 134P & 1.9P & 0 & 252 & 253 & 0.02 & 2694.55 & 0 & N/A &  & \\ 
\textsf{\tiny soc-sinaweibo}  & 58M &  &  &  & 261M &  & 212M & 804B &  &  & 662M & 27B & 259B & 157T & 8.48P & 3.80T &  &  &  &  &  &  &  &  & 33359.7 & 0 & N/A &  & \\ 
\textsf{\tiny soc-friendster}  & 65.6M &  &  &  & 1.8B &  & 4.17B & 708.1B &  &  & 8.96B & 131.4B & 307.5B & 364.7T & 247.3T & 5.79T &  &  &  &  &  &  &  &  & N/A & 0 & N/A &  & \\ 
\midrule
\textsf{\tiny bio-celegans}  & 453 & 237 & 8 & 0.12 & 2.0k & 100k & 3.3k & 69k & 765k & 15M & 3.0k & 37k & 4.5k & 495k & 2.9M & 363k & 147M & 1.0M & 21M & 1.0M & 0 & 4 & 5 & -0.23 & \textless$0.001$ & 0 & 1.7 & 4.05 & \\ 
\textsf{\tiny bio-diseasome}  & 516 & 50 & 4 & 0.43 & 1.2k & 132k & 1.4k & 5.4k & 596k & 22M & 1.4k & 923 & 42 & 18k & 27k & 19k & 148M & 653k & 2.6M & 672k & 0 & 2 & 3 & 0.07 & \textless$0.001$ & 0 & 0.44 & 0.07 & \\ 
\textsf{\tiny bio-dmela}  & 7.4k & 190 & 6 & 0.01 & 26k & 27M & 2.9k & 572k & 188M & 67B & 393 & 13k & 107k & 11M & 9.2M & 312k & 689B & 314M & 4.2B & 21M & 0 & 9 & 10 & -0.05 & 0.01 & 0 & 2.47 & 19.02 & \\ 
\textsf{\tiny bio-yeast-protein-inter}  & 1.8k & 56 & 2 & 0.06 & 2.2k & 1.7M & 222 & 11k & 4.0M & 1.0B & 41 & 198 & 140 & 31k & 72k & 2.6k & 3.7B & 2.4M & 21M & 406k & 0 & 4 & 5 & -0.16 & \textless$0.001$ & 0 & 0.53 & 0.13 & \\ 
\textsf{\tiny bio-yeast}  & 1.5k & 56 & 2 & 0.05 & 1.9k & 1.1M & 206 & 11k & 2.8M & 513M & 39 & 195 & 139 & 31k & 72k & 2.5k & 2.0B & 1.9M & 16M & 297k & 0 & 1 & 2 & -0.21 & \textless$0.001$ & 0 & 0.43 & 0.08 & \\ 
\textsf{\tiny bio-human-gene2}  & 14k & 7.2k & 1.3k & 0.59 & 9.0M & 89M & 4.9B & 10B & 91B & 353B & 2.3T & 3.7T & 90B & 4.4T & 5.3T & 8.4T & 445T & 14T & 94T & 44T & 0 & 1.9k & 1.9k & 0.08 & 8023.84 & 0 & N/A &  & \\ 
\textsf{\tiny bio-mouse-gene}  & 43k & 8.0k & 670 & 0.42 & 14M & 915M & 3.6B & 15B & 583B & 13T & 670B & 2.1T & 223B & 9.0T & 6.7T & 7.7T & 12P & 81T & 586T & 141T & 0 & 1.0k & 1.0k & 0 & 5515.6 & 0 & N/A &  & \\ 
\midrule
\textsf{\tiny ca-CSphd}  & 1.9k & 46 & 1 & 0 & 1.7k & 1.8M & 8 & 6.6k & 3.3M & 1.1B & 0 & 5 & 8 & 9.4k & 32k & 93 & 3.0B & 1.5M & 12M & 15k & 0 & 1 & 2 & -0.2 & \textless$0.001$ & 0 & 1.25 &  & \\ 
\textsf{\tiny ca-GrQc}  & 4.2k & 81 & 6 & 0.63 & 13k & 8.6M & 48k & 85k & 55M & 12B & 329k & 66k & 1.1k & 553k & 406k & 628k & 114B & 88M & 348M & 196M & 0 & 2 & 3 & 0.64 & \textless$0.001$ & 0 & 5.99 &  & \\ 
\textsf{\tiny ca-dblp-2012}  & 317k & 343 & 6 & 0.31 & 1.0M & 50B & 2.2M & 15M & 333B & 5.3P & 17M & 4.8M & 203k & 252M & 259M & 97M & 53P & 551B & 4.8T & 705B & 0 & 112 & 113 & 0.27 & 0.48 & 0 & 227.79 &  & \\ 
\textsf{\tiny ca-cit-HepTh}  & 23k & 8.7k & 213 & 0.27 & 2.4M & 260M & 191M & 1.6B & 52B & 1.9T & 13B & 47B & 7.3B & 538B & 976B & 385B & 558T & 1.9T & 31T & 3.9T & 0 & 560 & 561 & -0.03 & 132.66 & 0 & N/A &  & \\ 
\textsf{\tiny ca-cit-HepPh}  & 28k & 4.9k & 224 & 0.28 & 3.1M & 391M & 196M & 1.5B & 85B & 3.6T & 9.8B & 34B & 6.1B & 536B & 479B & 276B & 1.1P & 4.0T & 39T & 5.1T & 0 & 409 & 410 & 0.03 & 125.49 & 0 & N/A &  & \\ 
\textsf{\tiny ca-coauthors-dblp}  & 540k & 3.3k & 56 & 0.66 & 15M & 146B & 444M & 698M & 8.2T & 26P & 15B & 3.4B & 31M & 42B & 27B & 67B & 2.2E & 116T & 377T & 240T & 0 & 335 & 336 & 0.51 & 40.26 & 0 & N/A &  & \\ 
\textsf{\tiny ca-hollywood-2009}  & 1.1M & 11k & 105 & 0.31 & 56M & 571B & 4.9B & 33B & 60T & 204P & 1.4T & 635B & 168B & 21T & 17T & 8.9T & 14E & 1.5P & 35P & 5.2P & 0 & 2.2k & 2.2k & 0.35 & 13799.6 & 1 & N/A &  & \\ 
\midrule
\textsf{\tiny tech-as-caida2007}  & 26k & 2.6k & 4 & 0.01 & 53k & 350M & 36k & 15M & 1.4B & 3.1T & 54k & 1.7M & 407k & 285M & 7.8B & 47M & 18T & 1.1B & 368B & 912M & 0 & 0 & 1 & -0.19 & 0.19 & 0 & 36.83 &  & \\ 
\textsf{\tiny tech-p2p-gnutella}  & 63k & 95 & 4 & 0 & 148k & 2.0B & 2.0k & 1.6M & 9.2B & 41T & 16 & 826 & 42k & 15M & 8.1M & 71k & 289T & 11B & 98B & 127M & 0 & 0 & 1 & -0.09 & 0.02 & 0 & 7.44 &  & \\ 
\textsf{\tiny tech-RL-caida}  & 191k & 1.1k & 6 & 0.06 & 608k & 18B & 455k & 21M & 116B & 1.2P & 423k & 7.4M & 40M & 583M & 1.7B & 77M & 11P & 184B & 4.0T & 87B & 0 & 1 & 2 & 0.02 & 0.39 & 0 & 71.74 &  & \\ 
\textsf{\tiny tech-WHOIS}  & 7.5k & 1.1k & 15 & 0.31 & 57k & 28M & 782k & 5.3M & 413M & 69B & 12M & 31M & 2.9M & 229M & 566M & 194M & 1.5T & 1.1B & 37B & 5.5B & 0 & 2 & 3 & -0.04 & 0.14 & 0 & 44.52 &  & \\ 
\textsf{\tiny tech-as-skitter}  & 1.7M & 35k & 13 & 0.01 & 11M & 1.4T & 29M & 16B & 19T & 811P & 149M & 20B & 43B & 819B & 96T & 162B & 16E & 60T & 27P & 49T & 0 & 4 & 5 & -0.08 & 476.06 & 0 & N/A &  & \\ 
\midrule
\textsf{\tiny web-BerkStan-dir}  & 685k & 84k & 19 & 0.01 & 6.6M & 235B & 65M & 28B & 4.5T & 54P & 1.1B & 99B & 25B & 49B & 382T & 476B & 1.5E & 21T & 18P & 44T & 0 & 17 & 18 & -0.21 & 149.17 & 1 & N/A &  & \\ 
\textsf{\tiny web-edu}  & 3.0k & 104 & 4 & 0.27 & 6.5k & 4.6M & 10k & 81k & 19M & 4.6B & 40k & 4.6k & 18 & 435k & 1.3M & 186k & 29B & 20M & 239M & 30M & 0 & 2 & 3 & -0.17 & \textless$0.001$ & 0 & 0.52 &  & \\ 
\textsf{\tiny web-google-dir}  & 876k & 6.3k & 9 & 0.06 & 4.3M & 383B & 13M & 687M & 3.8T & 112P & 40M & 382M & 38M & 4.1B & 650B & 6.7B & 1.7E & 9.3T & 600T & 12T & 0 & 43 & 44 & -0.06 & 4.45 & 0 & N/A &  & \\ 
\textsf{\tiny web-indochina-2004}  & 11k & 199 & 8 & 0.57 & 48k & 64M & 210k & 481k & 539M & 244B & 1.2M & 88k & 9.2k & 5.5M & 12M & 4.9M & 3.0T & 1.1B & 5.4B & 2.4B & 0 & 48 & 49 & 0.12 & 0.01 & 0 & 24.36 &  & \\ 
\textsf{\tiny web-it-2004}  & 509k & 469 & 28 & 0.95 & 7.2M & 130B & 339M & 56M & 3.7T & 22P & 29B & 815M & 175M & 1.1B & 1.4B & 527M & 930P & 26T & 28T & 172T & 0 & 430 & 431 & 0.99 & 25.26 & 0 & N/A &  & \\ 
\textsf{\tiny web-baidu-baike}  & 2.1M & 98k & 15 & 0 & 17M & 2.3T & 25M & 31B & 36T & 17E & 28M & 4.5B & 9.2B & 3.3T & 571T & 327B & 1.9E & 141T & 64P & 54T & 0 & 15 & 16 & -0.3 & 3975.81 & 1 & N/A &  & \\ 
\textsf{\tiny web-wikipedia-growth}  & 1.9M & 226k & 39 & 0 & 37M & 1.7T & 127M & 123B & 68T & 1.1E & 288M & 38B & 68B & 29T & 3.1P & 3.2T & 8.1E & 635T & 220P & 234T & 0 & 180 & 181 & -0.53 & 22389.2 & 1 & N/A &  & \\ 
\textsf{\tiny web-ClueWeb09-50m}  & 148M & 308k & 6 & 0.01 & 447M & 11P & 1.2B & 494B & 66P & 18E & 5.6B & 243B & 774B & 34T & 24P & 3.4T & 39P & 100P & 8.4E & 184P & 0 & 62 & 63 & -0.35 & 15665.9 & 6 & N/A &  & \\ 
\midrule
\textsf{\tiny inf-italy-osm}  & 6.7M & 9 & 2 & 0 & 7.0M & 22T & 7.4k & 8.2M & 47T & 633P & 0 & 244 & 47k & 9.9M & 992k & 27k & 9.2E & 25T & 55T & 50B & 0 & 2 & 3 & 0.25 & 0.85 & 0 & N/A &  & \\ 
\textsf{\tiny inf-openflights}  & 2.9k & 242 & 10 & 0.25 & 16k & 4.3M & 73k & 639k & 45M & 4.2B & 286k & 1.5M & 319k & 17M & 17M & 9.0M & 63B & 91M & 1.8B & 201M & 0 & 26 & 27 & 0.05 & 0.01 & 0 & 2.46 &  & \\ 
\textsf{\tiny inf-power}  & 4.9k & 19 & 2 & 0.1 & 6.6k & 12M & 651 & 17k & 33M & 20B & 90 & 385 & 324 & 38k & 20k & 5.1k & 80B & 22M & 84M & 3.2M & 0 & 1 & 2 & 0 & \textless$0.001$ & 0 & 0.58 & 0.17 & \\ 
\textsf{\tiny inf-roadNet-CA}  & 2.0M & 12 & 2 & 0.06 & 2.8M & 1.9T & 120k & 5.6M & 5.4T & 1.2E & 40 & 13k & 249k & 11M & 2.4M & 521k & 5.3E & 3.8T & 11T & 236B & 0 & 2 & 3 & 0.12 & 0.35 & 0 & N/A &  & \\ 
\textsf{\tiny inf-roadNet-PA}  & 1.1M & 9 & 2 & 0.06 & 1.5M & 591B & 67k & 3.2M & 1.7T & 214P & 16 & 5.7k & 152k & 6.2M & 1.4M & 295k & 912P & 1.2T & 3.5T & 73B & 0 & 2 & 3 & 0.12 & 0.19 & 0 & N/A &  & \\ 
\textsf{\tiny inf-road-usa}  & 24M & 9 & 2 & 0.03 & 29M & 287T & 439k & 50M & 691T & 1.5E & 90 & 21k & 1.6M & 81M & 18M & 1.5M & 9.5E & 416T & 1.2P & 11T & 0 & 1 & 2 & 0.07 & 4.05 & 0 & N/A &  & \\ 
\midrule
\textsf{\tiny ia-email-EU-dir}  & 265k & 7.6k & 2 & 0 & 364k & 35B & 267k & 194M & 96B & 3.1P & 581k & 10M & 6.7M & 4.4B & 221B & 341M & 13P & 61B & 51T & 70B & 0 & 36 & 37 & -0.18 & 1.52 & 0 & 887.18 & N/A & \\ 
\textsf{\tiny ia-enron-only}  & 143 & 42 & 8 & 0.36 & 623 & 9.5k & 889 & 4.8k & 76k & 396k & 779 & 2.7k & 648 & 29k & 17k & 14k & 4.3M & 135k & 522k & 102k & 0 & 1 & 2 & -0.02 & \textless$0.001$ & 0 & 0.12 & 0.03 & \\ 
\textsf{\tiny ia-reality}  & 6.8k & 261 & 2 & 0 & 7.7k & 23M & 400 & 497k & 51M & 53B & 63 & 1.7k & 2.8k & 1.6M & 26M & 93k & 171B & 27M & 3.3B & 2.6M & 0 & 0 & 1 & -0.68 & \textless$0.001$ & 0 & 1.39 & 7.86 & \\ 
\textsf{\tiny ia-wiki-Talk-dir}  & 2.4M & 100k & 3 & 0 & 4.7M & 2.9T & 9.2M & 13B & 11T & 18E & 65M & 1.0B & 924M & 1.2T & 192T & 64B & 13E & 9.6T & 30P & 22T & 0 & 57 & 58 & -0.29 & 281.33 & 0 & N/A & N/A & \\ 
\textsf{\tiny ia-wikiquote-user-edits}  & 93k & 17k & 5 & 0 & 238k & 4.4B & 279k & 636M & 21B & 136T & 411k & 70M & 44M & 8.9B & 2.4T & 2.5B & 929T & 16B & 52T & 23B & 0 & 28 & 29 & -0.3 & 2.41 & 0 & 691.28 &  & \\ 
\textsf{\tiny ia-wiki-user-edits-page}  & 2.1M & 699k & 5 & 0 & 5.6M & 2.2T & 6.7M & 550B & 11T & 1.5E & 10M & 70B & 44B & 4.8T & 88P & 2.0T & 10E & 8.0T & 888P & 12T & 0 & 21 & 22 & 1 & 5691.92 & 4 & N/A &  & \\ 
\midrule
\textsf{\tiny brock200-3}  & 200 & 134 & 120 & 0.61 & 12k & 7.9k & 291k & 570k & 372k & 80k & 3.2M & 12M & 4.1M & 11M & 3.5M & 16M & 2.2M & 1.7M & 6.9M & 3.5M & 0 & 104 & 105 & -0.01 & 0.02 & 0 & 22.96 &  & \\ 
\textsf{\tiny brock200-4}  & 200 & 147 & 130 & 0.66 & 13k & 6.8k & 373k & 584k & 303k & 53k & 5.2M & 16M & 4.3M & 8.9M & 3.0M & 17M & 1.2M & 1.1M & 4.6M & 2.9M & 0 & 116 & 117 & -0.02 & 0.02 & 0 & 21.85 &  & \\ 
\textsf{\tiny brock400-3}  & 400 & 322 & 298 & 0.75 & 60k & 20k & 4.4M & 4.5M & 1.5M & 170k & 184M & 372M & 63M & 84M & 28M & 251M & 4.8M & 7.1M & 28M & 28M & 0 & 277 & 278 & -0.01 & 0.4 & 0 & 997.15 &  & \\ 
\textsf{\tiny brock400-4}  & 400 & 326 & 298 & 0.75 & 60k & 20k & 4.4M & 4.5M & 1.5M & 168k & 185M & 373M & 63M & 84M & 28M & 250M & 4.7M & 7.0M & 28M & 28M & 0 & 276 & 277 & -0.01 & 0.4 & 0 & 1010.26 &  & \\ 
\textsf{\tiny brock800-1}  & 800 & 560 & 518 & 0.65 & 208k & 112k & 23M & 38M & 20M & 3.7M & 1.3B & 4.1B & 1.1B & 2.4B & 801M & 4.4B & 350M & 324M & 1.3B & 800M & 0 & 486 & 487 & -0 & 4.11 & 0 & N/A &  & \\ 
\textsf{\tiny brock800-2}  & 800 & 566 & 520 & 0.65 & 208k & 111k & 23M & 38M & 20M & 3.6M & 1.3B & 4.2B & 1.1B & 2.4B & 794M & 4.4B & 341M & 318M & 1.3B & 793M & 0 & 485 & 486 & -0 & 4.15 & 0 & N/A &  & \\ 
\textsf{\tiny brock800-3}  & 800 & 558 & 518 & 0.65 & 207k & 112k & 23M & 38M & 20M & 3.7M & 1.3B & 4.1B & 1.1B & 2.4B & 802M & 4.4B & 353M & 325M & 1.3B & 801M & 0 & 482 & 483 & -0 & 4.1 & 0 & N/A &  & \\ 
\bottomrule
\multicolumn{10}{c}{\tiny N/A: timed out after $30$ hours of runtime \;\;\;}
\end{tabularx}
}
\end{table}
}

\subsection{Efficiency \& Runtime}
Table~\ref{table:graphlets-fb100} describes the properties of the $60$ networks considered here. It also shows the counts of graphlets of size $k=\{3,4\}$ and states the time (seconds) taken to count all graphlets. We only show counts of connected graphlets due to space limitations, however all counts are available in the online appendix. Notably, Algorithm~\ref{alg:parallel-graphlet_all} takes only few seconds to count all graphlets for large social, web, and technological graphs (among others). For example, for a large road network (\ie, inf-road-usa) with $24$M nodes and $29$M edges, Algorithm~\ref{alg:parallel-graphlet_all} takes only $4$ seconds  to count all graphlets. Also as shown in Table~\ref{table:graphlets-fb100}, for large facebook networks with nearly $2$M edges, Algorithm~\ref{alg:parallel-graphlet_all} takes only $15$ seconds, and for large web graphs with nearly $8$M edges, Algorithm~\ref{alg:parallel-graphlet_all} takes only $25$ seconds. 

We compare the empirical runtime of Algorithm~\ref{alg:parallel-graphlet_all} to the state-of-the-art baseline method RAGE~\cite{marcus2012rage}. For social and facebook networks, we observed that Algorithm~\ref{alg:parallel-graphlet_all} is on average $460$x faster than RAGE. For all other networks, we observed that Algorithm~\ref{alg:parallel-graphlet_all} is on average $600$x faster than RAGE. Notably, Algorithm~\ref{alg:parallel-graphlet_all} takes only $7$ seconds to count graphlets of facebook networks with $1.3$M edges, while RAGE takes almost an hour for the same networks. For larger networks with millions of nodes/edges, RAGE was timed out (as it did not finish within $30$ hours of runtime). Moreover, for dense graphs from the DIMACS challenge, RAGE takes almost $17$ minutes, while Algorithm~\ref{alg:parallel-graphlet_all} takes less than a second. We also compared to the baseline method FANMOD~\cite{wernicke2006fanmod} and Orca~\cite{hovcevar2014combinatorial}, we found that for a facebook network with $250$k edges, FANMOD takes roughly $2.5$ hours for counting all graphlets, RAGE takes almost $7$ minutes for the same network, and Orca takes almost $10$ seconds, while Algorithm~\ref{alg:parallel-graphlet_all} takes less than a second. 
Note that both RAGE and Orca count only connected graphlets, while our algorithm and FANMOD count both connected and disconnected graphlets. 

In Figure~\ref{fig:runtime-soc}, we plot the runtime of Algorithm~\ref{alg:parallel-graphlet_all} for a representative subset of $150$ social and information networks. The figure shows that our algorithm exhibits nearly linear-time scaling over networks ranging from $1$K to $100$M nodes.

\begin{figure}
        \centering
        \includegraphics[width=3in]{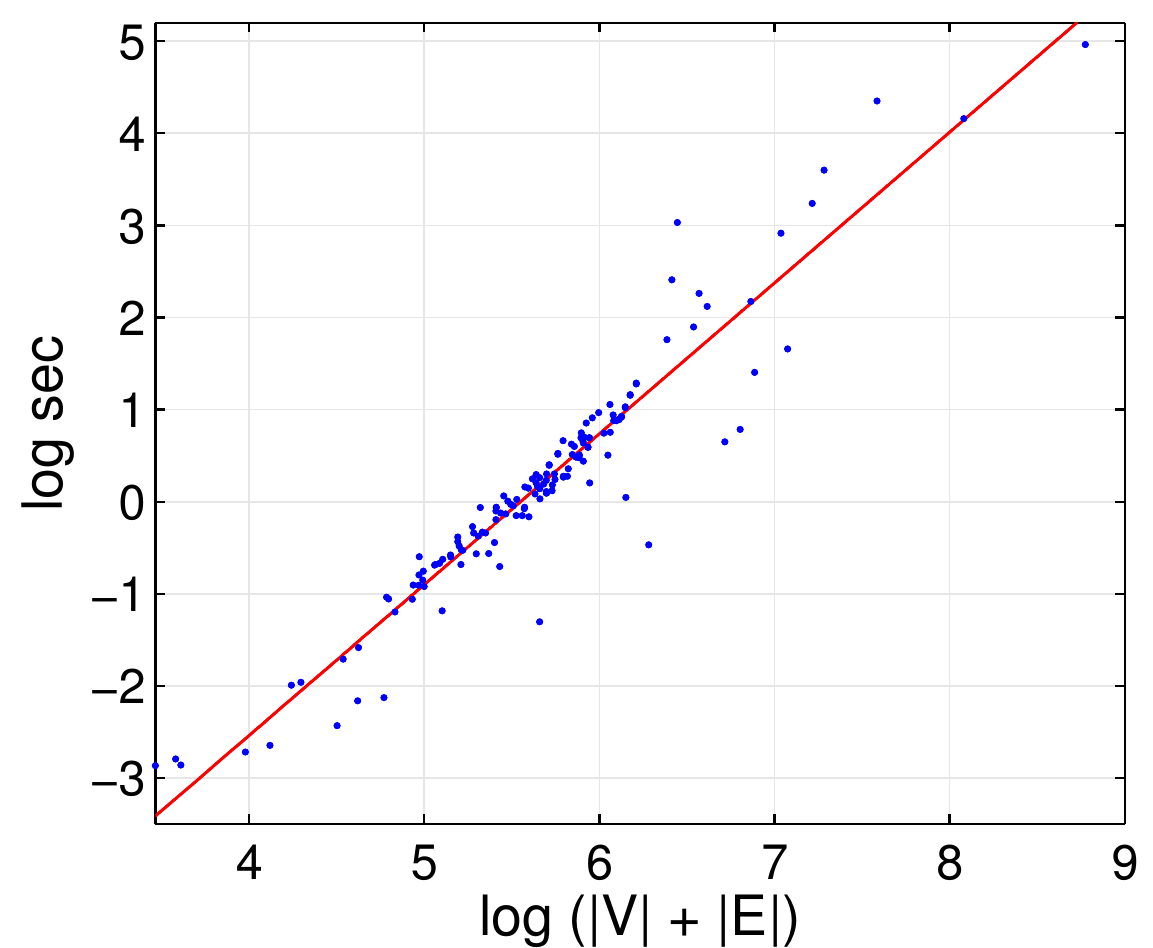}
        \caption{The empirical runtime of our exact graphlet counting (Alg.\ref{alg:parallel-graphlet_all}) in social and information networks scales almost linearly with the network dimension.}
        \label{fig:runtime-soc}
\end{figure}

\subsection{Scaling} \label{sec:parallel-speedup}
We used a $2$-processor Intel Xeon $3.10$ Ghz E5-2687W server, each processor has $8$ cores, and each core can run two threads. The two processors share $20$MB of L3 cache and $256$GB of memory. We evaluate the speedup of our parallel algorithm (\ie, how much faster our proposed algorithm is when we increase the number of cores), and we used the OpenMP library for multi-core parallelization. In the following plots, we show the speedups versus the number of processing units (cores). All speedups are computed relative to the runtime of Algorithm~\ref{alg:parallel-graphlet_all} with one processor. To avoid possible variance, all experiments are repeated 5 times and averaged. Figures~\ref{fig:parallel-speedup-socfb}--\ref{fig:parallel-speedup-tech-web} show the speedup plots for a variety of graphs. We discuss a few observations from the plots presented here.

The first and most important observation that we make is that we obtain significant speedups from the parallel implementation of Algorithm~\ref{alg:parallel-graphlet_all}. Figures~\ref{fig:parallel-speedup-socfb}--\ref{fig:parallel-speedup-tech-web} show strong scaling results for a variety of graphs from social, web, and technological domains. Algorithm~\ref{alg:parallel-graphlet_all}
scales to $16$ cores and yields a speedup of $10$--$15$ folds. For example, as shown in Figure~\ref{fig:parallel-speedup-socfb}, we achieve almost linear scaling for the socfb-Penn94 graph (15-fold speedup for 16 cores). 

The second observation links the performance of Algorithm~\ref{alg:parallel-graphlet_all} to the characteristics of the graphs.  We observe the most significant speedups for social and Facebook networks (see Figure~\ref{fig:parallel-speedup-socfb}). We obtain near linear speedup as we increase the number of cores. Social networks are computationally intensive relative to the other graphs. This is due to their clustering characteristics and the existence of a large number of small communities (\ie, triangles, cliques, and cycles) in social networks.

\begin{figure}
\centering
\begin{minipage}[c]{0.49\textwidth}
\centering
\includegraphics[width=1.0\linewidth]{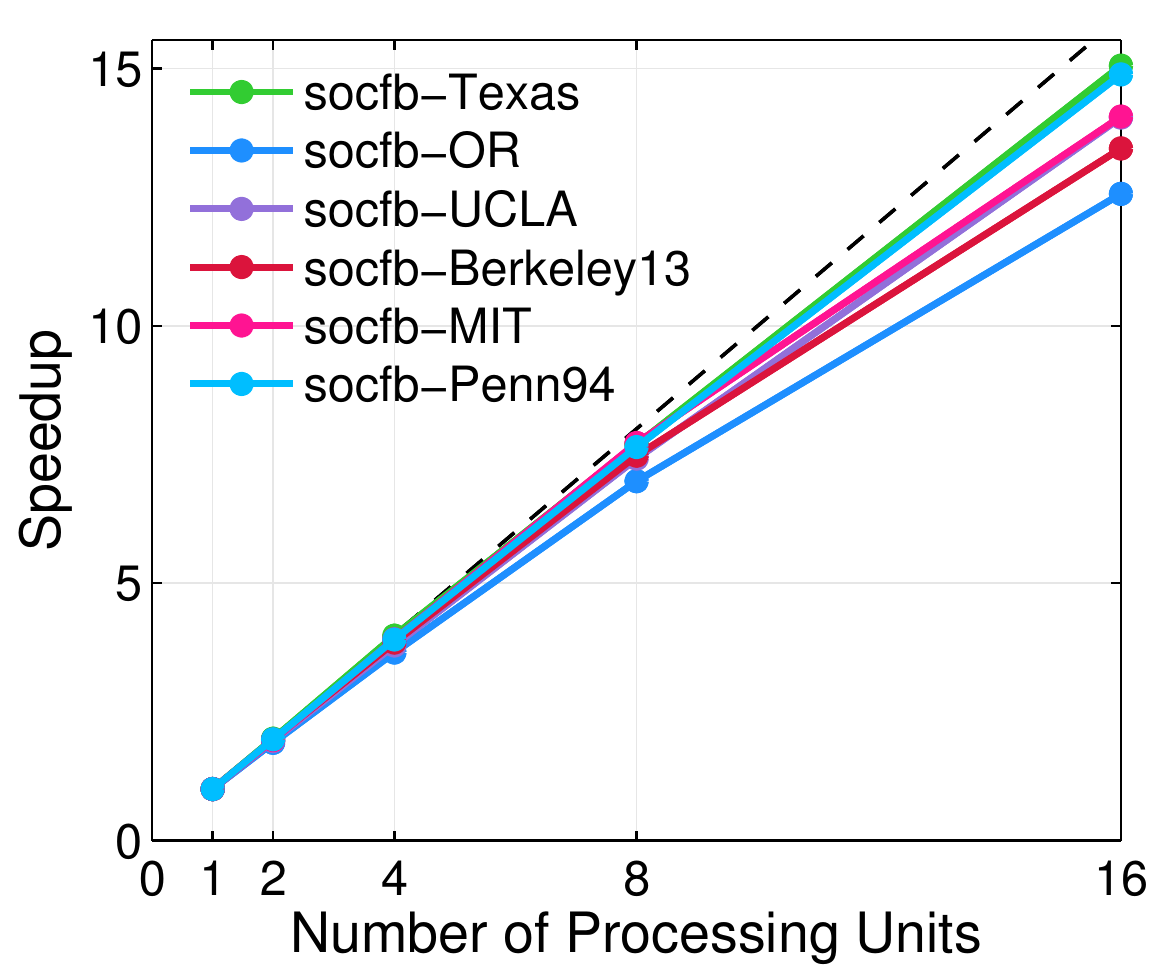}
\end{minipage}
\begin{minipage}[c]{0.49\textwidth}
\centering
\includegraphics[width=1.0\linewidth]{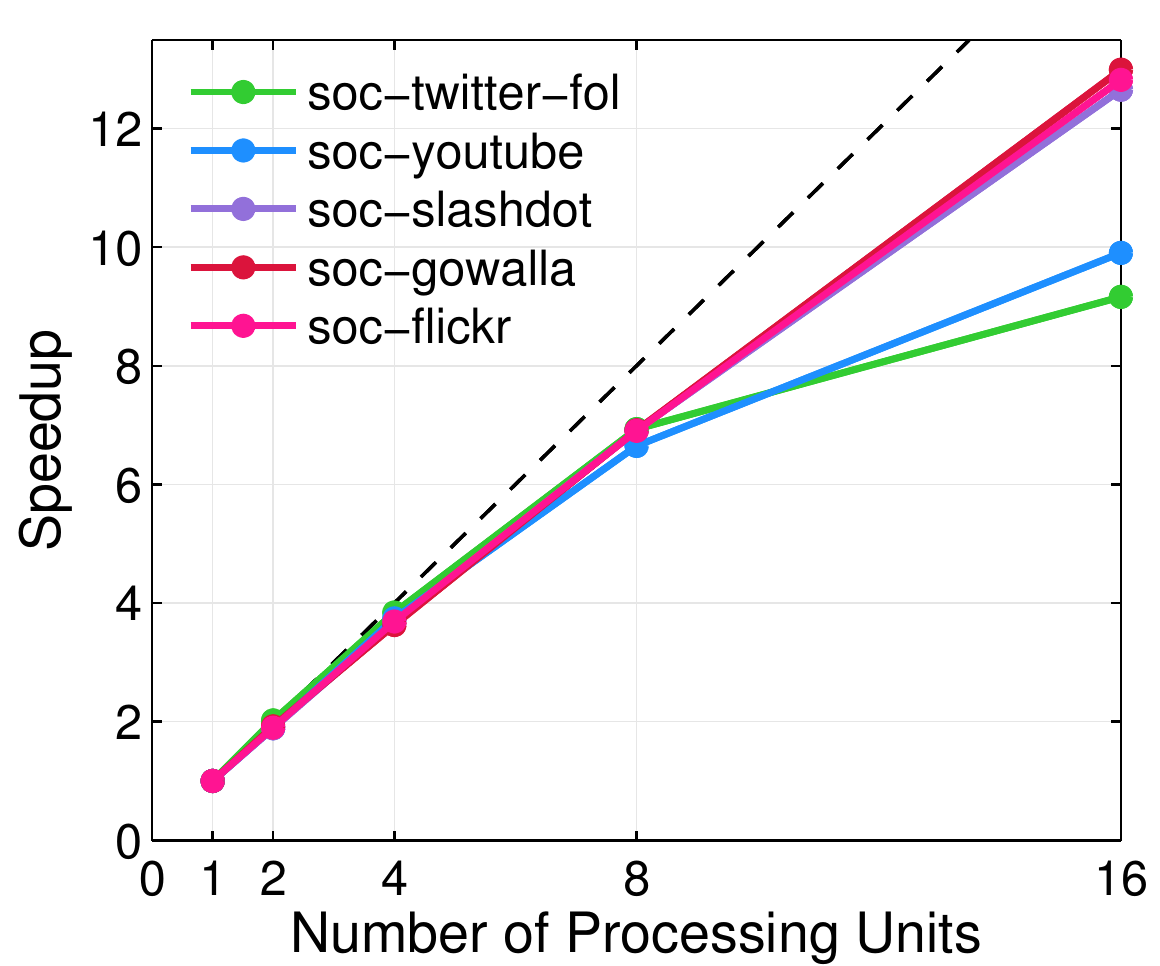}
\end{minipage}
\vspace{2.mm}
\caption{Strong scaling results for Facebook and social networks.
}
\label{fig:parallel-speedup-socfb}
\end{figure}

\begin{figure}
\centering
\begin{minipage}[c]{0.49\textwidth}
\centering
\includegraphics[width=1.0\linewidth]{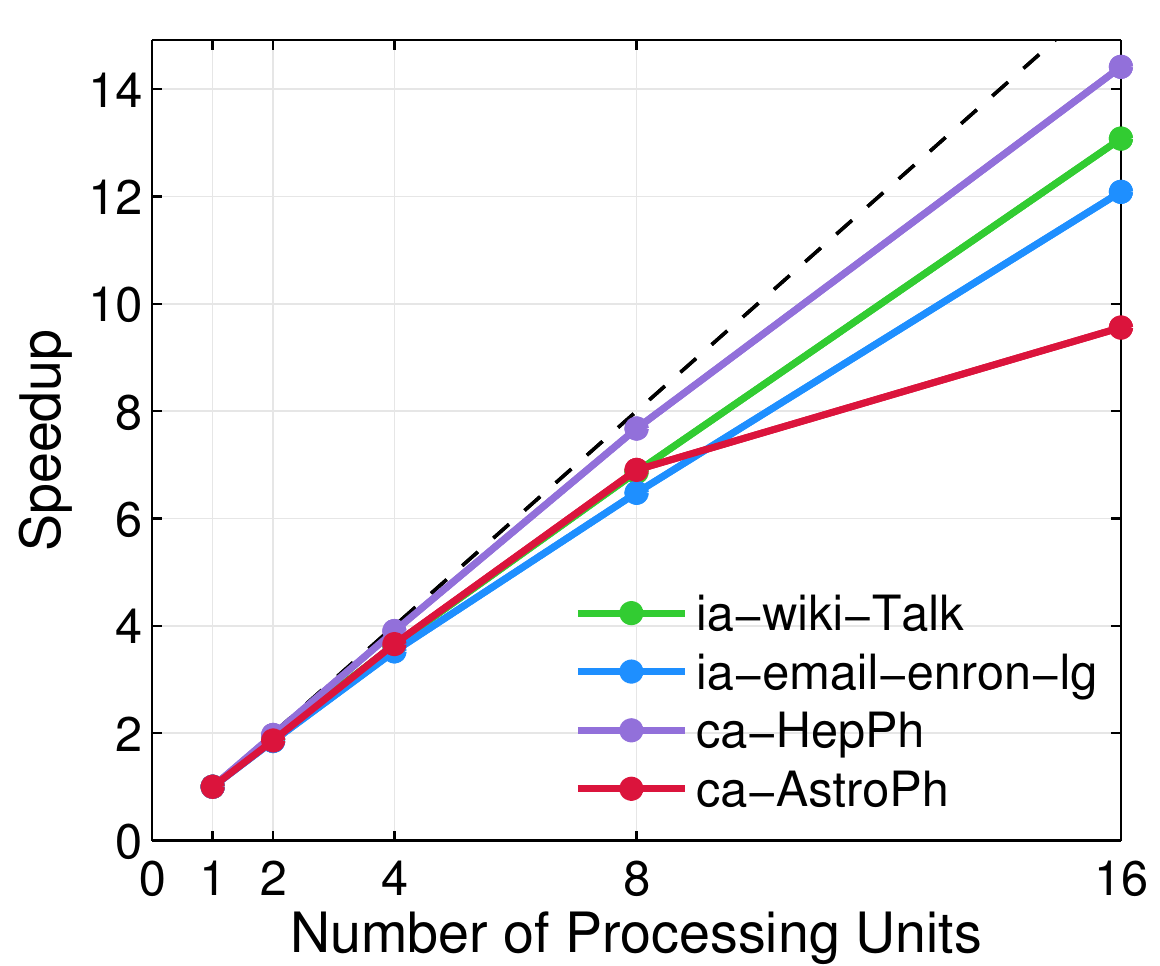}
\end{minipage}
\begin{minipage}[c]{0.49\textwidth}
\centering
\includegraphics[width=1.0\linewidth]{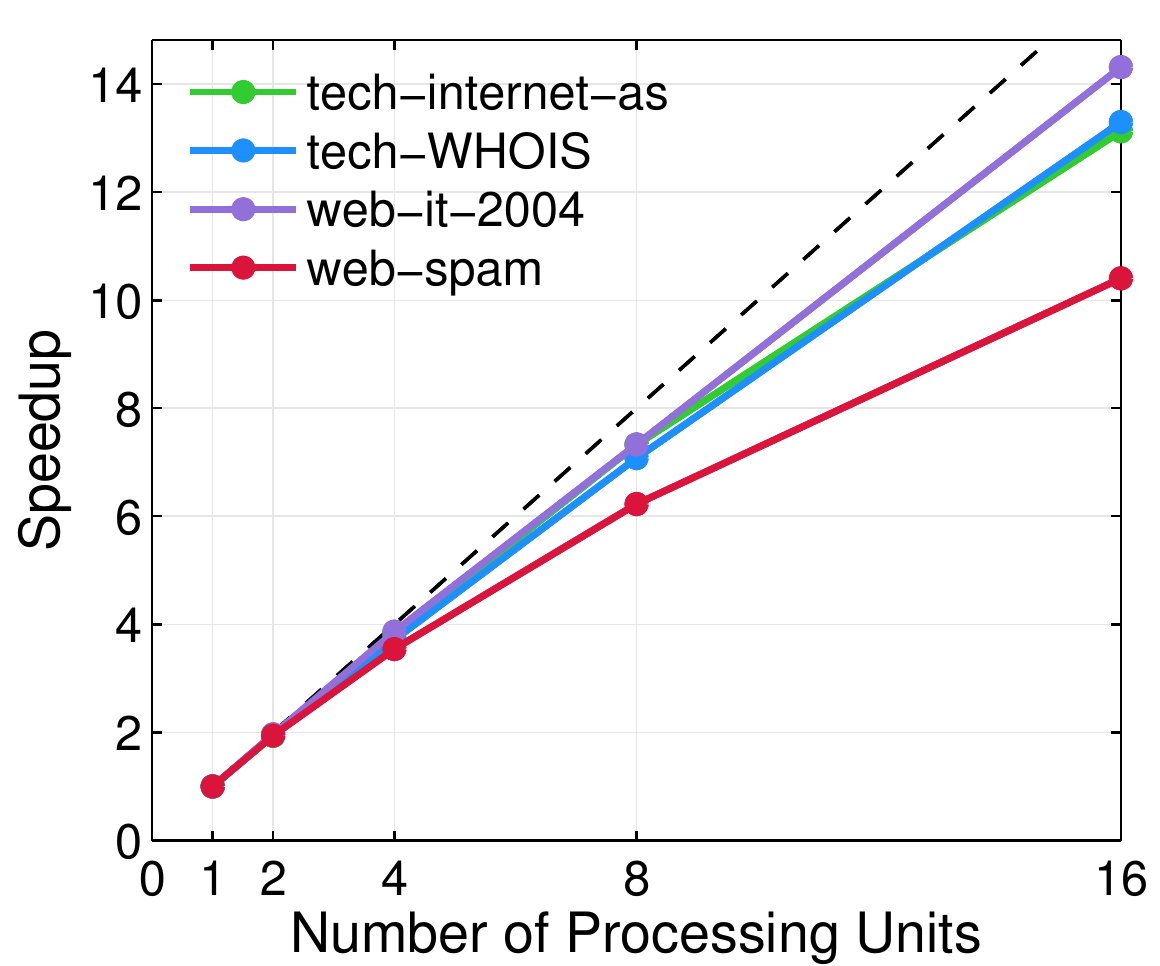}
\end{minipage}
\vspace{2.mm}
\caption{Strong scaling results for interaction, collaboration, technological, and web networks.
}
\label{fig:parallel-speedup-tech-web}
\end{figure}

The third observation we make is related to the optimal number of problems to dynamically assign to each processing unit when more work is requested (\ie, batch size $b$). That is the optimal performance that would be achieved when $b$ jobs are assigned in batch. Overall, we observed small performance fluctuations and found the optimal value of $b$ when we changed between 1 and 256 edges respectively. Interestingly, this observation is largely true only for sparse graphs, whereas graphs that are relatively dense (\eg, DIMACs graphs) work better when $b$ is small (\eg, even as small as $b=1$). This is likely due to the properties of these graphs and the auto-optimizer that we built into the library which automatically adapts the implementation of the algorithms to use additional data structures and achieve better performance for those relatively dense graphs at the cost of using additional space. Thus, our auto-optimizer appropriately balances the time and space trade-offs.

Note that the results for the job size experiments use degree for ordering the neighbors of each node in the succinct graph representation as well as for ordering the edge jobs to solve. In both cases, the ordering is from largest to smallest.

\section{Applications}
\label{sec:applications}
We also show some applications that could benefit from our fast graphlet counting algorithm (Algorithm~\ref{alg:parallel-graphlet_all}), which facilitates exploring and understanding networks and their structure. Graphlets provide an intuitive and meaningful characterization of a network at the global macro-level as well as the local micro-level, thus, they are useful for numerous applications. At the macro-level, graphlets are useful for finding similar networks (graph similarity queries), or finding networks that disagree most with that set (graph anomalies), or exploring a time-series of networks, among numerous other possibilities. Alternatively, graphlets are also extremely useful for characterizing networks and their behavior at the local node/edge-level as known as the micro-level. For instance, given an edge $(u,v) \in E$, find the top-k most similar edges (with applications in security, role discovery, entity-resolution, link prediction, and other related matching/similarity applications). Also, graphlets could be used for ranking nodes/edges to find unique patterns and anomalies such as large stars, cliques, etc.

\begin{figure}
\centering
\includegraphics[width=3.5in]{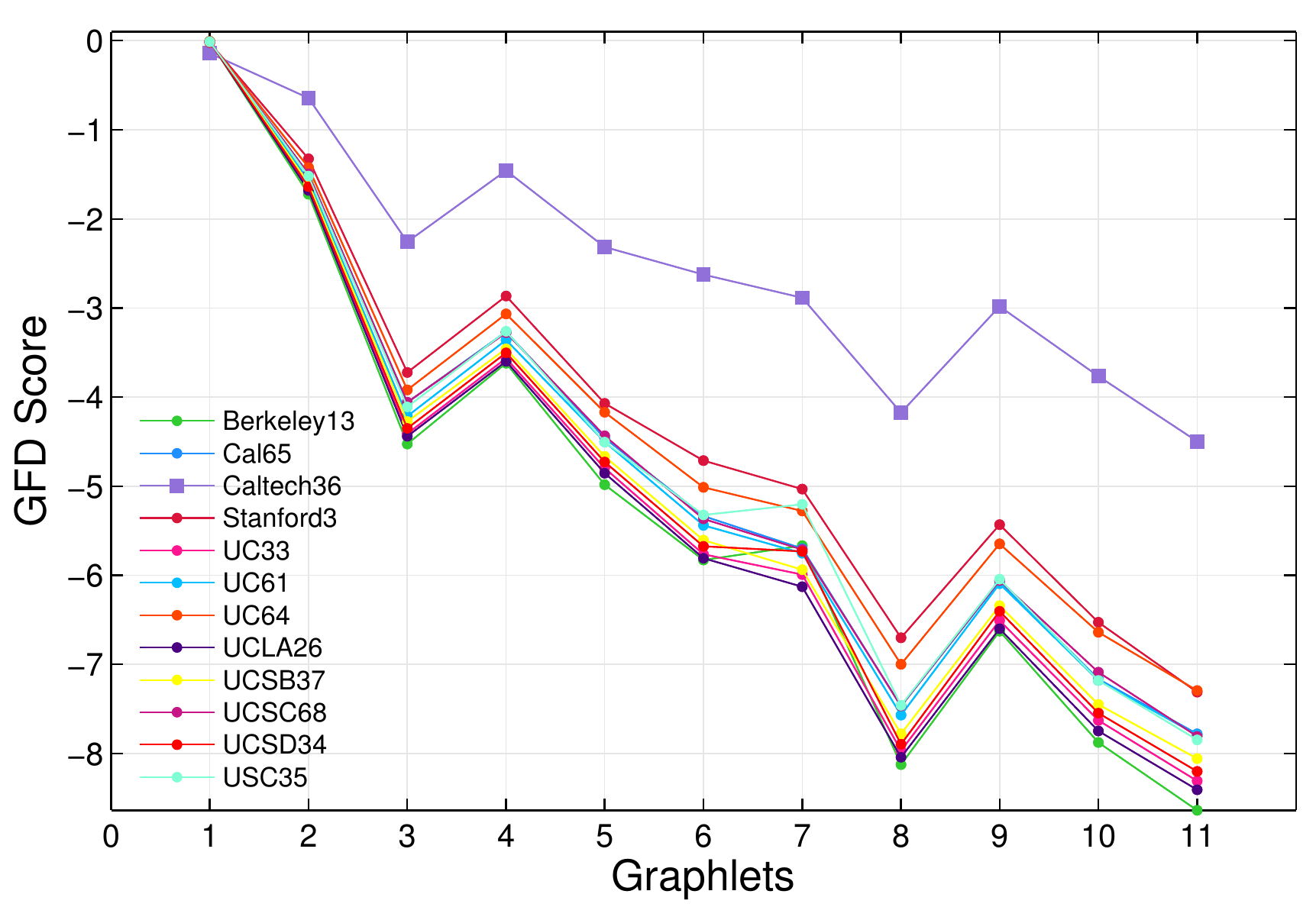}
\caption{Facebook social networks of California Universities.
Using the space of graphlets of size $k=4$, Caltech is noticeably different than others, which is consistent with the findings in~\protect\cite{traud2012social}.
}
\label{fig:fb-outlier}
\end{figure}

\subsection{Large-Scale Graph Comparison \& Classification}
Graphlets are also useful for large-scale comparison and classification of graphs. In this case, we relax the notion of equivalence and isomorphism to some form of \emph{structural similarity} between graphs, such that the graph similarity is measured using feature-based graphlet counts. In this section, we show how graphlets could be useful for network analysis, anomaly detection, and graph classification.

First, we study the full data set of Facebook100, which contains $100$ Facebook networks that represent a variety of US schools~\cite{traud2012social}. We plot the GFD (\ie, graphlet frequency distribution) score pictorially in Figure~\ref{fig:fb-outlier} for all California schools. The GFD score is simply the normalized frequencies of graphlets of size $k$~\cite{prvzulj2004modeling}. In our case, we use $k=4$. The figure shows Caltech noticeably different than others, consistent with the results in~\cite{traud2012social} which shows how Caltech is well-known to be organized almost exclusively according to its undergraduate "Housing" residence system, in contrast to other schools that follow the predominant "dormitory" residence system. The residence system seems to impact the organization of the social community structures at Caltech.

Second, we use counts of graphlets of size $k=\{2,3,4\}$-nodes as features to represent each graph in a large collection of graphs. Using the graphlet feature representation, we learn a model to predict the unknown label of each unlabeled graph (\eg, the label could be the function of protein graphs). We test our approach on protein graphs (D\&\,D collection of $1178$ protein graphs) and chemical compound graphs (MUTAG collection of $188$ chemical compound graphs)~\cite{vishwanathan2010graph}. We extract the graphlet features using Algorithm~\ref{alg:parallel-graphlet_all}. Then, we learn a model using SVM (RBF kernel), and we use $10$-fold validation for evaluation. Table~\ref{tab:classif-res} shows the accuracy of this approach is $76\%$ for protein function prediction, and $86\%$ for mutagenic effect prediction. Note that by using all graphlet-based features up to size $4$ nodes, we were able to obtain better accuracy than previous work (which achieved maximum $75\%$ and $83\%$ accuracy for D\&\,D and MUTAG respectively~\cite{shervashidze2009efficient}). Moreover, Algorithm~\ref{alg:parallel-graphlet_all} extracts all the features (graphlet counts) in almost one second. This yields a significant improvement over the graphlet feature extraction approach that was proposed in~\cite{shervashidze2009efficient}, which takes $2.45$ hours to extract graphlet features from the D\&\,D collection.

{
\begin{table}
\begin{center}
\caption{Accuracy \& Standard Error for Classification of Large Collection of Biological \& Chemical Graphs. We used counts of all graphlets of size $k=\{2,3,4\}$ as features.}
\medskip
\label{tab:classif-res}
\scalebox{0.8}{
\begin{tabular}{rrrrrr}
\toprule
graph & Type & No. Graphs & Accuracy$(\%)$ & Total Time(sec) & Avg Time per G (sec) \\
\midrule
D\&D & Protein & 1178 & 76.13 $\pm$ 0.03 & 1.05 & 8.95x$10^{-4}$ \\
MUTAG & Chemicals & 188 & 86.4 $\pm$ 0.21 & 0.14 & 7.47x$10^{-4}$ \\ 
\bottomrule
\end{tabular}}
\end{center}
\vspace{-3.mm}
\end{table}
}

Third, we compute graphlet counts on a $2$ billion edge social network called Friendster. Friendster is an on-line gaming network. Before re-launching as a game website in $2011$, Friendster was an online social network where users can form friendship links with each others. This data is provided by The Web Archive Project before the death of the social network. In these experiments, we use the induced subgraph of the nodes that either belong to at least one community or are connected to other nodes that belong to at least one community. Table~\ref{table:graphlets-fb100} shows a significantly large number of $4$-path (chains of $4$ connected nodes) and $3$-stars compared to the number of $4$-cliques and triangles. Although the induced subgraph that we used from Friendster is clearly biased toward communities,  the patterns that represent communities, such as cliques and triangles, are less likely in the induced graph. For example, the frequency of $4$-path patterns is $0.58$, while the frequency of $4$-cliques is $0.000014$. These results indicate that something wrong happened to the social network. Previous work on the autopsy of Friendster showed that there was a collapse in the community structure of Friendster, a cascade in user departure due to bad decisions in the design and interface changes. In a similar fashion, the low frequency of community-related graphlets (\eg, cliques) in Friendster also indicates the collapse of the social network.

\begin{figure}
\centering
\includegraphics[width=3.5in]{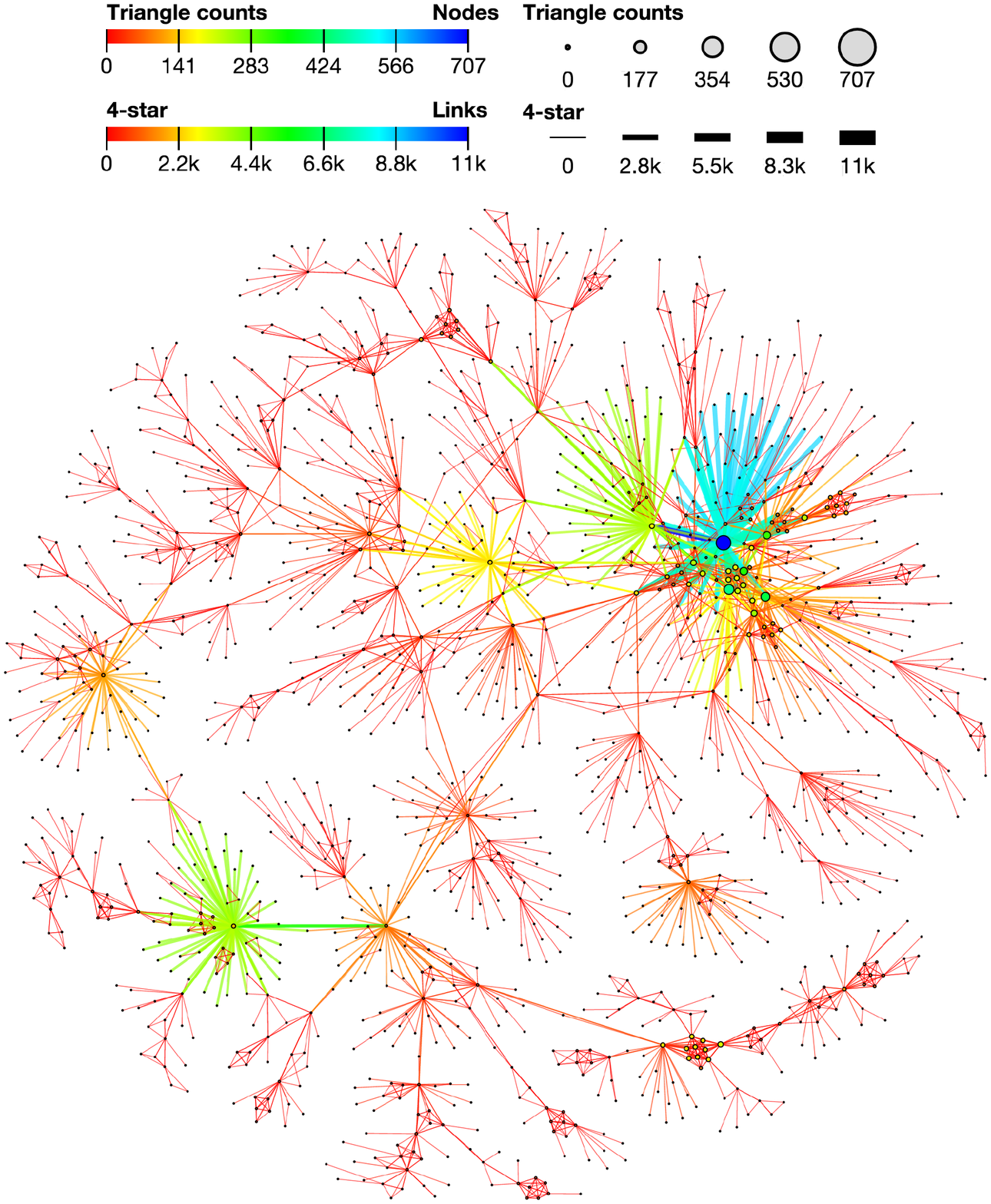}
\caption{\textbf{Visualization of the human diseasome network}: A network of disorders and disease genes linked by known disorder-gene associations~\protect\cite{goh2007human}. Edges are weighted/colored by their number of incident star graphlets of size $4$ nodes, nodes are weighted/colored by their triangle counts. 
The large star on the right denoted by light blue color corresponds to colon cancer; the large star on the lower left denoted by lime green color corresponds to deafness; and the large star on the right denoted by lime green color corresponds to leukemia. Notably this figure highlights the few phenotypes (such as colon cancer, leukemia, and deafness) correspond to hubs (large stars) that are connected to a large number of distinct disorders, which is consistent with~\protect\cite{goh2007human}.}
\label{fig:finding-stars}
\vspace{-3.mm}
\end{figure}

\subsection{Finding Large Stars, Cliques, and other Patterns Fast}
How can we quickly and efficiently find large cliques, stars, and other unique patterns?
Further, how can we identify the top-k largest cliques, stars, etc? Note that many of these problems are NP-hard, \eg, finding the clique of maximum size is a well-known NP-hard problem~\cite{HGTbook}. To answer these and other related queries, we leverage the proposed parallel graphlet counting method in Algorithm~\ref{alg:parallel-graphlet_all}. 
The idea is clearly shown in Figure~\ref{fig:finding-stars}. Figure~\ref{fig:finding-stars} provides a visualization of the human diseasome network~\cite{goh2007human}, where we used Alg.~\ref{alg:parallel-graphlet_all} to rank (weight) all the edges in the network by the number of star patterns of size $4$ nodes. The intuition behind the method is that if an edge (or node) has a (relatively) large number of stars of $4$ nodes (cliques, or another graphlet of interest), then it is also likely to be part of a star of a large size. Recall that removing a node from a $k$-star or $k$-clique forms a star or clique of size $k-1$~\cite{HGTbook}.
Accordingly, edges with large weights are likely to be members of large stars. Thus, as shown in Figure~\ref{fig:finding-stars}, a visualization based on our fast graphlet counting method can help to quickly highlight such large stars by using the counts (of stars of size $4$ nodes) as edge weights or colors. Notably, Figure~\ref{fig:finding-stars} highlights the few phenotypes (such as colon cancer, leukemia, and deafness) correspond to hubs (large stars) that are connected to a large number of distinct disorders, which is consistent with the findings in ~\cite{goh2007human}.

Note that the same approach is also applicable for finding cliques and other interesting patterns, since edges with a high number of $4$-cliques are likely to be members of the largest clique in the network. Figure~\ref{fig:interactive-graphlet-terrorist-social-network} shows how we can find large cliques in the terroristRel data~\cite{zhao:sna06}. 

\begin{figure}
\centering
\includegraphics[width=1.0\linewidth]{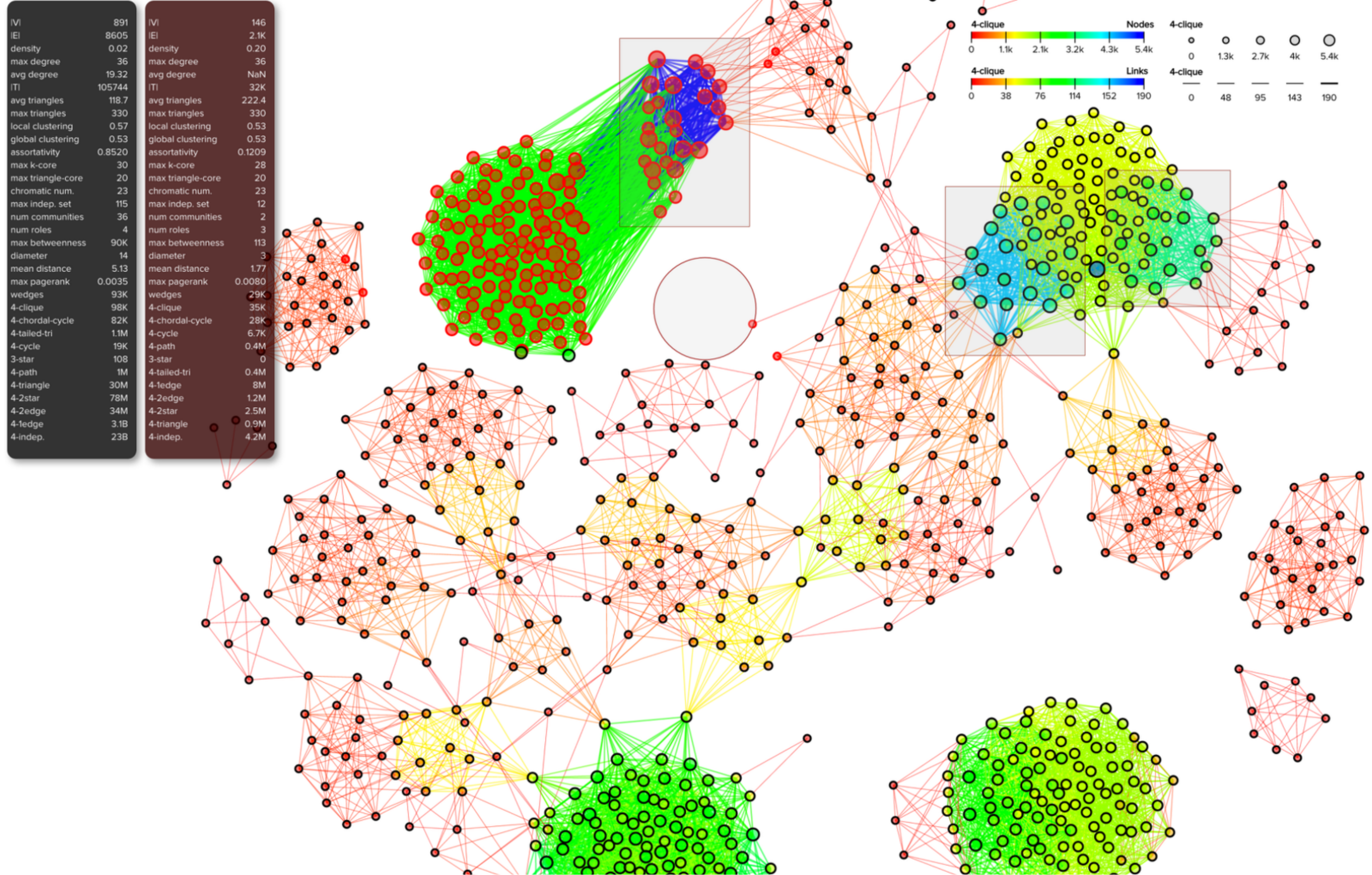}
\caption{\textbf{Visualization of the terroristRel network}: A network of terrorists and their relationships. Terrorists are linked to each other if they contact each other, use the same facility, are members of the same family, or belong to the same terrorist organization. Edges and nodes are weighted/colored by their number of incident cliques of size $4$ nodes. Notably this figure highlights how the structure of terrorist networks is decomposed of various clique patterns (terrorist organization) and how these cliques are interconnected. The figure highlights the largest clique on the top left denoted by dark blue.
}
\label{fig:interactive-graphlet-terrorist-social-network}
\end{figure}

\subsection{Real-time Visual Graphlet Mining} \label{sec:visual-graphlet-analytics}
Visual analytics is the science of analytical reasoning facilitated by interactive visual interfaces~\cite{thomas2005illuminating}.
This work develops an interactive visual graph analytics platform based on the proposed fast graphlet decomposition algorithm. 
In particular, we integrate interactive visualization with our state-of-the-art parallel graphlet decomposition algorithm in order to support discovery, analysis, and exploration of such data in real-time.

We utilize this multi-level graphlet analysis engine that uses graphlets as a basis for exploring, analyzing, and understanding complex networks interactively in real-time. And, we  highlight other key aspects including filtering, querying, ranking, manipulating, and a variety of multi-level network analysis and statistical techniques based on graphlets.

Notably, our proposed algorithm is shown to be fast and efficient for \emph{real-time} interactive exploration and mining of graphlets.
We expect this tool to be extremely useful to biologists and others interested in understanding biological (protein, brain networks, etc.) as well as chemical networks.

There are a number of important and unique challenges in designing methods for interactive exploration and mining of graphlets in real-time.
In particular, the real-time requirement of such a system requires fast parallel methods to achieve real-time interactive rates (e.g., with response times within  microseconds or less).
In particular, we derived dynamic update methods that are localized, that is, the update methods leverage the local structure of the graph for efficiently updating the counts when nodes/edges are selected, inserted, removed, etc.
Thus, given a single node or edge, the method updates the graphlet counts for that edge (as opposed to recomputing the full graphlet decomposition).

Figure~\ref{fig:interactive-graphlet-neural-brain-network} uses the interactive graphlet mining tool for real-time exploration of the brain neural network from C. Elegans~\cite{watts1998collective}. 
Additionally, the tool is also useful for exploring many other types of networks, e.g., a terrorist relationship network is shown in Figure~\ref{fig:interactive-graphlet-subgraph-analysis-terrorRel} whereas Figure~\ref{fig:interactive-graphlet-subgraph-analysis-comm-structure} uses graphlets as a basis for understanding and characterizing the communities and their structure. 
As an aside, the graph in Figure~\ref{fig:interactive-graphlet-subgraph-analysis-comm-structure} is generated using the block Chung-Lu graph model.
Thus, it is straightforward to see how graphlets can be used to characterize synthetic graph generators and for evaluating their utility (e.g., if the synthetic graph preserves the distribution of graphlets observed in a real-world network.).

\begin{figure}
\centering
\includegraphics[width=0.8\linewidth]{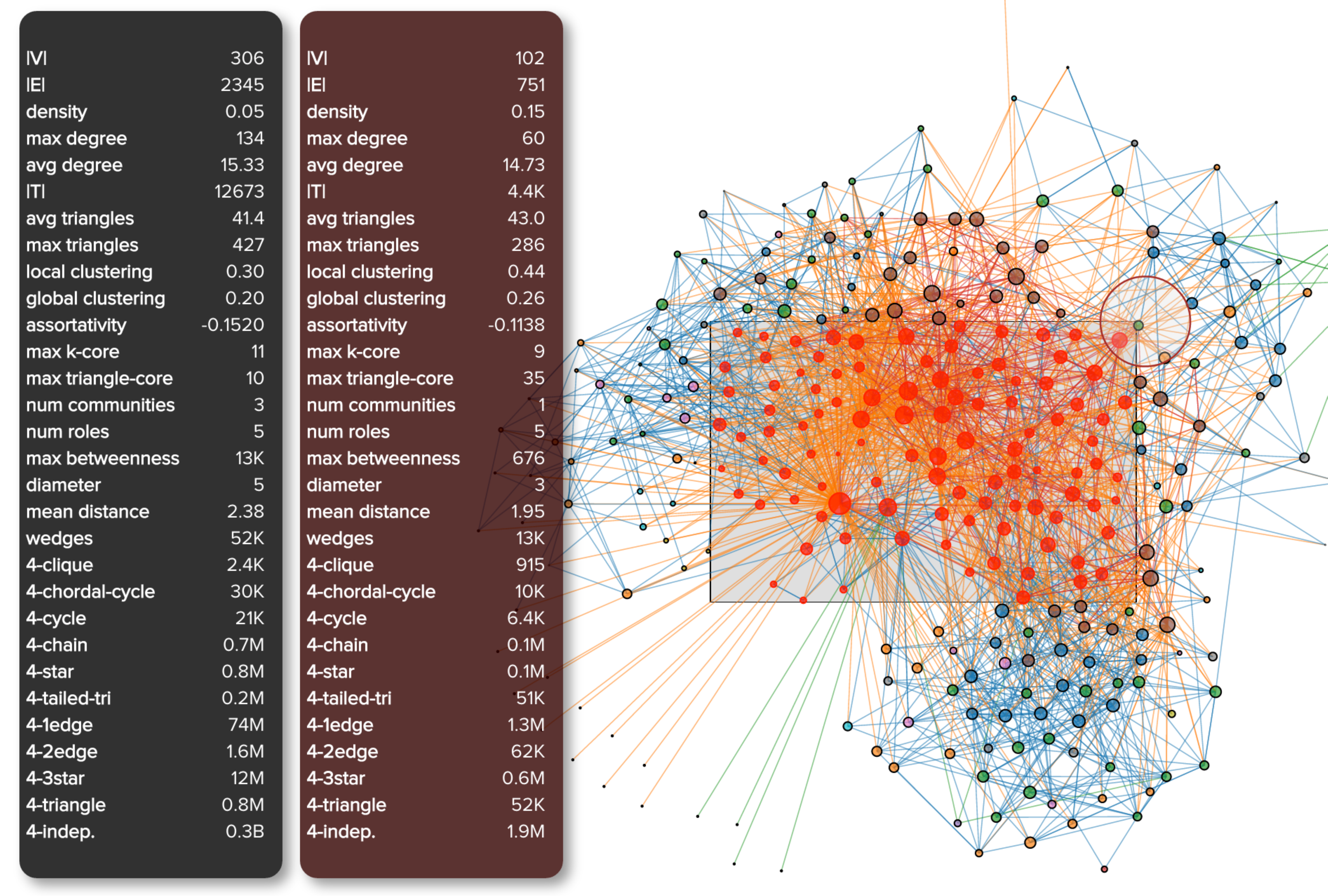}
\caption{Exploration of the brain neural network of C. Elegans~\protect\cite{watts1998collective} using our \textit{interactive graphlet visual analytics tool}.
Nodes are colored by their k-core number and weighted by betweenness, whereas the links are colored by eccentricity.}
\label{fig:interactive-graphlet-neural-brain-network}
\end{figure}

The \emph{visual graphlet analytics tool} is designed for rapid interactive visual exploration and graph mining (Figure~\ref{fig:interactive-graphlet-neural-brain-network}-\ref{fig:interactive-graphlet-subgraph-analysis-comm-structure}).
Graphlets are computed on-the-fly upon a simple drag-and-drop of a graph file into the web browser.
Additionally, the graphlet counts are updated efficiently after each selection, insertion, deletion, or change to the graph data.
Furthermore, it is designed to be consistent with the way humans learn via immediate-feedback upon every user interaction (e.g., change of a slider for filtering)~\cite{ahlberg1992dynamic,thomas2005illuminating}.
Users have rapid, incremental, and reversible control over all graph queries with immediate and continuous visual feedback.

\begin{figure}
\centering
\includegraphics[width=0.96\linewidth]{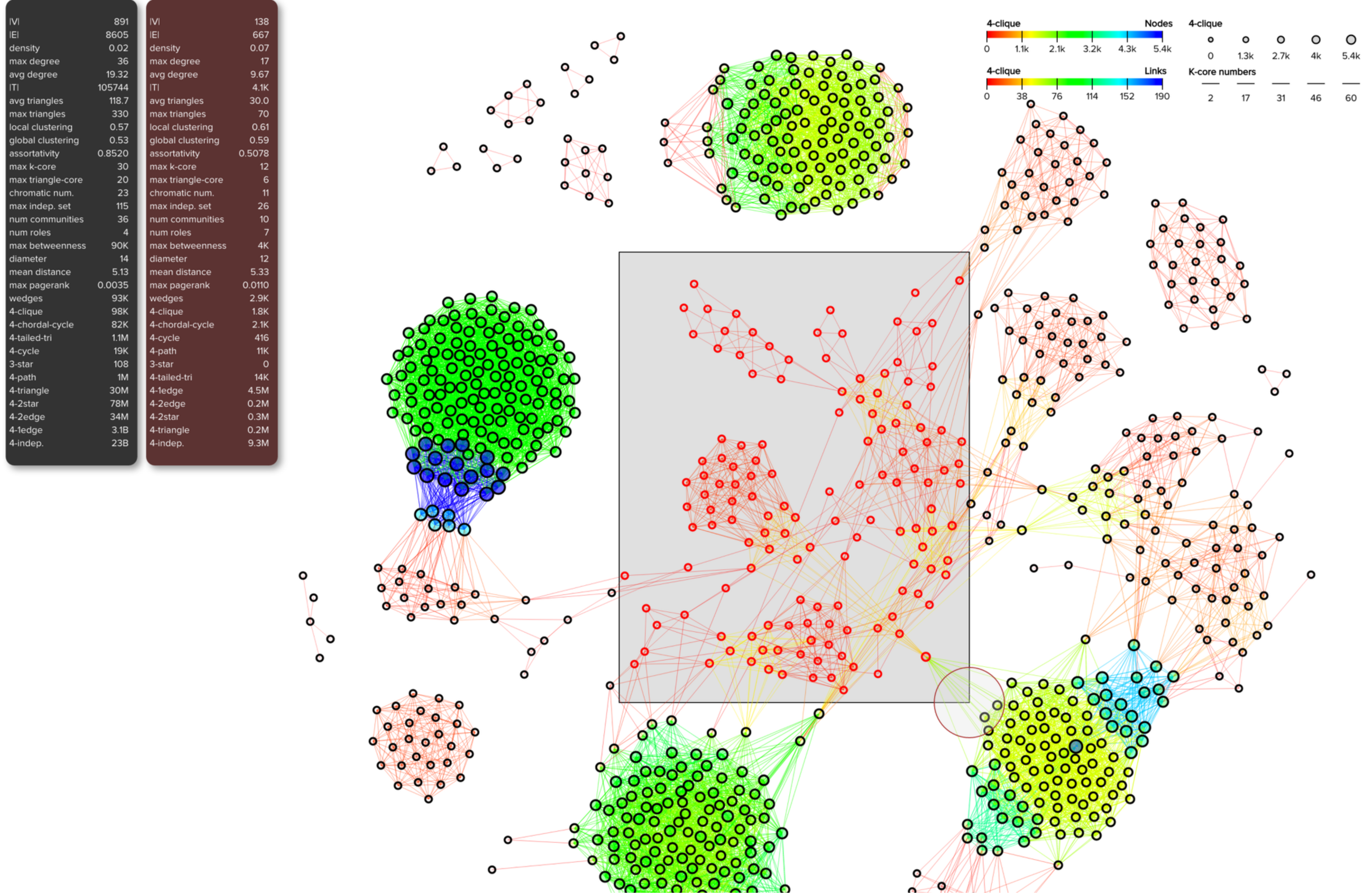}
\caption{
Illustration of the graphlet methods for real-time interactive graphlet analysis.
This demonstrates the efficiency and effectiveness of the proposed methods for interactive real-time graphlet computations.
In the screenshot above, the user selects a subgraph to interactively analyze via direct manipulation of the visualization using the mouse.
That is, the user adjusts the rectangular region above to highlight the subgraphs to analyze.
The graphlet statistics are updated each time a node/edge is added or removed from the rectangular region used to select the subgraph to explore via graphlets. Thus, the user can see how the graphlet statistics change as nodes and edges are added (or removed) from the user-specified rectangular region (which in turn indicates the nodes and edges to include in the analysis).
Note that we leverage localized graphlet update methods to achieve the performance required for real-time interactive graphlet mining and sense-making.
}
\label{fig:interactive-graphlet-subgraph-analysis-terrorRel}
\end{figure}

\begin{figure}
\centering
\includegraphics[width=0.96\linewidth]{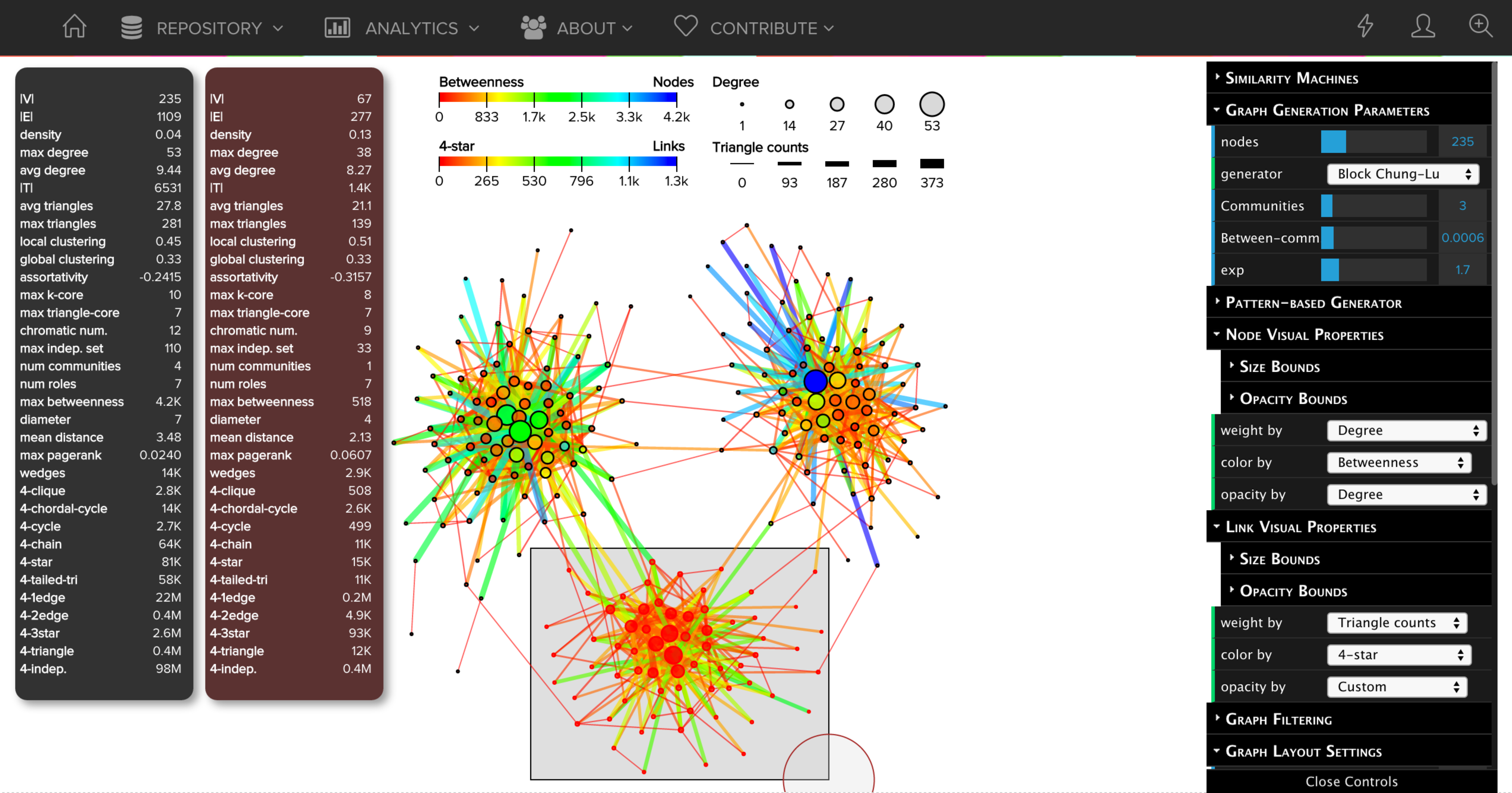}
\caption{
Interactive graphlet exploration of community structure via direct manipulation and selection of the visual representation
}
\label{fig:interactive-graphlet-subgraph-analysis-comm-structure}
\end{figure}

\section{Conclusion \& Future Work}
\label{sec:conc-future}
In this paper, we proposed a fast, efficient, and parallel algorithm for counting graphlets of size $k=\{3,4\}$-nodes that take only a fraction of the time to compute when compared with the current methods used. The proposed graphlet counting algorithm leverages a number of proven combinatorial arguments for different graphlets. For each edge, we count a few graphlets, and with these counts along with the combinatorial arguments, we obtain the exact counts of others in constant time. We systematically investigate the scalability of our algorithm on a large collection of $300$+ networks from a variety of domains.
In future work, we aim to extend our proposed algorithm to higher-order graphlets.

\bibliographystyle{agsm}
\bibliography{graphlet_journal}

\correspond{Nesreen K. Ahmed, Intel Labs, Intel Corporation, Santa Clara, CA 95054, USA. Email: nesreen.k.ahmed@intel.com}
\label{lastpage}
\end{document}